\documentclass[preprint, 12pt]{elsarticle}

\usepackage{graphicx}

\usepackage{amsmath}

\newtheorem{theorem}{Theorem}
\newtheorem{corollary}{Corollary}
\newtheorem{proposition}{Proposition}
\newdefinition{definition}{Definition}
\newtheorem{lemma}{Lemma}

\newdefinition{problem}{Problem}

\newproof{proof}{Proof}

\usepackage{thmtools}

\PassOptionsToPackage{hyphens}{url}
\usepackage[hidelinks,bookmarks=false]{hyperref}

\usepackage{amsfonts}
\usepackage{mathtools}
\usepackage{verbatim}
\usepackage{paralist}
\usepackage{xifthen}
\usepackage{mathdots}
\usepackage{nicefrac}

\usepackage[shortlabels]{enumitem}

\usepackage[ruled, vlined, linesnumbered, commentsnumbered]{algorithm2e}
\SetAlgorithmName{Algorithm}{Algorithm}{List of Algorithms}

\usepackage{tikz}
\usepackage{tikz-qtree}
\usetikzlibrary{shapes,shapes.geometric,arrows,fit,matrix,positioning,decorations.pathreplacing,calc,fpu,decorations.pathmorphing,decorations.markings, decorations}

\tikzset
{
    treenode/.style = {circle, draw=black, align=center, minimum size=0.5cm},
    subtree/.style  = {isosceles triangle, draw=black, align=center, minimum height=0.5cm, minimum width=1cm, shape border rotate=90, anchor=north},
    process/.style={rectangle, minimum width=2cm, minimum height=1cm, align=center, text width=2cm, draw},
    connector/.style={circle, minimum size=1cm, align=center, text width=0.5cm, draw},
    arrow/.style={thick, ->, >=stealth}
}
\usepackage{circuitikz}

\usepackage{scalerel,stackengine}
\stackMath
\newcommand\reallywidehat[1]{%
\savestack{\tmpbox}{\stretchto{%
  \scaleto{%
    \scalerel*[\widthof{\ensuremath{#1}}]{\kern-.6pt\bigwedge\kern-.6pt}%
    {\rule[-\textheight/2]{1ex}{\textheight}}
  }{\textheight}%
}{0.5ex}}%
\stackon[1pt]{#1}{\tmpbox}%
}

\newcommand{\tworowcell}[2][c]{%
	\begin{tabular}[#1]{@{}c@{}}#2\end{tabular}}

\usepackage{subcaption}
\usepackage{siunitx}
\usepackage{calc}

\allowdisplaybreaks[1]
	
\newcommand{\ints}{\mathbb{Z}}

\newcommand{\repfact}{\rho}
\newcommand{\V}{V}                   
\newcommand{\E}{A}                   
\newcommand{\F}{E}                   
\newcommand{\C}{L}                   
\renewcommand{\P}{P}                 
\newcommand{\Q}{Q}                   
\newcommand{\MP}{\mathcal{P}}        
\newcommand{\G}{G}                   
\newcommand{\nChildren}{t}           
\newcommand{\minRepOn}[1]{r^{min}_{#1}}    
\newcommand{\repOn}[1]{r_{#1}}       
\newcommand{\filled}{F}              
\newcommand{\unfilled}{U}            
\newcommand{\R}{R}                   
\renewcommand{\a}{\vec{a}}           
\renewcommand{\b}{\vec{b}}           
\newcommand{\perm}{\pi}              
\newcommand{\n}{n}                   
\newcommand{\m}{m}                   
\newcommand{\lf}{\ell}               
\renewcommand{\path}[2]{\pi (#1,#2)} 
\newcommand{\mpRepFact}{{\hat{\repfact}}} 
\newcommand{\skewfactor}{\delta}     

\newcommand{\f}{f}                        
\newcommand{\vecf}{\vec{f}}
\newcommand{\fnum}[2]{\f(#1, #2)}         
\newcommand{\ff}[1]{\vecf(#1)}
\newcommand{\g}[1]{\vec{g}(#1)}           
\newcommand{\alp}{\vec{\alpha}}           
\newcommand{\alpf}[1]{\alp({#1})}         
\newcommand{\sig}[1]{\vec{\sigma}(#1)}         

\newcommand{\intvecs}{\ints^{\repfact+1}}                        
\newcommand{\logroup}{\groupdef{\G}{\gop}{\logeq}}               
\newcommand{\groupdef}[3]{\langle#1, #2, #3\rangle}              
\newcommand{\graph}[2]{(#1,#2)}                                  
\newcommand{\Csub}[1]{\C_{#1}}                                   
\newcommand{\smallRep}{\floorfrac{\R}{|\unfilled|}}              
\newcommand{\bigRep}{\ceilfrac{(\R + 1)}{|\unfilled|}}           
\newcommand{\nHvy}{\R \bmod |\unfilled|}                         
\newcommand{\mpRT}{mpRT}                                         

\newcommand{\child}[1]{c_{#1}}                    
\newcommand{\children}[2]{\child{#1},...,\child{#2}} 
\newcommand{\subtree}[2]{T_{#1}^{#2}}   

\newcommand{\lleq}{\leq_L}
\newcommand{\lgeq}{\geq_L}
\newcommand{\llt}{<_L}
\newcommand{\lgt}{>_L}
\newcommand{\defined}{\mathrel{\mathop:}=}
\newcommand{\logeq}{\succeq}        
\newcommand{\loleq}{\preceq}        
\newcommand{\gop}{+}                

\DeclarePairedDelimiter{\ceil}{\lceil}{\rceil}
\DeclarePairedDelimiter{\floor}{\lfloor}{\rfloor}
\newcommand{\ceilfrac}[2]{\ceil{\nicefrac{#1}{#2}}}
\newcommand{\floorfrac}[2]{\floor{\nicefrac{#1}{#2}}}

\newcommand{\signature}{signature\xspace}                              
\newcommand{\totalSize}{size\xspace}                                          

\newcommand{\dfill}{definitely-filled\xspace}
\newcommand{\punfill}{possibly-unfilled\xspace}
\newcommand{\Dfill}{Definitely-filled\xspace}

\newcommand{\subplacement}{sub-multi-placement\xspace}
\newcommand{\out}{out\xspace}
\newcommand{\up}{up\xspace}

\newcommand{\girth}{girth}
\newcommand{\skewterm}{skew}
\newcommand{\block}{block} 

\renewcommand{\vec}[1]{\boldsymbol{#1}}

\newcommand{\A}{A}                   

\DeclareMathOperator{\med}{med}
\newcommand{\ef}{\vec{s}}

\newcommand{\punt}[1]{}

\renewcommand*{\appendixname}{}

\journal{Theoretical Computer Science}
\begin{document}

\begin{frontmatter}
	
	
	
	\title{Algorithms for Optimal Replica Placement Under Correlated Failure in Hierarchical Failure Domains\tnoteref{prelim}}
	\tnotetext[prelim]{A preliminary version of this work appeared in the Proceedings of the 9th Annual International Conference on Combinatorial Optimization and Applications (COCOA), 2015~\cite{Mills2015}.}
	
	
	\author{K. Alex Mills}
	\ead{k.alex.mills@utdallas.edu}
	
	\author{R. Chandrasekaran}
	\ead{chandra@utdallas.edu}
	
	\author{Neeraj Mittal\fnref{NSF}}
	\ead{neerajm@udtallas.edu}
	\fntext[NSF]{This work was supported, in part, by the National Science Foundation (NSF) under grants numbered CNS-1115733 and CNS-1619197.}
	
	\address{The University of Texas at Dallas, 800 W. Campbell Rd., Richardson, TX 75080}
	
	\begin{abstract}
		In data centers, data replication is the primary method used to ensure availability of customer data. To avoid correlated failure, cloud storage infrastructure providers model hierarchical failure domains using a tree, and avoid placing a large number of data replicas within the same failure domain (i.e. on the same branch of the tree). Typical best practices ensure that replicas are distributed across failure domains, but relatively little is known concerning optimization algorithms for distributing data replicas. Using a hierarchical model, we answer how to distribute replicas across failure domains optimally. We formulate a novel optimization problem for replica placement in data centers. As part of our problem, we formalize and explain a new criterion for optimizing a replica placement. Our overall goal is to choose placements in which \emph{correlated} failures disable as few replicas as possible. We provide two optimization algorithms for dependency models represented by trees. We first present an $O(\n + \repfact \log \repfact)$ time dynamic programming algorithm for placing $\repfact$ replicas of a single file on the leaves (representing servers) of a tree with $\n$ vertices. We next consider the problem of placing replicas of $m$ blocks of data, where each block may have different replication factors. For this problem, we give an exact algorithm which runs in polynomial time when the \emph{skew}, the difference in the number of replicas between the largest and smallest blocks of data, is constant.
	\end{abstract}
	

\end{frontmatter}

\section{Introduction}
\label{intro}
With the surge towards the cloud, our websites, services and data are increasingly being hosted by third-party data centers. These data centers are often contractually obligated to ensure that data is rarely, if ever unavailable. One cause of unavailability is co-occurring component failures, which can result in outages that can affect millions of websites \cite{Verge2013}, and cost millions of dollars in profits \cite{Pletz2013}. An extensive one-year study of availability in Google's cloud storage infrastructure showed that such failures are relatively harmful. Their study emphasizes that ``correlation among node failure dwarfs all other contributions to unavailability in our production environment" \cite{Ford2010}.

One of the main reasons for correlation among data center failure events is dependencies among system components. Much effort has been made in the literature to produce quality statistical models of this correlation \cite{Bakkaloglu2002,Ford2010,Nath2006,Weatherspoon2002} But in using such models researchers do not make use of the fact that many of these dependencies can be explicitly modeled, since they are known to the system designers. In contrast, we propose to make use of the failure domain models which are already used in commercially available cloud storage products \cite{Parallels,VMWare} to avoid correlated failure.

To achieve high availability, data centers typically store multiple replicas of data to tolerate the potential failure of system components. This gives rise to the \emph{replica placement problem}, an optimization problem which, broadly speaking, involves determining which servers in the system should store copies of a file so as to maximize a given objective (e.g. reliability, communication cost, response time, or access time). While our focus is on replica placements, our approach could also be used to place replicas of other entities that require high-availability, such as virtual machines or mission-critical tasks.

\begin{figure}[tb]
	\centering 
	\subcaptionbox{Scenario I\label{fig:scenarioI}}[0.48\textwidth]{
		\begin{tikzpicture}[scale=0.5, transform shape]
		\tikzstyle{every node}=[minimum size=0.9cm, circle, align=center, font=\LARGE,  draw=black]
		
		\pgfmathsetmacro{\xoffset}{1.5}
		\pgfmathsetmacro{\yoffset}{1.5}
		
		\node (agg1) at (0, 0 ) {};
		\node [on grid, right = 8 of agg1](agg2) {$v$};
		\node [on grid, on grid, below left  = 2 of agg1] (rack1) {};
		\node [on grid, below right = 2 of agg1] (rack2) {};
		\node [on grid, below left  = 2 of agg2] (rack3) {};
		\node [on grid, below right = 2 of agg2] (rack4) {$u$};
		\node [on grid, rectangle, below = 2.5 of rack1] (srv1) {};
		\node [on grid, rectangle, right = 1.25 of srv1] (srv2) {};
		\node [on grid, rectangle, below = 2.5 of rack2] (srv3) {};
		\node [on grid, rectangle, right = 1.25 of srv3] (srv4) {};
		
		\node [on grid, rectangle, below = 2.5 of rack3] (srv6) {};
		\node [on grid, rectangle, left = 1.25 of srv6] (srv5) {};
		\node [on grid, rectangle, below = 2.5 of rack4,fill=black!20] (srv8) {};
		\node [on grid, rectangle, left = 1.25 of srv8, fill=black!20] (srv7) {};
		\node [on grid, rectangle, right = 1.25 of srv8, fill=black!20] (srv9) {};
		

		\node [on grid, draw=none, right = 6 of agg2]   (cap1) {\huge Rows};
		\node [on grid, draw=none, below = 1.44 of cap1] (cap2) {\huge Racks};   
		\node [on grid, draw=none, below = 2.5 of cap2] (cap3) {\huge Servers};

		\draw[->] (agg1) -- (rack1);
		\draw[->] (agg1) -- (rack2);
		\draw[->] (agg2) -- (rack3);
		\draw[->] (agg2) -- (rack4);
		\draw[->] (rack1) -- (srv1);
		\draw[->] (rack1) -- (srv2);
		\draw[->] (rack2) -- (srv3);
		\draw[->] (rack2) -- (srv4);
		
		\draw[->] (rack3) -- (srv5);
		\draw[->] (rack3) -- (srv6);
		\draw[->] (rack4) -- (srv7);
		\draw[->] (rack4) -- (srv8);
		\draw[->] (rack4) -- (srv9);
		\end{tikzpicture}
		
		 \vspace{-6.5mm}}

	\subcaptionbox{Scenario II\label{fig:scenarioII}}[0.48\textwidth]{
	\begin{tikzpicture}[scale=0.5, transform shape]
	\tikzstyle{every node}=[minimum size=0.9cm, circle, align=center, font=\LARGE,  draw=black]
	
	\pgfmathsetmacro{\xoffset}{1.5}
	\pgfmathsetmacro{\yoffset}{1.5}
	
	\node (agg1) at (0, 0 ) {};
	\node [on grid, right = 8 of agg1](agg2) {$v$};
	\node [on grid, on grid, below left  = 2 of agg1] (rack1) {};
	\node [on grid, below right = 2 of agg1] (rack2) {};
	\node [on grid, below left  = 2 of agg2] (rack3) {};
	\node [on grid, below right = 2 of agg2] (rack4) {$u$};
	\node [on grid, rectangle, below = 2.5 of rack1] (srv1) {};
	\node [on grid, rectangle, right = 1.25 of srv1] (srv2) {};
	\node [on grid, rectangle, below = 2.5 of rack2] (srv3) {};
	\node [on grid, rectangle, right = 1.25 of srv3,fill=black!20] (srv4) {};
	
	\node [on grid, rectangle, below = 2.5 of rack3,fill=black!20] (srv6) {};
	\node [on grid, rectangle, left = 1.25 of srv6] (srv5) {};
	\node [on grid, rectangle, below = 2.5 of rack4] (srv8) {};
	\node [on grid, rectangle, left = 1.25 of srv8,fill=black!20] (srv7) {};
	\node [on grid, rectangle, right = 1.25 of srv8] (srv9) {};
	
	\node [on grid, draw=none, right = 6 of agg2]   (cap1) {\huge Rows};
	\node [on grid, draw=none, below = 1.44 of cap1] (cap2) {\huge Racks};   
	\node [on grid, draw=none, below = 2.5 of cap2] (cap3) {\huge Servers};
	
	\draw[->] (agg1) -- (rack1);
	\draw[->] (agg1) -- (rack2);
	\draw[->] (agg2) -- (rack3);
	\draw[->] (agg2) -- (rack4);
	\draw[->] (rack1) -- (srv1);
	\draw[->] (rack1) -- (srv2);
	\draw[->] (rack2) -- (srv3);
	\draw[->] (rack2) -- (srv4);
	
	\draw[->] (rack3) -- (srv5);
	\draw[->] (rack3) -- (srv6);
	\draw[->] (rack4) -- (srv7);
	\draw[->] (rack4) -- (srv8);
	\draw[->] (rack4) -- (srv9);

	\end{tikzpicture}
}
	\caption{Two scenarios represented by directed trees. Boxes represent placement candidates. Greyed boxes are candidates in the current placement.}
	\label{fig:scenarios}
\end{figure}

In this work, we present a new optimization objective for avoiding correlated failure, and novel algorithms to optimize this objective. See \autoref{fig:scenarios} for an example model, in which three identical replicas of the same block of data are distributed on servers in a data center. As can be seen in \autoref{fig:scenarioI}, a failure in the power supply unit (PSU) on a single rack could result in a situation where every replica of a data block is completely unavailable, whereas in \autoref{fig:scenarioII}, three PSU failures would need to occur in order to achieve the same result. Best practices avoid Scenario I by ensuring that each replica is housed on a separate rack \cite{Weil2006}. However, this simple heuristic can be suboptimal in some cases. For instance, failures that occur higher in the tree can impact the availability of every data replica stored on adjacent racks.

In common technical parlance each of the internal nodes represented in \autoref{fig:scenarios} is termed a \emph{failure domain}. Widely used, modern-day storage area networks such as Parallels' Cloud Storage \cite{Parallels}, and VMWare's Virtual SAN \cite{VMWare} allow the user to specify a hierarchical model of failure domains much like that seen in \autoref{fig:scenarios}. In these models, storage devices which can fail together due to a common hardware dependency are grouped together into a hierarchy. Such models have also appeared in the research literature \cite{RehnSonigo2007,Weil2006}. For instance, the designers of CRUSH proposed a distributed algorithm which pseudo-randomly distributes data across failure domains. In CRUSH, the system administrator is allowed to specify replica placement rules which are used to distribute replicas across multiple failure domains \cite{Weil2006}. While CRUSH allows the user to specify rules which may lead to a undesirable placement such as that seen in \autoref{fig:scenarioI}, our work focuses on the alternative approach of finding an \emph{optimal} replica placement. In the process, we develop a novel objective function which provides insight into what ``optimal" means in terms of replica placement.

Concurrent with our work, Korupolu and Rajaraman considered several important extensions and variants of the reliable replica placement problem which they term ``failure-aware placement" \cite{Korupolu2016}. Their work explores important variants of replica placement which allows for a user-specified reliability factor at each node. They define an adversarial optimization problem which finds a \emph{fractional} number of replicas placed at each server. Once having optimized a fractional placement, they provide a randomized rounding approach which attains the optimum value of the fractional solution in expectation. However, the problem which they formulate relies heavily on the assignment of reliability factors to nodes. In the case where all nodes have the same reliability factor, the algorithm of Korupolu and Rajaraman assigns each leaf node an equal assignment of replicas, regardless of the structure of the heirarchy. When randomized rounding is performed on such a fractional placement, all discrete placements will be equally likely. In contrast, our formulation distinguishes between discrete placements based upon the structure of the hierarchy.

Other work on reliability in storage area networks has focused on objectives such as mean time to data loss \cite{Chen2007,Lian2005}. These exemplify an approach towards correlated failure which we term ``measure-and-conquer". In measure-and-conquer approaches \cite{Bakkaloglu2002,Ford2010,Nath2006,Weatherspoon2002}, a measured degree of correlation is given as a parameter to the model.  In contrast, we model explicit causal relations among failure events which we believe give rise to the correlation seen in practice. More recently, Pezoa and Hayat \cite{Pezoa2014} have presented a model in which spatially correlated failures are explicitly modeled. However, their main goal is the accurate statistical modeling of task redistribution and scheduling in the data center, whereas we are focused on algorithms for replica placement with provable guarantees. In the databases community, work on replica placement primarily focused on finding optimal placements in storage area networks with regard to a particular distributed access model or mutual exclusion protocol \cite{Hu2001,Shekhar2001,Zhang2009}. Similarly, work from the networking community tends to address particular allocation policies or quality of service objectives such as load balancing \cite{RehnSonigo2007,Wu2008}, in contrast to the present work on correlated failure. In general, much of the work from these communities focuses on system models and goals which are substantially different from our own. Recently, there has been a surge of interest in computer science concerning cascading failure in networks. While our model is conceptually related to this work, it does not appear to directly follow from any published model \cite{Blume2011,Kim2010,Nie2014,Zhu2014}. Current work in this area is focused on fault-tolerant network design \cite{Blume2011}, modeling cascading failure \cite{Kim2010,Nie2014}, and developing techniques for adversarial analysis \cite{Zhu2014}. To our knowledge, no one has yet considered the problem of replica placement in such models.

\paragraph{Our Contributions:} In this work, we make the following contributions. We first present a novel optimization goal for avoiding correlated failure and formulate two novel replica placement problems which optimize for this goal.
Intuitively, in our problems, the optimization goal is to choose a placement in which \emph{correlated} failures disable as ``few'' replicas as possible. 
We then present two efficient algorithms for attaining our optimization goal in trees. Both algorithms are based on dynamic programming. 
The first algorithm finds an optimal placement of replicas for a single block of data. It has a running time of $O(\n + \repfact \log \repfact)$, where $\n$ denotes the number of vertices in the tree and $\repfact$ denotes the number of replicas to be placed.
Our second algorithm finds an optimal placement of replicas for multiple data blocks. We define the \emph{skew} of the desired placement to be the difference between the largest and smallest replication factor among all files. When the skew is at most a constant, we present a replica placement algorithm which runs in polynomial time.
Finally, we establish the NP-hardness of reliable replica placement in bipartite graphs, even when placing $\rho$ replicas of only one block.

\paragraph{Roadmap:}
The rest of the text is organized as follows. We describe our system model and formally define single- and multi-block replica placement problems in \autoref{model}. We describe our algorithm for finding an optimal placement for a single file in \autoref{s:single-block}. 
We describe our algorithm for finding an optimal placement for multiple files in \autoref{s:mp-balancing}. Finally, we present an overview of known complexity results in \autoref{s:np-hard} before discussing future work and concluding in \autoref{s:future-work}.

\section{Modeling}
\label{model}
We model dependencies among failure events as a directed, rooted tree in which all edges point away from the root (i.e. an \emph{arborescence}) where nodes represent failure events and a directed edge from node $u$ to node $v$ indicates that the occurrence of failure event $u$ triggers the occurrence of failure event $v$. These nodes correspond either to real-world hardware unsuitable for storage (e.g. a top-of-rack (ToR) switch), or to abstract events which have no associated physical component (e.g. software failure, and maintenance outages). We refer to this tree as the \emph{failure model}.


Given such a tree as input, we consider the problem of selecting nodes on which to store data replicas. Roughly, we define a \emph{placement problem} as the problem of selecting a subset of the leaf nodes, hereafter referred to as a \emph{placement}, from the failure model so as to satisfy some safety criterion. In our application, only leaf nodes, which represent storage servers, are candidates to be part of a placement.

Let $\graph{\V}{\E}$ be an arborescence with vertices in $\V$ and arcs in $\E$. Let $\F \subseteq \V$ denote the set of internal nodes, and let $\C$ denote the set of leaves. We are interested in finding a \emph{placement} of size $\repfact$, which is defined to be a set $\P \subseteq \C$, with $|\P| = \repfact$. 
There are two types of nodes in tree $\graph{\V}{\E}$: internal nodes, which represent failure events, and leaf nodes, which represent storage servers.  A directed edge from internal node $e_1$ to internal node $e_2$ indicates that, in the worst-case, the occurrence of failure event $e_1$ triggers the occurrence of failure event $e_2$. A directed edge from internal node $e$ to leaf node $\ell$ indicates that, in the worst-case, the occurrence of event $e$ compromises storage server $\ell$. We consider failure to act transitively as regards the unavailability of replicas. That is, if a failure event occurs, all failure events reachable from it in $(V,A)$ also occur.

To quantify the impact of the failure of an event, we define the notions of \textit{failure number} and \textit{failure aggregate}.

\begin{definition}[failure number]
	\label{def-failure}
	Given a vertex $u \in \V$ and a placement $\P$, the \textit{failure number} of $u$ with respect to $\P$, denoted $\fnum{u}{\P}$ is defined as
	$$\fnum{u}{\P} \defined | \{ \ell \in P \mid \ell \text{ is reachable from } u \text{ in } (\V,\E)  \}|.$$
	In particular, $\fnum{u}{\P}$ is the number of leaves in $\P$ whose correct operation could be compromised by the occurrence of event $u$.
\end{definition}
As an example, node $u$ in \autoref{fig:scenarios} has failure number $3$ in Scenario I, and failure number $1$ in Scenario II. Note that with this definition, leaf nodes also have a failure number.

The failure number captures a conservative criterion for a safe placement. Our goal is to find a placement which does not induce large failure numbers in any event.
To collect all of the failure numbers into a single metric, we define the \textit{failure aggregate}, a novel measure that accounts for the failure number of every event in the model.

\begin{definition}[failure aggregate]
	The \emph{failure aggregate} of a placement $\P$ is a vector in $\mathbb{N}^{\repfact+1}$, denoted $\ff{\P}$, where \mbox{$\ff{\P} \defined \langle p_0, p_1, \ldots, p_\repfact\rangle$}, and each $p_i$ is defined as $$p_i \defined \left| \big\{ e \in \F \cup \C \mid \fnum{e}{\P} = \repfact - i\big\} \right|.$$
\end{definition}

Intuitively, $p_i$ is the number of nodes whose failure allows $\repfact - i$ replicas to survive. In \autoref{fig:scenarios}, Scenario I has failure aggregate of $\langle 2, 0, 3, 10 \rangle$ and Scenario II has failure aggregate of $\langle 0, 1, 7, 7 \rangle$ in \autoref{fig:scenarioII}.

All of the problems we consider in this work involve optimizing the failure aggregate. When optimizing a vector quantity, we must choose a meaningful way to totally order the vectors. In the context of our problem, we find that ordering the vectors with regard to the \emph{lexicographic order} naturally encodes our intuition behind an ``optimal" placement.

\begin{definition}[lexicographic order]
	The \emph{lexicographic order} $\llt$ between vectors $\vec{x} = \langle x_0, ..., x_d\rangle$ and $\vec{y} = \langle y_0,...,y_d\rangle$ can be defined via the following formula:
	\[\vec{x} \llt \vec{y} \iff \exists~ j \in [0,\repfact] ~:~\big(  x_j < y_j ~\wedge ~\forall ~i < j : [x_i = y_i  ]\big)\]
	The above definition extends to a definition for the symbol $\lleq$ in the usual way. We use terms \emph{lexico-minimum} and \emph{lexico-minmizes} as an efficient short-hand for phrases ``minimum in the lexicographic order" and ``minimizes with respect to the lexicographic order" respectively.
\end{definition}

To see why using the lexicographic ordering is desirable, consider a placement $\P$ which lexico-minimizes $\ff{\P} = \langle p_0, p_1, ..., p_\repfact \rangle$ among all possible placements. Such a placement is guaranteed to minimize $p_0$, i.e. the number of nodes which compromise \emph{all} of the entities in our placement. Further, among all solutions minimizing $p_0$, $\P$ also minimizes $p_{1}$, the number of nodes compromising \emph{all but one} of the entities in $\P$, and so on for $p_{2}, p_{3},\ldots, p_{\repfact}$. Clearly, the lexicographic order nicely prioritizes minimizing the entries of the vector in an appealing manner. 

This gives rise to the following novel optimization problem.

\begin{problem}[Optimal Single-block Placement]
	\label{p:graph-single-placement}
	Given an arborescence $\graph{\V}{\A}$ with leaves in $\C$, and positive integer $\repfact$, with $\repfact < |\C|$, find a placement $\P \subset \C$ with size $\repfact$, such that $\ff{\P}$ is lexico-minimum.
\end{problem}

Essentially, \autoref{p:graph-single-placement} concerns placing $\repfact$ replicas of a single block of data. Notice that in this problem we enforce that no more than one replica may be placed at any given leaf. This is reasonable, as co-locating two replicas on the same server would defeat the purpose of replication. In \autoref{s:single-block}, we present an $O(\n + \repfact \log \repfact)$ algorithm for solving \autoref{p:graph-single-placement}. 

While a worthy goal, a solution to \autoref{p:graph-single-placement} only optimizes the placement of a single set of replicas. In the data center, multiple sets of replicas must co-exist simultaneously. To address this crucial use-case, we also present an algorithm which simultaneously optimizes multiple replica placements at once. To this end, we define a \emph{multi-placement} $\MP$ to be an $\m$-tuple of placements, $\MP \defined (\P_1, ..., \P_m)$. In the multi-placement context, we refer to placements $P_1, ...,P_m$ as \emph{\block{s}}, and we refer to each block by its position in the tuple (e.g. placement $P_1$ is block 1, ... placement $P_i$ is block $i$, etc.) In the single-block case, it made sense to ensure that no more than one replica may be placed at any leaf node. In the multi-block case, we must allow multiple replicas from different blocks to be collocated at a leaf node. To this end, we include a capacity $c(\ell)$ for each leaf node $\ell$ in our formulation, and ensure no more than $c(\ell)$ replicas are placed on $\ell$. However, each placement remains a subset of the set of leaves, which means no placement in a multi-placement may place more than one replica on any given leaf.

The failure aggregate $\ff{\P}$ defined above extends to multi-placements by taking the sum over all placements in the multi-placement. To allow each block to have a distinct number of replicas, we pad the failure aggregates on the \emph{left} with additional zeroes to achieve a vector with the proper length. More specifically, if placement $\P_i$ consists of $\repfact_i$ replicas, then each failure aggregate is defined as a vector of length $\max[\repfact_1, ...,. \repfact_m]$. We refer to this quantity as the \emph{girth} of the multi-placement, and denote it by $\repfact$. Using this notation, the definition of the failure aggregate does not require any modification. We can thus define $\g{\MP} \defined \sum_{i = 1}^m \ff{\P_i}$, where $\ff{\P_i}$ takes on values in $\mathbb{N}^{\repfact+1}$, where $\repfact$ is understood to be the girth of the multi-placement $\MP$. This leads naturally to the following problem.

\begin{problem}[Optimal Multi-block Placement]
	\label{p:graph-multi-placement}
	Given an arborescence, $\graph{\V}{\A}$, with leaves in $\C$, where each leaf $\ell \in \C$ has assigned capacity \mbox{$c(\ell) \in \mathbb{Z}^+$}, a positive integer $m$, and $m$ positive integers $\repfact_1, \repfact_2, \ldots, \repfact_m$ for which  \mbox{$\sum_{i=1}^m \repfact_i \leq \sum_{\ell \in \C} c(\ell),$} find a multi-placement $\MP = (\P_1, ..., \P_m)$, which lexico-minimizes $\g{\MP}$ subject to the constraints that 
	\begin{enumerate}[1)] \item for each $\ell \in \C$, $\MP$ contains no more than $c(\ell)$ copies of $\ell$, and 
		\item $|\P_i| = \repfact_i$ for each $i=1, \ldots, m$.
	\end{enumerate}
\end{problem}

In the context of an instance of \autoref{p:graph-multi-placement}, we define:
\begin{enumerate}[label={\alph*)}]
	\item the \emph{\totalSize} as the sum of the sizes of all blocks, denoted by $\mpRepFact = \sum_{i=1}^m \rho_i$,
	\item the \emph{\girth} as the maximum size of any block, denoted by $\rho = \max (\repfact_1 ,..., \repfact_m)$,
	\item the \emph{\skewterm} as the absolute difference between the largest and smallest replication factors of each block, denoted by $\skewfactor$. For convenience, we assume that $\skewfactor \geq 1$, that is, $\skewfactor = \max ( \max_i \rho_i - \min_i \rho_i, 1 )$. 
\end{enumerate} 

Storage area networks used widely in practice make use of multi-placements with bounded \skewterm\footnote{From the VMWare Virtual SAN Administrator's Guide: ``For $n$ failures tolerated, $n+1$ copies of the virtual machine object are created"\cite{VMWare}. In contrast, the Parallels system allows the number of replicas per chunk to vary between a minimum and maximum value \cite{Parallels}. While our current work does not fully address the replication practices of the Parallels system, our algorithm can still be applied. The Parallels system must store a minimum number of replicas, and we can place these replicas optimally using our algorithm. Moreover, the recommended settings in a Parallels cluster use a maximum and minimum replication factor of 2 and 3 respectively (i.e. the recommended skew is at most one).}.
In \autoref{s:mp-balancing} we present an exact dynamic programming algorithm which, for any specification $\repfact_1,...,\repfact_m$ finds an optimal multi-placement of $m$ blocks with skew $\delta$ and girth $\repfact$. Our algorithm runs in polynomial-time when $\delta$ is a fixed constant. In any case, $\delta < \repfact$, and in practice, both values are small constants \cite{Parallels,VMWare}.

Throughout the paper, anytime we minimize or compare vector quantities we are minimizing or comparing them in the lexicographic order. Moreover, we reserve the symbols $\P$, $\ff{\P}$ and $p_i$, to have their meaning as defined in this section. We will also consistently use $\fnum{u}{\P}$ to refer to the failure number of node $u$ in placement $\P$. The symbol $\mpRepFact$ will be consistently used to denote the \totalSize of a multi-placement, whereas $\repfact$ will denote the size of a placement in \autoref{s:single-block} and the girth of a multi-placement in \autoref{s:mp-balancing} and beyond.

\section{Solving Single-block Replica Placement}
\label{s:single-block}

In this section, we describe an algorithm for solving \autoref{p:graph-single-placement}. First, we prove that any optimal placement must be balanced. We then exploit this balancing property to develop an $O(\n + \repfact \log \repfact)$ algorithm for finding an optimal placement of size $\repfact$ on the leaves of an arborescence with $\n$ nodes.

As an aside, we note that a greedy algorithm also works for this problem. Briefly, the greedy solution forms a partial placement $\P'$, to which new replicas are added one at a time, until $\repfact$ replicas have been placed overall. $\P'$ starts out empty, and at each step, the leaf $u$ which lexico-minimizes $\ff{\P' \cup \{u\}}$ is added to $\P'$. That this na\"{i}ve greedy approach works correctly is not immediately obvious. It can be shown via an exchange argument that each partial placement found by the greedy algorithm is a subset of \emph{some} optimal placement. Proving this is straight-forward, yet somewhat tedious (see \autoref{app:greedy-proof}). However, as we show below, the running time of the greedy approach does not compare favorably with our $O(\n + \repfact \log \repfact)$ algorithm.

The greedy approach solves \autoref{p:graph-single-placement} in $O(\n^2\repfact)$ time. Each iteration requires checking $O(|\C|)$ leaves for inclusion. For each leaf $\lf$ which is checked, every node on a path from $\lf$ to the root must have its failure number recomputed. Both the length of a leaf-root path and the number of leaves can be bounded by $O(\n)$ in the worst case, yielding $O(\n^2\repfact )$ time overall.

\subsection{The Balancing Property for Optimal Placements}

Our algorithm for optimal replica placements hinges on the fact that any optimal placement must be \emph{balanced}. Intuitively, this means that the number of replicas which are descendants of any internal node $u$ are distributed among the children of $u$ so that no single child accommodates too many replicas. The intuition behind balancing is made precise in \autoref{d:bal}, which follows below.

Let $\Csub{u}$ be the set of leaves which are descendants of node $u$. We refer to $|\Csub{u}|$ as the \emph{capacity of node u}.
\begin{definition}\label{d:bal}
	Given a placement $\P$, let node $u$ have children $\children{1}{\nChildren}$.  Let $\Csub{\child{i}}$ be the set of leaves which are descendants of child $\child{i}$. Node $u$ is said to be \emph{balanced with respect to placement $P$} if, for all $\child{i}, \child{j} \in \{\children{1}{\nChildren}\}$ $$|\Csub{\child{i}}| > \fnum{\child{i}}{\P} \implies \fnum{\child{j}}{\P} \leq \fnum{\child{i}}{\P} + 1,$$
	and the above condition is referred to as the \emph{balancing condition}.
	
	Moreover, placement $\P$ is said to be \emph{balanced} if, all nodes $u \in V$ are balanced with respect to $\P$.
\end{definition}
The balancing condition holds trivially if $|\Csub{\child{i}}| = \fnum{\child{i}}{P}$. We say that children where $|\Csub{\child{i}}| = \fnum{\child{i}}{P}$ are \emph{filled}, and children where $|\Csub{\child{i}}| > \fnum{\child{i}}{P}$ are \emph{unfilled}. As a consequence of the balancing condition the replicas are ``evenly spread" among the unfilled children. Our algorithm exploits the following result to achieve an $O(\n + \repfact \log \repfact)$ running time for minimizing $\ff{\P}$ in a tree.

\begin{theorem}\label{thm-balanced-sufficiency}
	Any placement $\P$ in which $\ff{\P}$ is lexico-minimum among all placements for a given tree must be balanced.
\end{theorem}

\begin{proof}
	Suppose $\P$ is not balanced, yet $\ff{\P}$ is lexico-minimum. We derive a contradiction as follows.
	
	Let $u$ be an unbalanced node, then $u$ must have children $\child{i}$ and $\child{j}$ such that $\child{i}$ is unfilled and $\fnum{\child{i}}{\P} + 1 < \fnum{\child{j}}{\P}$. Since $\child{i}$ is unfilled, we must be able to take one of the replicas placed on a leaf of $\child{j}$ and place it on $\child{i}$ instead. Leaves $\lf_i \in \Csub{\child{i}} \setminus P$ and $\lf_j \in \Csub{\child{j}} \cap P$ must exist, because $\child{i}$ is unfilled, and $\Csub{\child{j}} \cap P$ is non-empty. Set $\P^\ast \defined (\P \setminus \{\lf_j\}) \cup \{\lf_i\}$. We will show $\P^\ast$ is a strictly better placement than $P$.
	
	Let $\ff{\P} = \langle p_0, ..., p_\repfact \rangle$, and $\ff{\P^\ast} = \langle p_0^\ast, ..., p_\repfact^\ast \rangle$. For convenience, let $\fnum{\child{j}}{\P} = a$. To show that $\ff{\P} \lgt \ff{\P^\ast}$, we aim to prove that $p_{\repfact - a}^\ast < p_{\repfact - a}$, and that $p^\ast_k = p_k$ for all $k<\repfact - a$. We will concentrate on proving the former, and afterwards show that the latter follows easily.
	
	Let $\path{\lf_i}{\child{i}}$ (respectively $\path{\lf_j}{\child{j}}$) be the nodes on the unique path from $\lf_i$ to $\child{i}$ (respectively $\lf_j$ to $\child{j}$). As a result of the swap, note that only the nodes on these paths have their failure numbers modified. Therefore, to prove $p^\ast_{\repfact-a} < p_{\repfact-a}$, it suffices to consider the failure numbers of the nodes in $\path{\lf_i}{\child{i}} \cup \path{\lf_j}{\child{j}}$. Let $S^-$ (respectively $S^+$) be the set of nodes whose failures change from $a$ (respectively change to $a$), as a result of the swap. Formally, we define:
	$$S^- \defined \{ v \in \V \mid \fnum{v}{\P} = a, \fnum{v}{\P^\ast} \neq a \},$$
	$$S^+ \defined \{ v \in \V \mid \fnum{v}{\P} \neq a, \fnum{v}{\P^\ast} = a \}.$$
	By definition, $p^\ast_{\repfact-a} = p_{\repfact-a} - |S^-| + |S^+|.$ We claim that $|S^-| \geq 1$ and $|S^+| = 0$, which suffices to show that $p^\ast_{\repfact-a} < p_{\repfact-a}$, as required.
	
	To show that $|S^-| \geq 1$, note that $\fnum{\child{j}}{\P} = a$ by definition, and after the swap, the failure number of $\child{j}$ decreases. Therefore, $|S^-| \geq 1$.
	
	To show that $|S^+| = 0$, we must prove that no node in $\path{\lf_i}{\child{i}} \cup \path{\lf_j}{\child{j}}$ has failure number $a$ after the swap has occurred. We show the stronger result, that all such nodes' failure numbers are strictly less than $a$.
	
	Let $v_j$ be a node on the path $\path{\lf_j}{\child{j}}$, and consider the failure number of $v_j$. Notice that for every such $v_j$, we have that
	$$\fnum{v_j}{\P^\ast} \leq \fnum{\child{j}}{\P^\ast} = a -1 < a,$$
	where the first inequality follows since the failure number of any node is less than or equal to that of any of its ancestors, and $\fnum{\child{j}}{\P^\ast} = a -1$, since the number of replicas on $\child{j}$ decreases by 1 as a result of the swap. Therefore, $\fnum{v_j}{\P^\ast} < a$, for any $v_j \in \path{\lf_j}{\child{j}}$.
	
	Likewise, let $v_i$ be a node on the path $\path{\lf_i}{\child{i}}$, and consider the failure number of $v_i$. Since the swap added a replica at node $\child{i}$, clearly $\fnum{v_i}{P^\ast} = \fnum{v_i}{\P} + 1$. Recall also that $\fnum{\child{i}}{\P} + 1 < \fnum{\child{j}}{\P}$, therefore, for all $v_i$, we have
	$$\fnum{v_i}{\P} \leq \fnum{\child{i}}{\P} < \fnum{\child{j}}{\P} - 1 = a -1,$$
	which establishes that $\fnum{v_i}{\P} < a-1$. Substituting $\fnum{v_i}{\P^\ast}$ yields that $\fnum{v_i}{\P^\ast} < a$ for any $v_i \in \path{\lf_i}{\child{i}}$.
	
	Therefore, no node in $\path{\lf_i}{\child{i}} \cup \path{\lf_j}{\child{j}}$ has failure number $a$, so $|S^+| = 0$, as desired. Moreover, since we showed that $\fnum{v}{\P^\ast} < a$, for any node $v \in \path{\lf_i}{\child{i}} \cup \path{\lf_j}{\child{j}}$, and these are the only nodes whose failure numbers change, we have also proven that $p_k = p_k^\ast$ for all $k \leq \repfact - a$, thus completing the proof. \qed
\end{proof}

\subsection{An $O(n + \repfact\log\repfact)$ Algorithm for Optimal Placements}
To ease our exposition, we first describe an $O(n\repfact)$ version of the algorithm. In \autoref{s:rhologrho}, we make modifications to improve the running time to $O(n + \repfact \log \repfact)$.

Our algorithm finds an optimal placement that is balanced. Conceptually, we can think of our algorithm as assigning $\repfact$ replicas to the root of the tree, and assigning these replicas to children of the root in some fashion. We then recursively carry out the same procedure on each child, thereby proceeding down the tree, at each step ensuring that the replicas are assigned according to the balancing property. However, as we move down the tree, we observe that for certain nodes, the balancing condition only determines the number of replicas placed up to an additive factor of $\pm 1$. Therefore, on the way back up the tree, we solve a minimization problem to determine how many replicas to place on each node in order to minimize $\ff{\P}$. To concisely communicate the ``total number of replicas assigned to node $u$", we shall refer to the number of replicas assigned to a node of the tree as its \emph{mass}. 

More specifically, we proceed as follows. Before the recursive procedure begins, we first record the capacity of each node $u$ via a post-order traversal of the tree. Our algorithm is then executed in two consecutive phases. During the \emph{divide} phase, the algorithm is tasked with dividing the mass assigned to node $u$ among the children of $u$. First, for each child $\child{i}$, we determine $\minRepOn{\child{i}}$, the minimum possible mass on $\child{i}$ in any balanced placement. After the divide phase, we have determined which children $c$ are \emph{\dfill}, $(\minRepOn{\child{i}} = |\Csub{\child{i}}|)$ and which are \emph{\punfill}, $(\minRepOn{\child{i}} < |\Csub{\child{i}}|)$. \Dfill children have a mass equal to their capacity, and require no further optimization. To achieve balancing, each \punfill child labeled $\child{i}$ must have a mass of either $\minRepOn{\child{i}}$ or $\minRepOn{\child{i}} + 1$. 
The algorithm is then recursively called on each \punfill child to obtain optimal subproblems of mass $\minRepOn{\child{i}}$ and $\minRepOn{\child{i}} + 1$ for their subtrees. After this recursive call is complete, two optimal solutions are available at each \punfill child. The \emph{combine} phase then chooses which of the two placements should be used at each \punfill child so as to obtain a minimum overall, thereby determining the final mass for each such child.

\subsubsection{Divide Phase}

When node $u$ is first considered by the divide phase, there are at most two possible values for its final mass. Let these values be $\minRepOn{u}$ and $\minRepOn{u} + 1$. Let $u$ have $t$ children, labeled $\children{1}{\nChildren}$, with capacities $|\Csub{\child{1}}|, ..., |\Csub{\child{\nChildren}}|$. The divide phase determines which children are \dfill, and which are \punfill.

The set of \punfill children can be determined in $O(\nChildren)$ time in a manner similar to the algorithm for the Fractional Knapsack Problem \cite{Dantzig1957}. We iteratively process the children of $u$ and, based upon their capacities, determine whether they are \dfill or \punfill in a balanced placement of $\repfact$ replicas. We ensure that in each iteration, at least one-half of the children with undetermined status have their \dfill/\punfill status determined. To determine which half, the median capacity child with undetermined status is found using the median of medians algorithm \cite{Cormen2009}. Based upon the number of replicas yet to be ``claimed" by \dfill children, either 
\begin{inparaenum}[a)]
	\item the set of children with capacity greater than or equal to the median are labeled as \punfill, or
	\item the set of children with capacity less than or equal to the median are labeled as \dfill.
\end{inparaenum}
The algorithm then recurses on the remaining \punfill children. In the process of computing the \dfill and \punfill children values of $\minRepOn{\child{i}}$ are determined for all \punfill children $\child{i}$. See \autoref{a:get-filled} for pseudocode describing this procedure, and \autoref{f:divide-example} for a sample execution. The proof of correctness is straight-forward but tedious, and is provided in \autoref{app:labeling-proof}. Because the algorithm runs in time $O(\nChildren)$ for a node with $\nChildren$ children, the divide phase takes $O(\n)$ time over the entire tree.

\begin{algorithm}[t]
	\SetKwFunction{getFilled}{Label-Children}
	\SetKwProg{Fn}{Function}{begin}{end}
	\Fn{\getFilled{$\{ \children{1}{\nChildren} \}$, $r$}}{
		$\filled \gets \emptyset$;  \tcp*[r]{$\filled$ := \dfill children}
		$\unfilled \gets \emptyset$; \tcp*[r]{ $\unfilled$ :=  \punfill children}
		$M \gets \{ \children{1}{\nChildren} \}$  \tcp*[r]{$M$ := unassigned children}
		$s \gets r$  \tcp*[r]{$s$ := number of replicas not yet permanently assigned}
		\While{$M \neq \emptyset$}{
			\label{ln:while}
			$med \gets \text{ median capacity of children in $M$ }$\;
			$M_{\ell} \gets \{\child{i} \in M : |\Csub{\child{i}}| < med  \}$\;
			$M_{e} \gets \{\child{i} \in M : |\Csub{\child{i}}| = med \}$\;
			$M_{g} \gets \{\child{i} \in M : |\Csub{\child{i}}| > med \}$\;
			$x \gets s - \sum_{\child{i} \in M_{\ell}} |\Csub{\child{i}}|$ \;
			\uIf{$x < (med -1) \cdot (|\unfilled| + |M_{e}| + |M_{g}|)$}{\label{l:case1}
				$\unfilled \gets \unfilled \cup M_{e} \cup M_{g}$\tcp*{$M_{e} \cup M_{g}$ \punfill}
				$M \gets M - (M_{e} \cup M_{g})$\; 
				
			}\uElseIf{$x \geq (med) \cdot (|\unfilled| + |M_{e}| + |M_{g}|)$} {\label{l:case2}
				$\filled \gets \filled\cup M_{\ell} \cup M_{e}$\tcp*{$M_{\ell} \cup M_{e}$ \dfill}
				$M \gets M - (M_{\ell} \cup M_{e})$\;
				$s \gets x - \sum_{\child{i} \in M_{e}} |\Csub{\child{i}}|$\;
			}
			\Else(\tcp*[f]{$M_{\ell}$ \dfill, $M_{e} \cup M_{g}$ \punfill}) {
				$\unfilled \gets \unfilled \cup M_{e} \cup M_{g}$ \; $\filled \gets \filled \cup M_{\ell}$\;
				$M \gets \emptyset$\;
			}
		}
		\Return{($\filled$, $\unfilled$)} \tcp*[r]{return \dfill and \punfill children}
	}
	\caption{Determine \dfill and \punfill nodes}\label{a:get-filled}
\end{algorithm}

\begin{figure}
\begin{tikzpicture}[scale=0.95, transform shape,text height=2ex]
\usetikzlibrary{patterns}
\newlength{\hatchspread}
\newlength{\hatchthickness}
\newlength{\hatchshift}
\newcommand{\hatchcolor}{}
\tikzset{hatchspread/.code={\setlength{\hatchspread}{#1}},
	hatchthickness/.code={\setlength{\hatchthickness}{#1}},
	hatchshift/.code={\setlength{\hatchshift}{#1}},
	hatchcolor/.code={\renewcommand{\hatchcolor}{#1}}}
\tikzset{hatchspread=3pt,
	hatchthickness=0.4pt,
	hatchshift=0pt,
	hatchcolor=black}
\pgfdeclarepatternformonly[\hatchspread,\hatchthickness,\hatchshift,\hatchcolor]
{custom north east lines}
{\pgfqpoint{\dimexpr-2\hatchthickness}{\dimexpr-2\hatchthickness}}
{\pgfqpoint{\dimexpr\hatchspread+2\hatchthickness}{\dimexpr\hatchspread+2\hatchthickness}}
{\pgfqpoint{\dimexpr\hatchspread}{\dimexpr\hatchspread}}
{
	\pgfsetlinewidth{\hatchthickness}
	\pgfpathmoveto{\pgfqpoint{\dimexpr\hatchshift-0.15pt}{-0.15pt}}
	\pgfpathlineto{\pgfqpoint{\dimexpr\hatchspread+0.15pt}{\dimexpr\hatchspread-\hatchshift+0.15pt}}
	\ifdim \hatchshift > 0pt
	\pgfpathmoveto{\pgfqpoint{-0.15pt}{\dimexpr\hatchspread-\hatchshift-0.15pt}}
	\pgfpathlineto{\pgfqpoint{\dimexpr\hatchshift+0.15pt}{\dimexpr\hatchspread+0.15pt}}
	\fi
	\pgfsetstrokecolor{\hatchcolor}
	\pgfusepath{stroke}
}
\tikzstyle{every node}=[minimum size=0.5cm, circle, align=center,  draw=black]

\node (r) {$u$};
\node [below=0.5cm of r, xshift=-0.75cm] (3) {};
\node [left=of 3]   (2) {};
\node [left=of 2]   (1) {};
\node [right=of 3]   (4) {};
\node [right=of 4]   (5) {};
\node [right=of 5]   (6) {};

\draw[->] (r) -- (1.north east);
\draw[->] (r) -- (2.north east);
\draw[->] (r) -- (3.north);
\draw[->] (r) -- (4.north);
\draw[->] (r) -- (5.north west);
\draw[->] (r) -- (6.north west);

\node [draw=none, below=0.25cm of 1] (11) {1};
\node [draw=none, below=0.25cm of 2] (12) {2};
\node [draw=none, below=0.25cm of 3] (13) {4};
\node [draw=none, below=0.25cm of 4] (14) {5};
\node [draw=none, below=0.25cm of 5] (15) {9};
\node [draw=none, below=0.19cm of 6] (16) {11};

\node [draw=none, left=0cm of 11] {$|\Csub{i}| = $};

\node [draw=none, above right=-0.25cm of r, xshift=0.25cm] {$\repOn{u} \in \{20, 21\}$};

\node [draw=none, rectangle, right=1cm of 16] (e1) {$med = 4.5$};
\node [draw=none] at ($(13)!0.5!(14)$) (mid) {$\LARGE\mid$};
\node [draw=none, rectangle, below=0cm of e1]     (e2) {$x=20 \not\leq 3.5(3) = 10.5$};
\node [draw=none, rectangle, below=-0.1cm of e2]  (e3) {$x=20 > 4.5(3) = 13.5$};

\node [draw=none] at (e3 -| 11) {$F$};
\node [draw=none] at (e3 -| 12) {$F$};
\node [draw=none] at (e3 -| 13) {$F$};

\draw [transform canvas={yshift=-2.5mm}] (e3.east -| 1.west) -- (e3.east);

\node [draw=none, yshift=-0.5cm] at (14 |- e3) (34) {5};
\node [yshift=-0.5cm] at (15 |- e3) (35) {9};
\node [draw=none, yshift=-0.5cm] at (16 |- e3) (36) {11};

\node [draw=none, rectangle, right=1cm of 36] (e4) {$med = 9$};
\node [draw=none, rectangle, below=0cm of e4] (e5) {$x=13 \leq 8(2) = 16$};

\draw [pattern=custom north east lines, hatchspread=10pt] ($(e4.north west -| 1.west)+(0,-0.1)$) rectangle (e4.south east -| mid.west);

\node [draw=none] at (e5 -| 15) {$U$};
\node [draw=none] at (e5 -| 16) {$U$};

\draw [transform canvas={yshift=-2.5mm}] (e5.east -| 1.west) -- (e5.east);

\node [yshift=-0.5cm] at (14 |- e5) (34) {5};

\node [draw=none, rectangle] at (34 -| e4) (e6) {$med = 5$};
\node [draw=none, rectangle, below=0cm of e6] (e7) {$x=13 \leq 5(3) = 15$};

\node [draw=none] at (e7 -| 14) {$U$};

\draw [pattern=custom north east lines, hatchspread=10pt] ($(e6.north west -| 1.west)+(0,-0.1)$) rectangle (e6.south east -| mid.west);
\draw [pattern=custom north east lines, hatchspread=10pt] ($(e6.north west -| 5.west)+(0,-0.1)$) rectangle (e6.south east -| 6.east);

\end{tikzpicture}

	\caption{Example execution of \autoref{a:get-filled}. Either $20$ or $21$ replicas are placed on node $u$, which has six children, with capacities $1,2,4,5,9$ and $11$. Each iteration is divided by a line. The node with the median capacity is circled in each iteration, except in the first iteration, where no node has capacity equal to the median. Computation of the branch conditions at lines \ref{l:case1} and \ref{l:case2} are shown on the right.}\label{f:divide-example}
\end{figure}

We next show that computing two values at each child is all that is required to compute \emph{both} placements of mass $\minRepOn{u}$ and $\minRepOn{u}+1$ at node $u$. This avoids a combinatorial explosion by keeping constant the number of subproblems considered at each node throughout the recursion. 

\begin{theorem}\label{thm-two-values}
	Let $\unfilled$ and $\filled$ be the set of \punfill and \dfill children found by \autoref{a:get-filled}, and let $\R$ be  the minimum number of replicas to be distributed among the \punfill children, i.e. $\R := \minRepOn{u} - \sum_{\child{i} \in \filled} |\Csub{\child{i}}|$. In any case where $\minRepOn{u}$ or $\minRepOn{u} + 1$ replicas must be balanced among $t$ \punfill children, it suffices to consider placing either $\smallRep$ or $\bigRep$ children at each \punfill child.
\end{theorem}

\begin{proof}
	We first show that if $x$ replicas are to be distributed among $|\unfilled|$ \punfill children with sufficient capacity in a \emph{balanced} manner, then each child needs to store either $\floorfrac{x}{|\unfilled|}$ or $\ceilfrac{x}{|\unfilled|}$ replicas. Assume, on the contrary, that every child has capacity of at least $\ceilfrac{x}{|\unfilled|}$ but some child stores either at most $\floorfrac{x}{|\unfilled|} - 1$ replicas or at least $\ceilfrac{x}{|\unfilled|} + 1$ replicas.
	There are two cases depending on whether or not $x \bmod |\unfilled| = 0$.
	
	\begin{enumerate}[label={Case \roman*}]
		
		\item ($x \bmod |\unfilled| = 0$):~~In this case, $\floorfrac{x}{|\unfilled|} = \ceilfrac{x}{|\unfilled|} = \nicefrac{x}{|\unfilled|}$. If some child stores at most $\nicefrac{x}{|\unfilled|} - 1$ replicas, then some other child must store at least $\nicefrac{x}{|\unfilled|} + 1$, and vice versa. This, in turn, violates the balancing property.

		\item ($x \bmod |\unfilled| \neq 0$):~~In this case, $\floorfrac{x}{|\unfilled|} < \nicefrac{x}{|\unfilled|} <  \ceilfrac{x}{|\unfilled|}$. If some child stores at most $\floorfrac{x}{|\unfilled|} - 1$ replicas, then some other child must store at least  $\ceilfrac{x}{|\unfilled|}$ replicas. Likewise, if some child stores at least $\ceilfrac{x}{|\unfilled|} + 1$ replicas, then some other child must store at most $\floorfrac{x}{|\unfilled|}$ replicas. Both situations violate the balancing property.
		
	\end{enumerate}
	
	We now prove the main result. Note that we need to place either $\R$ or $\R+1$ replicas on $|\unfilled|$ \punfill children identified by the labeling algorithm. If $|\unfilled| = 1$, then, clearly, the \punfill child needs to store either $\R$ or $\R+1$ replicas. Thus assume that $|\unfilled| \geq 2$. Note that the capacity of each \punfill child  is at least $\ceilfrac{\R}{|\unfilled|}$.  There are three cases depending on the values of $\R \bmod |\unfilled|$ and $(\R + 1) \bmod |\unfilled|$.
	
	\begin{enumerate}[label={Case \roman*}]
		
		\item ($\R \bmod |\unfilled| = 0$):~~With $\R$ replicas, we need to store exactly $\nicefrac{\R}{|\unfilled|}$ replicas on each \punfill child. With $\R + 1$ replicas, we need to store either $\floorfrac{(\R+1)}{|\unfilled|}$ or $\ceilfrac{(\R+1)}{|\unfilled|}$ replicas on each \punfill child. But since $\R \bmod |\unfilled| = 0$, $\floorfrac{\R}{|\unfilled|} = \floorfrac{(\R+1)}{|\unfilled|}$, yielding the two values claimed.

		\item $\left((\R + 1) \bmod |\unfilled| = 0\right)$:~~With $\R$ replicas, we need to store either $\floorfrac{\R}{|\unfilled|}$ or $\ceilfrac{\R}{|\unfilled|}$ replicas on each \punfill child. With $\R + 1$ replicas, we need to store exactly $\nicefrac{(\R+1)}{|\unfilled|}$ replicas on each \punfill child. But since $(\R + 1) \bmod |\unfilled| = 0$,  $\ceilfrac{\R}{|\unfilled|} = \ceilfrac{(\R+1)}{|\unfilled|} = \nicefrac{(\R+1)}{|\unfilled|}$, yielding the two values claimed.
		
		\item ($\R \bmod |\unfilled| \neq 0$ and $(\R + 1) \bmod |\unfilled| \neq 0$):~~In this case, $\floorfrac{\R}{|\unfilled|} = \floorfrac{(\R+1)}{|\unfilled|}$ and $\ceilfrac{\R}{|\unfilled|} = \ceilfrac{(\R+1)}{|\unfilled|}$, yielding the two values claimed. 
		
	\end{enumerate}
	
	\noindent
	This completes the proof.
	\qed
\end{proof}

\subsubsection{Combine Phase}\label{s:conquer-phase}

Once the recursive call completes, we combine the results from each of the children to achieve the lexico-minimum value of the objective function overall. Let $\unfilled$ be the set of \punfill children found in the divide phase. The combine phase decides which $\nHvy$ \punfill children receive a mass of $\bigRep$, and which receive a mass of only $\smallRep$. We call a placement with size $\bigRep$ \emph{heavy}, and a placement of size $\smallRep$ \emph{light}. We must select $\nHvy$ children to receive heavy placements in such a way that the overall placement is lexico-minimum. Recall that we must return two optimal placements, one of size $\minRepOn{u}$ and another of size $\minRepOn{u} + 1$. We show how to obtain an optimal placement of size $\minRepOn{u}$, the $\minRepOn{u} + 1$ case is easily obtained thereafter.

Let the \punfill children be given as $\children{1}{|\unfilled|}$. For each \punfill child $\child{i}$ let $\a_i$ (respectively $\b_i$) represent the minimum value of $\ff{\P}$, where $\P$ is any light placement (respectively any heavy placement) of replicas on child $\child{i}$. Recall that the values of optimal heavy and light placements were recursively computed for each child, so values of $\a_i, \b_i \in \intvecs$ are readily available. We formulate an optimization problem by setting decision variables $x_i \in \{0,1\}$, for which  $x_i = 0$ if child $\child{i}$ receives a light placement, or $1$ if $\child{i}$ receives a heavy placement. The problem can then be described as an assignment of values to $x_i$ according to the following system of constraints.
\begin{equation}\label{e:optProb}
\min \sum_{i=1}^{|\unfilled|} \a_i + (\b_i - \a_i)x_i, ~~~\text{subj. to:}~~~ \sum_{i=1}^{|\unfilled|} x_i = \nHvy.
\end{equation}

An assignment of $x_i$ which satisfies the requirements of (\ref{e:optProb}) can be found by computing $\b_i - \a_i$ for all $i$, and greedily assigning $x_i = 1$ for the children with the $\nHvy$ smallest values of $\b_i - \a_i$. This solution is clearly feasible. \autoref{t:optCorrect} below states that this assignment is also optimal.

Our proof of \autoref{t:optCorrect} relies crucially on the fact that the lexicographic order on integer vectors forms a linearly-ordered Abelian group under the operation of component-wise addition. For completeness, we state the properties of a linearly-ordered Abelian group here.

\begin{definition}
	A linearly-ordered Abelian group is a triple $\logroup$, where $\G$ is a set of elements, $\gop$ is a binary operation on $\G$, and $\logeq$ is a linear (total) order on $\G$ such that all of the following properties are satisfied \cite{Levi1947}.
	\begin{enumerate}[a)]
		\item \emph{Associativity:} for all $x,y,z\in \G$, $x + (y +z) = (x+y)+z$
		\item \emph{Commutativity:} for all $x,y \in \G$, $x + y = y + x$.
		\item \emph{Identity:} there is an element $0 \in \G$ such that for all $x \in G$, $0 + x = x$.
		\item \emph{Inverses:} for all $x \in G$, there is an element $x^{-1} \in G$, such that $x + (x^{-1}) = 0$. In commutative (Abelian) groups, $x^{-1}$ is typically denoted $-x$.
		\item \emph{Translation-invariance:} for all $x,y,z \in \G$, if $x \logeq y$, then $x + z \logeq y + z$.
	\end{enumerate}
\end{definition}

It is straight-forward to show that $\groupdef{\intvecs}{+}{\lgeq}$ is a linearly-ordered Abelian group, where $+$ is component-wise addition, and $\lgeq$ is the lexicographic order. Armed with this fact, we can now formally state and prove correctness of the greedy optimization procedure.

\begin{theorem}\label{t:optCorrect}
	Let $\perm = (\perm_1, \perm_2, ..., \perm_{|\unfilled|})$ be a permutation of indices $\{1, ..., |\unfilled|\}$ such that
	$$\b_{\perm_1} - \a_{\perm_1} \loleq \b_{\perm_2} - \a_{\perm_2} \loleq ... \loleq \b_{\perm_{|\unfilled|}} - \a_{\perm_{|\unfilled|}}.$$ 
	If $\langle x_1, ..., x_{|\unfilled|} \rangle$ is defined according to the following rule: set $x_{\perm_i} = 1$ if and only if \mbox{$i \leq \nHvy$}, else $x_{\perm_i} = 0$, then $\langle x_1, ..., x_{|\unfilled|} \rangle$ is an optimal solution to (\ref{e:optProb}).
\end{theorem}

\begin{proof}
	First, notice that any optimal solution to (\ref{e:optProb}) also minimizes the quantity $\sum_{i} (\b_i - \a_i)x_i$. Therefore, it suffices to minimize this quantity. For convenience, we consider $\langle x_1, ..., x_{|\unfilled|}\rangle$ to be the characteristic vector of a subset $S$ of indices $\{1, ..., |\unfilled|\}$. We will show that no other such subset $S'$ can yield a characteristic vector $\langle x'_1, ..., x'_{|\unfilled|} \rangle$ which is strictly better than $\langle x_1, ..., x_{|\unfilled|}\rangle$ as follows.
	
	Let $\beta = \nHvy$, and let $S = \{\perm_1, ..., \perm_\beta\}$ be the first $\beta$ entries of $\perm$ taken as as set. Suppose that there is some $S'$ which represents a feasible assignment which is strictly better than that represented by $S$. Clearly, $S' \subseteq \{1,...., |\unfilled|\}$, such that $|S'| = \beta$ and $S \neq S'$. Since $S \neq S'$, and $|S'| = |S|$, we have that there must be some $i \in S \setminus S'$ and $j \in S' \setminus S$. We claim that we can improve on $S'$ by forming $S^\ast = (S' \setminus\{j\}) \cup \{i\}$. Specifically, we claim that
	\begin{equation}\label{e:non-increasing}
	\sum_{k \in S^\ast} (\b_k - \a_k) \lleq \sum_{k \in S'} (\b_k - \a_k),
	\end{equation}
	which implies that replacing a single element in $S'$ with one from $S$ does not cause the quantity minimized in (\ref{e:optProb}) to increase.
	
	To prove (\ref{e:non-increasing}), note that $j \notin S$ and $i \in S$ implies that $(\b_i - \a_i) \loleq (\b_j - \a_j)$. We now apply the translation-invariance of $\groupdef{\intvecs}{+}{\lgeq}$, which states that for any $x,y,z \in \intvecs$, $x \lleq y \implies z + x\lleq z + y$. Let $x = (\b_i - \a_i)$, $y = (\b_j - \a_j)$, and let $z = \sum_{k \in (S^\ast \setminus {i})} (\b_k - \a_k)$. This yields
	$$\sum_{k \in (S^\ast \setminus \{i\})} (\b_k - \a_k) + (\b_i - \a_i) \lleq \sum_{k \in (S^\ast \setminus \{i\})} (\b_k - \a_k) + (\b_j - \a_j).$$
	But since $S^\ast \setminus \{i\} = S' \setminus \{j\}$, we have that 
	\begin{align*}
	\sum_{k \in (S^\ast \setminus \{i\})} (\b_k - \a_k) + (\b_i - \a_i) &\lleq \sum_{k \in (S' \setminus \{j\})} (\b_k - \a_k) + (\b_j - \a_j)\\
	\sum_{k \in S^\ast} (\b_k - \a_k) &\lleq \sum_{k \in S'} (\b_k - \a_k)
	\end{align*}
	Thereby proving (\ref{e:non-increasing}). This shows that any solution which cannot be represented by $S$ can be modified to swap in an extra member of $S$ without increasing the quantity minimized by (\ref{e:optProb}). By induction, it is therefore possible to include every element of $S - S'$ until $S'$ is transformed into $S$. Therefore, $\langle x_1, ..., x_{|\unfilled|} \rangle$ is an optimal solution to (\ref{e:optProb}). \qed
\end{proof}

The required greedy solution can be quickly formed by first selecting the \punfill child having the $(\nHvy)^{th}$ largest value of $\b_i - \a_i$ using linear-time selection. Thereafter, the partition procedure from quicksort can be used to find those children having values below this selected child. For clarity of notation, we assume from here on that the \punfill children $\children{1}{|\unfilled|}$ are sorted in increasing order of $\b_i - \a_i$, even though \emph{the algorithm performs no such sorting}.

At the end of the combine phase, we compute and return the sum
\begin{equation}\label{e:combine-sum}
\sum_{i \leq \nHvy} \b_i + \sum_{i > \nHvy} \a_i + \sum_{\child{j} \in F} \ff{\Csub{\child{j}}} + \alpf{\minRepOn{u}},
\end{equation}
where $\alpf{k}$ is a vector of size $\repfact+1$ with a one in the ${k}^{th}$ entry, and zeroes everywhere else. Since node $u$ has mass $\minRepOn{u}$, the $\alpf{\minRepOn{u}}$ term accounts for node $u$'s contribution to $\ff{\P}$. Thus equation (\ref{e:combine-sum}) gives the value of an optimal placement of $\minRepOn{u}$ replicas placed on node $u$.

\begin{algorithm}[t]
	\caption{An $O(n \repfact)$ algorithm for optimal single-block placement.}\label{a:placeReplicas}
	\SetKwFunction{placeReplicas}{Place-Replicas}
	\SetKwFunction{getFilled}{Label-Children}
	\SetKwFunction{partition}{Partition}
	\SetKwFunction{fill}{Filled-Value}
	\SetKwProg{Fn}{Function}{begin}{end}
	let $\partition(S, k)$ partition $S$ into sets $L, H$, where $L$ contains the $k$ smallest elements, and $H$ contains the remaining $|S| - k$ elements\;
	\Fn{\placeReplicas{$u, r$}}{
		let $\children{1}{\nChildren}$ be children of $u$ \tcp*{Divide phase}
		$U, F \gets$ \getFilled{$\{\children{1}{\nChildren}\}, r$} \tcp*{$O(t)$ time}
		$R \gets \repfact - \sum_{\child{i} \in F} |\Csub{\child{i}}|$\;
		
		\For(\tcp*[f]{Combine phase}){$\child{i} \in U$} { 
			$\a_i \gets $\placeReplicas{$\child{i}, \smallRep$}\label{l:recurseBegin}\;
			$\b_i \gets $\placeReplicas($\child{i}, \bigRep$)\label{l:recurseEnd}\;
		}
		$L, H \gets$ \partition{$\{\b_1-\a_1,...,\b_{|U|} - \a_{|U|}\}, \nHvy$} \tcp*{$O(|U|\repfact)$ time }
		\Return $\sum_{\child{i} \in L} \b_i + \sum_{\child{i}\in H} \a_i + \sum_{\child{j}\in F} \ff{\Csub{\child{j}}} + \alpf{\repfact}$ \tcp*{$O(t\repfact)$ time}
	}
	
\end{algorithm}

Pseudocode for the entire algorithm appears in \autoref{a:placeReplicas}. Implementing this procedure directly yields an $O(\n\repfact)$ time algorithm, where $n$ is the number of nodes in the tree, and $\repfact$ is the number of replicas to be placed. In the next section, we describe several improvements which are used to achieve a running time of $O(n + \repfact \log \repfact)$.

\subsubsection{Transforming to Achieve $O(\n + \repfact \log \repfact)$ Time}\label{s:rhologrho}

First, observe that the maximum failure number returned from child $\child{i}$ is $\minRepOn{\child{i}} + 1$. This, along with the property that every node's failure number is greater than or equal to that of its descendants, implies that the vector returned from $\child{i}$ will have a zero in indices $0,...,\repfact - \minRepOn{\child{i}} - 2$. To avoid wasting time, we modify the algorithm to return only the non-zero suffix of this vector, which has length at most $O(\minRepOn{\child{i}})$. At each node, we can compute (\ref{e:combine-sum}) by summing the entries of the vector in decreasing order of their index, and skipping zero entries. Specifically, to compute $\vec{v}_1 + ... + \vec{v}_\nChildren$, we first allocate an empty vector $\vec{w}$ of size $\minRepOn{u} + 1$, to store the result of the sum. For each vector $\vec{v}_i$, we set $\vec{w}[j] \gets \vec{w}[j] + \vec{v}_i[j]$ for indices $j$ from indices $\repfact - \minRepOn{\child{i}}$ up to $\repfact$. After all the vectors have been processed, $\vec{w} = \vec{v}_1 + ... + \vec{v}_\nChildren$. This algorithm takes only $O(\minRepOn{1}) + ... + O(\minRepOn{t}) = O(\minRepOn{u})$ time to compute a single sum. Using smaller vectors also implies that the $(\nHvy)^{th}$ best child can be found in $O(\minRepOn{u})$ time, since each \punfill child returns a vector of size at most $O(\minRepOn{u}/|\unfilled|)$, each comparison need take no more than $O(\minRepOn{u}/|\unfilled|)$ time. Since there are only $|\unfilled|$ children to compare, we obtain $O(\minRepOn{u})$ time in total for linear-time selection when using these sparse vectors. With these modifications, the entire combine phase takes only $O(\minRepOn{u})$ time at every node $u$. We will collectively refer to the techniques presented in this paragraph as \emph{prefix truncation} in later sections.

To bring down the running time of the combine phase, note that in any placement, nodes at the same depth have $\repfact$ replicas placed on them in total. We can therefore achieve an $O(\repfact \log \repfact)$ time combine phase overall by ensuring that the combine phase only needs to occur in at most $O(\log \repfact)$ levels of the tree. To do this, observe that when $\minRepOn{u} = 0$, any leaf with minimum depth forms an optimal placement of size $1$. Moreover, we can easily construct pointers from each node to its minimum depth leaf during an $O(n)$ time preprocessing phase. Therefore, the combine phase does not need to be executed once $\minRepOn{u} = 0$. To ensure that there are only $O(\log \repfact)$ levels, we transform the tree to guarantee that as the combine phase proceeds down the tree, $\minRepOn{u}$ decreases by at least a factor of two at each level. The balancing property ensures that this will automatically occur when there are two or more \punfill children at each node. However this is not guaranteed when the tree contains what we term a \emph{degenerate chain}, a path of nodes each of which only have a single \punfill child. By removing degenerate chains, we can obtain an $O(\repfact \log \repfact)$ combine phase.

\begin{figure}[tb]
	\subcaptionbox{A degenerate chain.\label{f:degenerate-unfilled-case}}[0.6\textwidth]
	{
		\scalebox{1.0}{\begin{tikzpicture}[scale=0.525]

			\node[inner sep=0]
			{\includegraphics[scale=0.45]{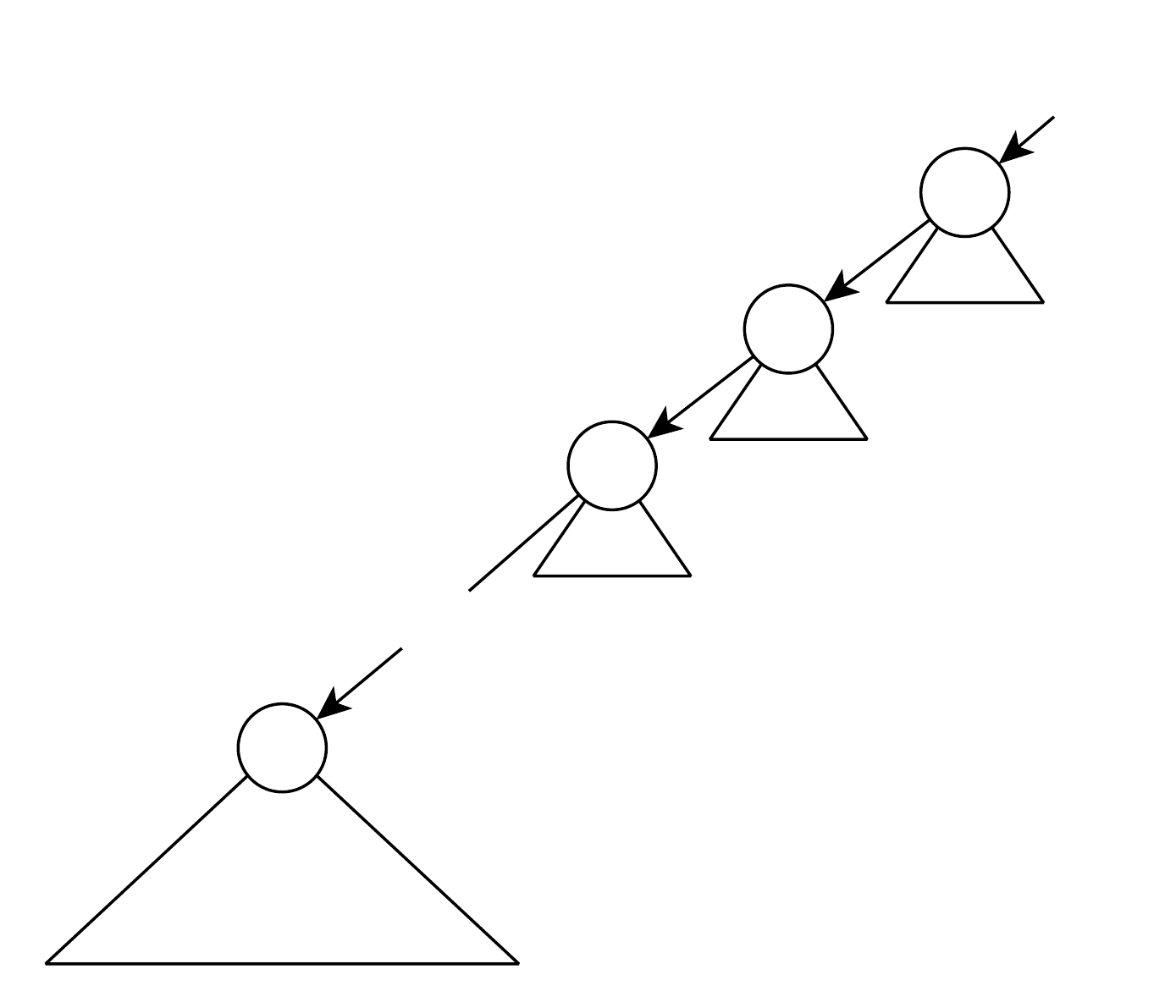}};

			\draw (-1.75, -1.1) node[circle, fill=white] {$\iddots$};
			\draw[decorate, decoration={brace, amplitude=5pt}] (-4, -2) -- (4.0, 4.3) node [midway, xshift=-8, yshift=8, rotate=37] {length $O(\repfact)$};
			\draw (-3.15, -4) node {\small $T_k$};
			\draw (-3.15, -2.55) node {\small $v_k$};
			\draw (0.35, -0.45) node {\small $T_3$};
			\draw (0.3, 0.4) node {\small $v_3$};
			\draw (2.1, 1.75) node {\small $v_2$};
			\draw (2.1, 1) node {\small $T_2$};
			\draw (4, 2.4) node {\small $T_1$};
			\draw (3.9, 3.1) node {\small $v_1$};
			\draw(4, 1.8) node {\tiny $O(1)$ leaves};
			\draw(2.15, 0.4) node {\tiny $O(1)$ leaves};
			\draw(0.2, -1.1) node {\tiny $O(1)$ leaves};
			
			\end{tikzpicture}
			
		}
	}
	\subcaptionbox{Contracted pseudonode.\label{f:contracted-nodes}}[0.35\textwidth]
	{
		\scalebox{1.0}{\begin{tikzpicture}[scale=0.525]
			\centering
			\node[inner sep=0]
			{\includegraphics[scale=0.5]{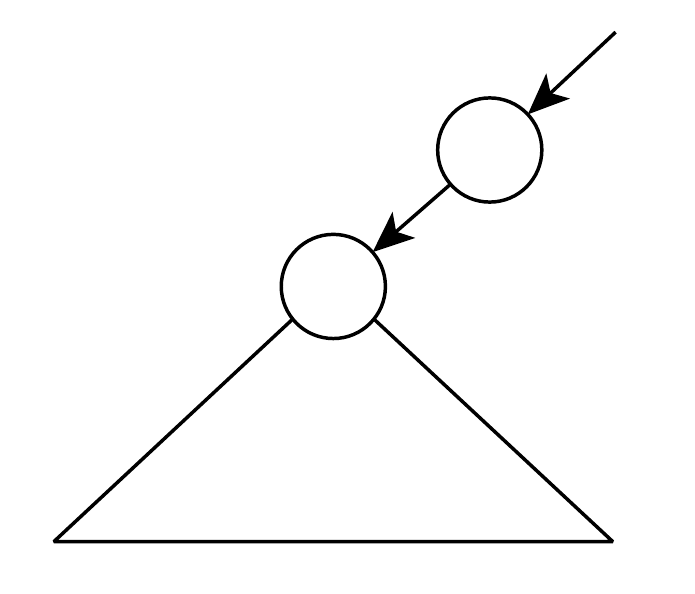}};

			\draw (-0.125, -1.5) node[align=center] {$T_k$};
			\draw (-0.125, 0) node[align=center] {$v_k$};
			
			\draw (1.35, 1.35) node[align=center] {$w$};
			

			\end{tikzpicture}
			
		}
	}
	\caption{Illustration of a degenerate chain.}
\end{figure}

\autoref{f:degenerate-unfilled-case} illustrates a degenerate chain of length $k$. In this figure, for all $i \in \{1,...,k-1\}$, all nodes in $T_i$ are \dfill. Moreover, node $v_k$ has at least two \punfill children. Thus, for $i\in \{1,...,k-1\}$, each $v_i$ has only a single \punfill child, namely $v_{i+1}$. It is easy to see that if the number of leaves in each set $T_i$ is constant with respect to $\repfact$, then the length of the chain can be as large as $O(\repfact)$. This would imply that there can be $O(\repfact)$ levels in the tree where the entire combine phase is required. To remove degenerate chains, we contract nodes $v_1, ..., v_{k-1}$ into a single pseudonode, $w$, as in \autoref{f:contracted-nodes}. However, we must take care to ensure that the values returned by pseudonode $w$ take into account contributions from the entire contracted structure. We will continue to use $v_i$ and $T_i$ throughout this section to refer to nodes in a degenerate chain. The remainder of this section treats the removal of degenerate chains at a high level. Interested readers can find detailed pseudocode in \autoref{app:transform-pseudo}.

Let $\vec{a_w}, \vec{b_w}$ be the values returned by pseudonode $w$. In order for the transformation to be correct, we need to ensure that these values are the same as those which would have been returned had no transformation been performed. To ensure this, we must consider and include the contribution of each node in the set $S_w = T_1 \cup ... \cup T_{k-1} \cup \{v_1, ..., v_k\}$. It is easy to see that the failure numbers of $v_1, ..., v_{k-1}$ depend only upon whether $\minRepOn{v_k}$ or $\minRepOn{v_k} + 1$ replicas are placed on node $v_k$, while the \punfill nodes in sets $T_1,..., T_{k-1}$ have no such dependency. Since the value of $\minRepOn{v_k}$ is only available at the end of the divide phase, we detect and contract degenerate chains immediately afterwards.

The transform phase runs as a breadth first search. Detecting degenerate chains is relatively straightforward, but careful memory management is required to keep the running time of this phase below $O(\n + \repfact \log \repfact)$. When contracting a degenerate chain, we must sum each \dfill  node's contribution to the vectors $\vec{a_w}$ and $\vec{b_w}$. Intermediate values of this sum must be stored in an array of size $O(\minRepOn{v_1})$ as we contract the chain. The key to achieving an $O(\n + \repfact \log \repfact)$ transform phase lies in allocating space for this array only \emph{once} for each chain. Taking care in this way allows us to bring the running time for contracting a degenerate chain of size $|S_w|$ down to $O(|S_w| + \minRepOn{v_1})$.

When we sum this expression over all degenerate chains, we obtain a running time of $O(n + \repfact \log \repfact)$ for the transform phase. To reach this result, examine the sum over values of $\minRepOn{v_1}$ for all \emph{pseudonodes} having the same depth. Since there are at most $\repfact$ replicas among such pseudonodes, this sum can be at most $O(\repfact)$ at any depth. After the degenerate chains have been contracted, there are only $O(\log \repfact)$ levels where $\minRepOn{u} > 1$. Thus, pseudonodes can be only be present in the first $O(\log \repfact)$ levels of the final tree. Therefore the $O(\minRepOn{v_1})$ term sums to $O(\repfact \log \repfact)$ overall. Since the $O(|S_w|)$ term clearly sums to $O(n)$ overall, the transform phase takes at most $O(n + \repfact \log \repfact)$ time.

Including the transform phase implies that there are only $O(\log \repfact)$ levels where the combine phase needs to be run in its entirety. Therefore, the combine phase takes $O(\repfact \log \repfact)$ time overall. When combined with the $O(n)$ divide phase and the $O(n + \repfact \log \repfact$) transform phase, this yields an $O(n + \repfact \log \repfact)$ algorithm for solving replica placement in a tree.

\section{Solving Multi-block Replica Placement}\label{s:mp-balancing}

In this section, we describe an exact algorithm to optimize the simultaneous placement of multiple blocks of replicas at once (i.e. \autoref{p:graph-multi-placement}). As previously discussed, this problem naturally occurs in data centers, where multiple sets of replicas co-exist. Recall that in \autoref{model} we defined a \emph{multi-placement} $\MP$ as an ordered set of $\m$ placements, $\MP \defined (\P_1, ..., \P_m)$. 
Recall also that in the multi-placement case, each leaf node has a capacity of $c(\ell)$, the maximum number of replicas which it can accommodate. Moreover, no leaf node will accept more than one replica from any given single placement (or block), as this would in some sense defeat the purpose of replication.
We further extended the failure aggregate to multi-placements by defining the failure aggregate of a multi-placement as the sum of the failure aggregates of the individual placements (padded as needed).

In the single-placement case, it is easy to see how the failure aggregate of a placement is comprised of local contributions from each node. Each node contributes a factor to the objective based upon its failure number. In the multi-placement case, each node's local contribution to the failure aggregate is not as clear. To clearly state the contribution of each node to the failure aggregate of a multi-placement, we will introduce the concept of the \emph{\signature} of a multi-placement, and the concept of a \emph{\subplacement} of a multi-placement.

Each multi-placement has an associated \emph{\signature}, a vector which summarizes the number of replicas assigned to each block. The $k^{th}$ component of the signature counts the number of blocks which have $k$ replicas assigned to them. Formally, the signature of multi-placement $\MP$ is a vector $\vec{\sigma}(\MP) = \langle n_0, ..., n_\repfact\rangle$, where $\repfact $ is the size of the largest placement, and $n_k$ is the number of placements in $\MP$ which have size $\repfact-k$, (i.e. \sloppy\mbox{$n_k = \left| \{i : [1,...,m] : |P_i| = \repfact-k\} \right|$}). 

As an example application of this concept, recall that the input to an instance of
\autoref{p:graph-multi-placement} specifies a series of $m$ integers, $\repfact_1, ..., \repfact_m$, where $\repfact_i$ is the number of replicas placed on block $i$ in a valid solution. It is easy to see that these integers uniquely specify the \signature which a valid solution is allowed to have. For example, if $m=5$ and \mbox{$(\repfact_1, \repfact_2, \repfact_3, \repfact_4, \repfact_5) = (1,2,2,3,3)$} in an instance of \autoref{p:graph-multi-placement}, then only multi-placements with \signature $\vec{\sigma}(\MP) = \langle 2, 2, 1, 0\rangle$ are valid solutions. Observe also that \emph{any} multi-placement with a \signature of $\langle 2,2,1,0 \rangle$ can be made to satisfy the requirement that $(\repfact_1, \repfact_2, \repfact_3, \repfact_4, \repfact_5) = (1,2,2,3,3)$ simply by relabeling the blocks appropriately.

As a second example, observe that the signature of a multi-placement summarizes the failure numbers associated with the root node. For instance, if $\vec{\sigma}(\MP) = \langle 2, 2, 1,0\rangle$, then there is one block of $\MP$ for which the root has failure number $1$, two blocks for which the root has failure number $2$, and two blocks for which the root has failure number $3$.

This last observation hints at the importance of signatures which we wish to convey. The signature of a multi-placement accumulates
the failure numbers of the root node across all $\m$ blocks of the multi-placement into an alternate vector form. This vector form has certain advantages since its number of non-zero entries depends on the skew in the desired multi-placement.
Using the concept of a \emph{\subplacement}, we can similarly collect the failure numbers of \emph{any node in the tree}. For any given multi-placement the \emph{\subplacement} at node $u$ consists only of the replicas of the multi-placement which are assigned to leaves of the subtree rooted at node $u$. More formally, if $\Csub{u}$ is the set of leaves assigned to node $u$, and $\MP= (P_1, P_2, ...,P_m)$, then $\MP_u = (P_1 \cap \Csub{u}, P_2\cap \Csub{u}, ..., P_\m \cap \Csub{u})$. To illustrate this concept, refer to the tree in \autoref{f:sub-placement-ex}. The multi-placement $\MP = (\{a,b,c\}, \{b,d,e\}, \{b,c,e\})$ has a \subplacement at node $u$ given by $\MP_u = (\{a,b,c\}, \{b\}, \{b\})$. Notice that the signature of $\MP_u$ is given by \mbox{$\vec{\sigma}(\MP_u) = \langle 1, 0, 2, 0 \rangle$}, and furthermore, this collects the failure numbers of $u$ \emph{with respect to the original multi-placement $\MP$} into a convenient form which matches that of the failure aggregate. Specifically, node $u$ has failure number 3 with respect to one block of $\MP$ (block 1), and a failure number of 1 with respect to two blocks of $\MP$ (blocks 2 and 3).

\begin{figure}
	\centering
\begin{tikzpicture}[parent anchor=center,child anchor=center,every node/.style={draw, minimum height=6mm}]
\node[circle] (root) {}
[level distance=12.5mm,sibling distance=15mm]
child {node[circle] (u) {$u$} 
	[sibling distance=7mm]
	child{node[circle] {} 
		[sibling distance=7.5mm]
		child{node[rectangle] {$a$} }
		child{node[rectangle] {$b$} }
	}
	child{node[rectangle] {$c$} }
}
child {node[circle] (anchor) {$w$}
	child{node[circle] {}
		[sibling distance=7.5mm] 
		child{node[rectangle] {$d$} }
		child{node[rectangle] {$e$} }
	}
};
\node[draw=none, right=1cm of anchor,yshift=-5mm]{
	\normalsize
	$\begin{array}{lllll}
	\MP &= (~\{a,b,c\},& \{b,d,e\}, &\{b,d\}&)\vspace{2mm}\\
	\MP_u &= (~\{a,b,d\},& \{b\}, &\{b\}  &)\vspace{2mm}\\
	\MP_w &= (~\{\}, &\{d,e\}, &\{e\} &)
	\end{array}$};
\end{tikzpicture}
	\caption{Multi-placement $\MP$ and two of its \subplacement{s}, $\MP_u$ and $\MP_w$.} \label{f:sub-placement-ex}
\end{figure}

This suggests a way to rewrite the failure aggregate of a multi-placement into a more convenient form in which each node's individual contribution to the overall failure aggregate is made transparent. Specifically, the failure aggregate of a multi-placement is just the sum of the signatures of all of its \subplacement{s}, as formalized in the following lemma.

\begin{lemma}\label{lem:signature-solution}
	For any $u \in V$, let $\MP_u$ be the \subplacement of $\MP$ at node $u$, then
	\[\g{\MP} = \sum_{u \in V} \sig{\MP_u}. \]
\end{lemma}
\begin{proof}
	Let $\sig{\MP_u} = \langle n^u_0, ...,n^u_\repfact \rangle$, then $n^u_j$ counts the number of placements of $\MP$ in which node $u$ has failure number $\repfact-j$. Furthermore, let $\ff{\P_i} = \langle p^i_0, ...,p^i_\repfact \rangle$, then $p^i_j$ counts the number of nodes which have failure number $\repfact-j$ with respect to placement $P_i$. We first show that \[\sum_{u\in V} n_j^u = \sum_{i=1}^m p^i_j\] for any failure number $\repfact-j$. This correspondence is easy to see as follows. Each node $u$ contributes a factor of 1 for each placement in which $u$ has failure number $\repfact-j$ on \emph{both} sides of the correspondence. On the RHS, node $u$ is counted as one of the nodes which has failure number $\repfact-j$ with respect to block $i$, so $u$ contributes a factor of 1 to the term $p^i_j$. On the LHS, node $u$ contributes a factor of 1 to the term $n_j^u$, since block $i$ is one of the blocks with respect to which $u$ has failure number of $\repfact-j$. Clearly, for every factor of 1 contributed on the LHS another factor of 1 must be contributed on the RHS, thus the two sums are equal. Hence,
	\begin{multline*}\sum_{u \in V}\sig{\mathcal{P}_u} = \left\langle \sum_{u\in V} n_0^u, \sum_{u\in V} n_1^u, ..., \sum_{u \in V} n_\repfact^u \right\rangle =
	\\ \left\langle \sum_{i=1}^m p^i_0, \sum_{i=1}^m p^i_1, ..., \sum_{i=1}^m p^i_\repfact \right\rangle =  \sum_{i = 1}^m \ff{P_i} = \g{\mathcal{P}},
	\end{multline*}\qed
\end{proof}

Each node's contribution to the failure aggregate of a multi-placement is thus clear. Node $u$ contributes the signature of its \subplacement, $\sig{\MP_u}$, to the overall value of $\g{\MP}$. We can therefore optimize $\g{\MP}$ by locally optimizing values of $\sig{\MP_u}$. We can do so directly via a dynamic program as follows. At each node $u$ we compute and store $G_u(\vec{\sigma})$, the optimal value which can be attained by any \subplacement at node $u$ which has a signature of $\vec{\sigma}$. We show how a table for $G_u(\vec{\sigma})$ can be recursively computed in \autoref{s:mp-dp}. Since there are roughly $O(m^{\repfact+1})$ possible signatures for which the value of $\G_u$ must be computed, it is not immediately clear that such an approach will be tractable. As we shall see, if we are given a signature and we want to find an optimal multi-placement which has that signature, we can achieve a significant reduction in running time.

First, in the absence of an associated multi-placement, a signature is just a vector $\vec{\sigma} = \langle\sigma_0, ..., \sigma_\repfact \rangle \in \mathbb{N}^{\repfact + 1}$. 
Recall that we defined the \emph{skew} of a multi-placement as the absolute difference between the maximum and minimum failure numbers of each block. This definition extends readily to signatures. Specifically, we define the skew of a signature to be the difference between the indices of its maximum and minimum non-zero entries, formally
\[skew(\langle\sigma_0, ..., \sigma_\repfact \rangle) = \max_{\sigma_i \neq 0} i - \min_{\sigma_j \neq 0} j.\] If the signature consists only of zeroes, we define $skew(\langle 0, ..., 0\rangle) = 0.$ Obviously, the skew of a multi-placement and the skew of its signature are equivalent. Likewise, we define the \emph{girth} of a signature as the maximum index $i$ for which $\sigma_i$ is non-zero.

Our key insight is that \emph{in order to find an optimal multi-placement which has skew $\delta$ we only need to compute values of $G_u$ for signatures which have skew $\delta$.} As there are roughly $O(m^\delta)$ such signatures, we can obtain an algorithm which works well when the skew of the desired multi-placement is small. Specifically, we provide an exact algorithm which runs in polynomial time for fixed values of $\delta$. Since $\delta$ is typically small in practice, this comprises a significant speed-up over the brute-force approach.

The remainder of this section is organized as follows. 
In \autoref{s:mp-theory} we prove that it suffices to recursively consider signatures with skew no greater than $\delta$. 
In \autoref{s:mp-algorithm} we present an exact algorithm for finding an optimal multi-placement based on the above property.

\subsection{The Bounded Skew Property for Optimal Multi-placements}\label{s:mp-theory}

We begin by defining an \emph{exchange} of a multi-placement, which is simply a multi-placement that can be formed by rearranging the assignment of replicas among two blocks.
\begin{definition}
	An \emph{exchange} of a multi-placement $\MP = (P_1, ..., P_m)$ is a multi-placement $\mathcal{Q} = (Q_1,...,Q_m)$ in which
	\begin{enumerate}[a)]
		\item $\bigcup_i P_i = \bigcup_i Q_i$
		\item $|P_i| = |Q_i|$ for all $i$,
		\item there exist indices $i,j$ such that for every $k$ different from $i$ and $j$, $P_k = Q_k$.
	\end{enumerate}
	Blocks $i$ and $j$ are referred to as the \emph{targets} of the exchange, and we say that the exchange \emph{targets blocks $i$ and $j$}.
\end{definition}
Notice that part (b) implies that the signature of an exchange matches that of the original multi-placement.

We define a \emph{localized exchange} as an exchange which involves only \subplacement{s} of sibling nodes $u$ and $v$.
\begin{definition}
	A \emph{localized exchange} of a multi-placement $\MP$ is an exchange, say $\mathcal{Q}$ targeting blocks $i$ and $j$, for which there exist sibling nodes $u$ and $v$ such that for all choices of a node $w$ which is not a an ancestor or descendant of $u$ or $v$, the \subplacement{s} of $\MP$ and $\mathcal{Q}$ at node $w$ are the same, (i.e. $\MP_w = \mathcal{Q}_w$). Moreover, we refer to nodes $u$ and $v$ as the \emph{carriers} of the localized exchange.
\end{definition}
Refer to \autoref{f:implocex-example} for an example of a localized exchange and a few non-examples.

\begin{figure}
	\centering
\begin{tikzpicture}[parent anchor=center,child anchor=center,every node/.style={draw, minimum height=6mm}]
\node[circle] (root) {}
[level distance=15mm,sibling distance=25mm]
child { node[circle] (a) {$a$}
	[sibling distance=7.5mm]
	child{ node[rectangle] {$1$} }
	child{ node[rectangle] {$2$} }		
}
child {node[circle] (u) {} 
	[sibling distance=20mm]
	child {node[circle] (b) {$b$}
		[sibling distance=7mm]
		child{node[circle] {}
			[sibling distance=7.5mm]
			child{ node[rectangle]{$3$}}
			child{ node[rectangle]{$4$}}
		}
		child{node[rectangle]{$5$}}
		child{node[rectangle]{$6$}}
	}
	child {node[circle] (c) {$c$}
		[sibling distance=7.5mm]
		child{node[circle] {}
			[sibling distance=7.5mm]
			child{ node[rectangle] {$7$}}
			child{node[rectangle] {$8$}	}
		}
		child{node[rectangle] {$9$}}
	}
};
\node[draw=none, right=1cm of c]{
	$\begin{array}{ll}
	\MP &= \left(\{1,4,6\}, \{2,3,6,8\}\right)\\
	\mathcal{Q} &= \left(\{1,8,6\}, \{2,3,6,4\}\right)\vspace{1mm}\\
	\multicolumn{2}{l}{\tworowcell{\text{$\mathcal{Q}$ is a localized exchange of $\MP$ with carriers} \\\text{$b$ and $c$ which targets blocks $1$ and $2$.}}}\vspace{1mm}\\
	\MP &= \left(\{1,4,6\}, \{2,3,6,8\}\right)\\
	\mathcal{Q}' &= \left(\{8,4,6\}, \{2,3,6,1\}  \right)\vspace{1mm}\\
	\multicolumn{2}{l}{\tworowcell{\text{$\mathcal{Q}'$ is an exchange of $\MP$, but is not localized,} \\\text{since nodes $a$ and $c$ are not siblings.}}}\vspace{1mm}\\
	\mathcal{Q}'' &= \left(\{1,4\}, \{2,3,7,8\}  \right)\vspace{1mm}\\
	\multicolumn{2}{l}{\tworowcell{\text{$\mathcal{Q}''$ is not an exchange of $\MP$, since both} \\\text{properties (a) and (b) are violated.}}}\\
	\end{array}$};
\end{tikzpicture}
	\caption{An example and non-example of a localized exchange ($\mathcal{Q}$ and $\mathcal{Q}'$ respectively), and a non-example of an exchange ($\mathcal{Q}''$).}\label{f:implocex-example}
\end{figure}

Finally, we define an \emph{improving localized exchange} as a localized exchange $\mathcal{Q}$ in which the objective function improves over that of $\MP$. Because of the constraints imposed on a localized exchange, we only need to check four failure numbers to determine that the objective function has improved. We build this check into the definition below.
\begin{definition}
	An \emph{improving localized exchange} of $\MP$ is a localized exchange $\mathcal{Q}$ with targets $i$ and $j$ and carriers $u$ and $v$ in which 
	\begin{align}\label{e:implocex}\max \left[\fnum{u}{\P_i}, \fnum{u}{\P_j} \right] &>  \max \left[ \fnum{u}{Q_i}, \fnum{u}{Q_j}\right],
	\intertext{and either}
	\label{e:implocex2}
	\max \left[\fnum{v}{\P_i}, \fnum{v}{\P_j} \right] &> \max \left[ \fnum{v}{Q_i}, \fnum{v}{Q_j}\right],
	\intertext{or, in the \subplacement{s} of $\MP$ and $\mathcal{Q}$ at node $v$, denoted by \mbox{$\MP_v = (\P^v_1,...,P^v_m)$}, and $\mathcal{Q}_v = (Q^v_1,...,Q^v_m)$ respectively,}
	\label{e:implocex3}
	P^v_i =  Q^v_j &\text {  and  } P^v_j = Q^v_i.
	\end{align}
\end{definition}
Note that \eqref{e:implocex} holds in every improving localized exchange, while only one of \eqref{e:implocex2} or \eqref{e:implocex3} is required to hold. To justify calling these exchanges \emph{improving}, we must prove that if an improving localized exchange exists, it constitutes a strictly better placement, which claim we state as the following lemma.

\begin{lemma}\label{lem:implocex}
	If a multi-placement $\MP$ has an improving localized exchange $\mathcal{Q}$, then $\g{\MP} \lgt \g{\mathcal{Q}}$.
\end{lemma}
\begin{proof}
	Let $\MP$ be a multi-placement with improving localized exchange $\mathcal{Q}$ with targets $i$ and $j$ and carriers $u$ and $v$. Clearly, the only nodes whose failure numbers could be different in $\mathcal{Q}$ and $\MP$ are nodes which are either ancestors or descendants of node $u$ or node $v$ (for concision, we consider nodes $u$ and $v$ to be descendants of themselves). Thus if $\g{\MP} \lgt \g{\mathcal{Q}}$, as claimed, it must be due to a difference in failure number(s) which occurs among these nodes.
	
	Without loss of generality, let $a = \fnum{u}{\P_i} = \max[ \fnum{u}{\P_i}, \fnum{u}{\P_j} ] = \max [ \fnum{u}{\P_i}, \fnum{u}{\P_j}, \fnum{v}{\P_i}, \fnum{v}{\P_j}].$
	
	First, we dispense with the descendants of nodes $u$ and $v$.
	
	Suppose that \eqref{e:implocex3} holds, then it is easy to see that $\g{\MP_v}$ and $\g{\mathcal{Q}_v}$ are equivalent, since
	\[\g{\MP_v} = \sum_{i=1}^m \ff{P_i^v}\stackrel{(*)}{=} \sum_{i=1}^m \ff{Q_i^v} = \g{\mathcal{Q}_v} \]
	where $(*)$ holds by \eqref{e:implocex3} and property (c) of an exchange. Thus, in this case, the overall contribution of descendants of $v$ to the objective value does not change, so they can be \emph{disregarded}. Moreover, since the failure number of every node is upper-bounded by the failure number(s) of its ancestor(s), for any node $w$ which is a descendant of $u$, we have that $\fnum{w}{Q_i} \leq a-1$ and $\fnum{w}{Q_j} \leq a-1$. Thus after making the exchange, the failure numbers of descendants of $u$ (with respect to blocks $i$ and $j$) can be at most $a-1$, and the descendants of $v$ can be disregarded.
	
	Suppose instead that \eqref{e:implocex2} holds. Then by combining \eqref{e:implocex} and \eqref{e:implocex2} we obtain, 
	\begin{multline*}
	\max \left[\fnum{u}{Q_i}, \fnum{u}{Q_j}, \fnum{v}{Q_i}, \fnum{v}{Q_j}  \right]  < \\
	\max \left[\fnum{u}{P_i}, \fnum{u}{P_j}, \fnum{v}{P_i}, \fnum{v}{P_j}  \right] = a.
	\end{multline*}
	Thus, for any node $w$ which is a descendant of $u$ or $v$, $\fnum{w}{Q_i} \leq a-1$, and $\fnum{w}{Q_j} \leq a-1$. And so after making the exchange, the failure numbers of descendants of $u$ and $v$ w.r.t. blocks $i$ and $j$ can be at most $a-1$.
	
	To summarize, each descendant of nodes $u$ and $v$ either has failure number at most $a-1$ with respect to blocks $i$ and $j$, or can be disregarded as a node at which no improvement can occur.
	
	Next, we consider the ancestors of $u$ and $v$. Recall that $u$ and $v$ are siblings, so let $x$ be the parent of $u$ and $v$. We show that the failure number of $x$ with respect to blocks $i$ and $j$ does not change as a result of the exchange, since
	\begin{align*}
	f(x,Q_i) &= f(u,Q_i) + f(v,Q_i) + \sum_{\substack{c \in children(x)\\ u \neq c \neq v}} \fnum{c}{Q_i}\\
	&= |Q^u_i \cup Q^v_i |  + \sum_c \fnum{c}{Q_i} ~~~ \textit{(by definition of failure number)}\\
	&= |P^u_i \cup P^v_i |  + \sum_c \fnum{c}{Q_i} ~~~ \textit{(by properties (b)-(c) of an exchange)}\\
	&= |P^u_i \cup P^v_i |  + \sum_c \fnum{c}{P_i} ~~~ \textit{(since $\mathcal{Q}$ is localized w}/ \textit{carriers $u$ and $v$)}\\       
	&= f(u,P_i) + f(u, P_i) + \sum_c \fnum{c}{P_i} = f(x,P_i),
	\end{align*}
	and the sum over $c$ is everywhere understood to be constrained as in the first line of the above derivation.
	We can similarly show that $f(x,Q_j) = f(x,P_j)$ replacing $i$ by $j$ in the above.
	Since the failure number of $x$ does not change as a result of the exchange, it is easy to see that the failure numbers of all ancestors of $u$ and $v$ likewise do not change. So, for all  $k \leq \repfact - a$, $\g{\MP}_k = \g{\mathcal{Q}}_k$. Moreover, by \eqref{e:implocex}, it is easy to see that $\g{\MP}_{\repfact - a} > \g{\mathcal{Q}}_{\repfact - a}$, since at least the failure number of $u$ strictly decreases as a result of the exchange, and $f(u,P_i) = a$. Thus $\g{\MP} \lgt \g{\mathcal{Q}}$. \qed
\end{proof}

This lemma comes in handy when proving theorems about optimal multi-placements by allowing us to construct an improving localized exchange to demonstrate the existence of a better placement. We apply it to prove the following theorems about optimal multi-placements.

\begin{theorem}\label{t:delta-skew}
	In every optimal multi-placement $\MP = \langle P_1, ..., P_m \rangle$ with skew at most $\delta > 0$, for every node $u$, the \subplacement $\MP_{u}$ has skew at most $\delta$.
\end{theorem}
\begin{proof}
	Suppose that in an optimal multi-placement with skew $\delta$, some node has a \subplacement with skew strictly greater than $\delta$. Let $u$ be a least depth such node. Node $u$ cannot be the root, because if it were then $\MP$ would have a skew greater than $\delta$, contradicting the assumption made in the statement of the theorem. Therefore node $u$ has a parent, which we will denote by $w$. Because $u$ was chosen to be least depth, \subplacement $\MP_w$ must have a skew of at most $\delta$.
	
	In order for \subplacement $\MP_u$ to have skew strictly greater than $\delta$, there must exist blocks $i$ and $j$ of $\MP$ for which $\fnum{u}{P_i} = a$, and \mbox{$\fnum{u}{P_j} \leq a - \delta - 1$}. Because $\MP_w$ has skew $\delta$, the skew must be corrected by some sibling of $u$, denoted by $v$, for which $\fnum{v}{P_i} \leq b-1$ and $\fnum{v}{P_j} = b$. If no such sibling exists then node $w$ must have a skew of at least $\delta+ 1$, contradicting that $u$ was chosen to have least depth. We will proceed in two cases. In each case, we will construct an improving localized exchange $\mathcal{Q} = \langle Q_1, ..., Q_m\rangle$ with target blocks $i$ and $j$ and carriers $u$ and $v$.
	
	\begin{enumerate}[{Case} 1)]
		\item If $\fnum{v}{P_i} = b-1$ then form $Q_i$ and $Q_j$ by swapping all of node $v$'s replicas which are in block $i$ with all of node $v$'s replicas which are in block $j$. 
		Clearly, this swap satisfies \eqref{e:implocex3}.
		After making this swap, the signature of $\mathcal{Q}$ will no longer match that of $\MP$, which is a violation of property (b) of an exchange. To maintain property (b), we take one replica from block $i$ which is placed on child $u$ and give it to block $j$. Below we summarize the result of this exchange.
		\begin{align*}
		&\MP : && \fnum{u}{P_i} = a && \fnum{u}{P_j} \leq a - \delta - 1 \\
		& \mathcal{Q} : && \fnum{u}{Q_i} = a - 1 && \fnum{u}{Q_j} \leq a - \delta
		\end{align*}
		
		Clearly, the maximum failure number in the top row is strictly greater than the maximum failure number in the bottom row, and thus $\mathcal{Q}$ is an improving localized exchange of $\MP$ which satisfies equations \eqref{e:implocex} and \eqref{e:implocex3}. Thus, by \autoref{lem:implocex}, $\MP$ is not an optimal multi-placement, a contradiction.
		
		\item If instead $\fnum{v}{P_i} < b-1$, then we must form $\mathcal{Q}$ differently. In this scenario, $\fnum{v}{P_i} \leq b-2$. We swap one of node $v$'s replicas from block $j$ to block $i$. As a result, $\fnum{v}{Q_i} \leq b-1$ and $\fnum{v}{P_j} = b-1$. To maintain property (b), we swap one of node $u$'s replicas from block $i$ to block $j$, exactly as in the previous case. Below we summarize the result of this exchange.
		\begin{align*}
		&\MP : && \fnum{u}{P_i} = a && \fnum{u}{P_j} \leq a - \delta - 1 && \fnum{v}{P_i} \leq b - 2 && \fnum{v}{P_j} = b \\
		& \mathcal{Q} : && \fnum{u}{Q_i} = a - 1  && \fnum{u}{Q_j} \leq a - \delta && \fnum{v}{Q_i} \leq b-1 && \fnum{v}{Q_j} = b-1
		\end{align*}
		
		It is easy to see by inspection that the resulting exchange is an improving localized exchange satisfying \eqref{e:implocex} and \eqref{e:implocex2}. Thus $\mathcal{Q}$ is an improving localized exchange, and we obtain a contradiction via \autoref{lem:implocex}.\qed
	\end{enumerate}
\end{proof}

Note that \autoref{t:delta-skew} restricts the structure of an optimal solution. However, \autoref{t:delta-skew} by itself is not sufficient to infer that a bottom-up dynamic program can restrict its attention to only signatures with bounded skews.  This is because a dynamic program typically needs to maintain partial results (e.g., best combination of signatures of a \emph{subset} of children at a node) as it works its way towards finding an optimal solution. We need to show that even these partial results satisfy the same structure as an optimal solution, namely  signatures in partial results also have bounded skews.
This is the focus of the next theorem.

In order to clearly state the theorem, we introduce some notation for combining multi-placements $\MP = (P_1, ..., P_m)$ and $\MP' = (P_1',...,P_m')$. Specifically, we define the \emph{direct sum} of two multi-placements as \mbox{$\MP \oplus \MP' = (P_1 \cup P_1' ,..., P_m \cup P_m')$}. Associativity and commutativity of $\oplus$ easily follow from that of set union.

\begin{theorem}\label{t:skew-ordering}
	Fix an arbitrary linear order, $\prec$, on the nodes of tree $T$. For any choice of signature $\vec{\sigma}$ having skew $\delta$, there exists an optimal multi-placement $\MP^\ast$ with signature $\vec{\sigma}$ in which, for every node $u$ having children $c_1,...,c_t$, where $c_1 \prec c_2 \prec ... \prec c_t$, and for all values of $s \in \{1,...,t\}$, we have that
	\begin{equation*}
	\bigoplus_{j=1}^s \MP^\ast_{c_j} \text{ has skew at most $\delta$}.
	\end{equation*}
\end{theorem}

The previous two theorems imply that any optimal solution with a signature of skew at most $\delta$ can be constructed from signatures whose skews are also upper-bounded by $\delta$. Moreover, the \emph{order} in which we combine the partial results from the children to construct a partial result at the parent is not relevant. No matter how the children of a node are ordered, a route to an optimum multi-placement which uses only partial results obtained using signatures with skew at most $\delta$ exists.  

To prove \autoref{t:skew-ordering}, we will require the technical lemma stated below.

\begin{lemma}\label{c:skew-ordering}
	Let $u$ be the root of tree $T$, and let $u$ have children given in the arbitrary order $c_1,...,c_t$. Then, for every choice of signature $\vec{\sigma}$ with skew $\delta$, there is an optimal multi-placement $\MP^\ast = (P^\ast_1, ..., P^\ast_m)$ with signature $\vec{\sigma}$ in which for all values of $s \in \{1,...,t\}$,
	\begin{equation}\label{eq:dirsum-skew-delta}
	\bigoplus_{j=1}^s \MP^\ast_{c_j} \text{ has skew at most $\delta$}.
	\end{equation}
\end{lemma}

Once established, it is clear that this lemma can be used to easily show \autoref{t:skew-ordering} by applying \autoref{c:skew-ordering} at the root, and then recursively at each of the subtrees formed by descendants of the children of the root. In this manner, once \autoref{c:skew-ordering} is established, one can easily show \autoref{t:skew-ordering} by structural induction. Thus, we will focus on proving \autoref{c:skew-ordering}.

\begin{proof}[\autoref{c:skew-ordering}]
	The proof is by construction. We show that any optimal multi-placement that does not satisfy the statement of the theorem can be transformed into another optimal multi-placement that satisfies the statement of the theorem using a series of exchanges.

	In our proof, we will need to refer to several \subplacement{s} of $\MP$ and their constituent placements. We refer to the multi-placement $\bigoplus_{j=1}^k \MP_{c_j}$ as the \emph{partial multi-placement of $\MP$ up to $k$}. But we also need a concise symbol which refers to the replicas from the multi-placement $\bigoplus_{j=1}^k \MP_{c_j}$ which are in block $i$ and are placed on the subtree of $T$ rooted at node $v$. We will denote this set of replicas by $P_{k,i}^v$, and refer to it as a placement. Note that we can refer to placement $\MP$ as ``\subplacement{"} $\MP_u$ without causing confusion and, having done so, it is also clear that we can refer easily to each of its constituent placements, since $\MP= \MP_u = \bigoplus_{j=1}^t \MP_{c_j} = (P^u_{t,1}, ..., P^u_{t,m})$. Our argument focuses on constructing an exchange $\mathcal{Q}$. We will use the symbol $\Q_{k,i}^v$ to refer to the same portion of $\mathcal{Q}$ that $\P_{k,i}^v$ refers to.

	Consider an optimal multi-placement $\MP$ with \signature{} $\vec{\sigma}$. If $\MP$ satisfies the statement of the theorem, then the theorem clearly holds. Otherwise, we show how to construct an exchange that takes us ``closer to our goal" by reducing the skew between a pair of ``offending" blocks to at most $\delta$. By repeatedly performing such exchanges, we can eventually construct another optimal multi-placement with the same \signature{} $\vec{\sigma}$ that satisfies the statement of the theorem. 
	To that end, let $k$ be the minimum value such that
	\[
	\bigoplus_{j=1}^k \MP_{c_j} \text{ has skew at least $\delta + 1$}
	\]
	Since $\bigoplus_{j=1}^k\MP_{c_j}$ has a skew of at least $\delta + 1$ there must exist blocks $i$ and $j$ such that
	\begin{equation}\label{e:root-bound}
	|\fnum{u}{\P_{k,i}^u} - \fnum{u}{P_{k,j}^u}| > \delta
	\end{equation}
	There may be multiple candidate block pairs $(i,j)$ for which the above statement holds. Let the set of such pairs be denoted by
	\[
	pairs(k) = \{(i,j) : i < j \text{ and } \left|\fnum{u}{\P_{k,i}^u} - \fnum{u}{\P_{k,j}^u} \right| > \delta)  \}.
	\]

	Let $(i,j)$ be the lexico-minimum element in $pairs(k)$. Since this block pair does not have a skew of $\delta$, clearly, one of two statements must be true. For some value of $a$, we have either
	\begin{align}
	\fnum{u}{P_{k,i}^u} &= a, & \fnum{u}{P_{k,j}^u} &\leq a - \delta - 1,\label{eq:lmin-case1}\\
	\intertext{or,}
	\fnum{u}{P_{k,i}^u} &\leq a -\delta - 1, &  \fnum{u}{P_{k,j}^u} &= a.\label{eq:lmin-case2}
	\end{align}
	both of which follow easily from \eqref{e:root-bound}. In certain cases of this argument we can assume one of \eqref{eq:lmin-case1} or \eqref{eq:lmin-case2} without loss of generality, but in others we must deal with each separately.
	
	First, notice that since each of the child \subplacement{s} are disjoint, for any value of $k \in \{1,...,t\}$, we must have that
	\begin{equation}\label{e:disjoint-decomp}
	\fnum{u}{\P_{k,i}^u} = \sum_{j=1}^k \fnum{c_j}{\P_{k,i}^{c_j}}.
	\end{equation}
	We will appeal to this fact at several points in the sequel.
	
	Suppose that \eqref{eq:lmin-case1} holds. Then, since $k$ was chosen to be minimum, it must be the case that
	\[
	\fnum{c_k}{\P_{k,i}^{c_k}} = b ~\text{ and }~ \fnum{c_k}{\P_{k,j}^{c_k}} \leq b-1,
	\]
	since otherwise, by the decomposition of \eqref{e:disjoint-decomp}, it is not possible both for \eqref{e:root-bound} to hold and for $k$ to be minimum.
	Moreover, since the overall skew of $\MP$ is no greater than $\delta$, there must be some child $c_\ell$ with $k < \ell$ which fixes the skew w.r.t. blocks $i$ and $j$. That is, at child $c_\ell$ we must have that
	\[
	\fnum{c_\ell}{\P_{k,i}^{c_\ell}} \leq d-1 ~\text { and }~ \fnum{c_\ell}{\P_{k,i}^{c_\ell}} = d.	
	\]
	Symmetric statements hold when \eqref{eq:lmin-case2} holds instead of \eqref{eq:lmin-case1}. We will defer their statement until they are needed.
	
	We proceed in four cases, based upon which of the upper bounds involving $b$ and $d$ are tight. In the first three cases (Cases 1-3) we use \autoref{lem:implocex} to derive a contradiction. In the fourth case (Case 4) we remove the pair $(i,j)$ from $pairs(k)$ while not adding any previously visited pairs to $pairs(k')$ for any $k' \leq k$, thereby moving closer to our goal.
	
	We can treat Cases 1-3 by assuming that \eqref{eq:lmin-case1} holds without loss of generality.  In each of these cases, a symmetric argument applies when \eqref{eq:lmin-case2} is true. The same symmetry does not apply in Case 4, which must be handled more carefully depending on which of equations \eqref{eq:lmin-case1} or \eqref{eq:lmin-case2} holds.
	
	\begin{enumerate}[{Case} 1)]
		\item \mbox{$\fnum{c_k}{\P_{k,i}^{c_k}} = b ~;~ \fnum{c_k}{\P_{k,j}^{c_k}} \leq b - 2 ~;~
			\fnum{c_\ell}{\P_{k,i}^{c_\ell}} = d-1 ~;~ \fnum{c_\ell}{\MP_{k,j}^{c_\ell}} = d$}.\vspace{2mm}
		In this case we construct $\mathcal{Q}$, an improving localized exchange of $\MP$ with target
		blocks $i$ and $j$ and carriers $c_k$ and $c_\ell$ as follows.
		First, we swap the indices of blocks $i$ and $j$ by setting $Q_{k,i}^{c_\ell} = \P_{k,j}^{c_\ell}$ and $Q_{k,j}^{c_\ell} = \P_{k,i}^{c_\ell}$. This effectively adds a replica to block $i$ and removes a replica from block $j$. To fix the signature, we move one of $c_k$'s replicas from block $i$ to block $j$. The result is summarized as\vspace{2mm}
		\mbox{$\fnum{c_k}{Q_{k,i}^{c_k}} = b -1 ~;~ \fnum{c_k}{Q_{k,j}^{c_k}} \leq b - 1 ~;~
			Q_{k,i}^{c_\ell} = \P_{k,j}^{c_\ell} ~;~ Q_{k,j}^{c_\ell} = \P_{k,i}^{c_\ell}$}.\vspace{2mm}
		
		Thus $\mathcal{Q}$ is an improving localized exchange of $\MP$, and the construction is complete.\vspace{2mm}
		
		\item \mbox{$\fnum{c_k}{\P_{k,i}^{c_k}} = b ~;~ \fnum{c_k}{\P_{k,j}^{c_k}} = b - 1 ~;~
			\fnum{c_\ell}{\P_{k,i}^{c_\ell}} \leq d-2 ~;~ \fnum{c_\ell}{\MP_{k,j}^{c_\ell}} = d$}.\vspace{2mm}
		This case is entirely symmetric to the prior one. We do with the replicas of $c_k$ what we we did in the prior case with the replicas of $c_\ell$ and vice versa. The result is summarized as\vspace{2mm}
		\mbox{$	Q_{k,i}^{c_k} = \P_{k,j}^{c_k} ~;~ Q_{k,j}^{c_k} = \P_{k,i}^{c_k} ~;~
			\fnum{c_\ell}{Q_{k,i}^{c_\ell}} \leq d -1 ~;~ \fnum{c_\ell}{Q_{k,j}^{c_\ell}} = d - 1
			$}\vspace{2mm}
		
		Thus $\mathcal{Q}$ is an improving localized exchange of $\MP$, and the construction is complete.\vspace{2mm}
		
		\item \mbox{$\fnum{c_k}{\P_{k,i}^{c_k}} = b ~;~ \fnum{c_k}{\P_{k,j}^{c_k}} \leq b - 2 ~;~
			\fnum{c_\ell}{\P_{k,i}^{c_\ell}} \leq d-2 ~;~ \fnum{c_\ell}{\P_{k,i}^{c_\ell}} = d$}.\vspace{2mm}
		In this case, we move one of $c_k$'s replicas from block $i$ to block $j$ and move one of $c_\ell$'s replicas from block $j$ to block $i$. This is summarized as\vspace{2mm}
		\mbox{$\fnum{c_k}{Q_{k,i}^{c_k}} = b-1 ~;~ \fnum{c_k}{Q_{k,j}^{c_k}} \leq b - 1 ~;~
			\fnum{c_\ell}{Q_{k,i}^{c_\ell}} \leq d-1 ~;~ \fnum{c_\ell}{Q_{k,i}^{c_\ell}} = d-1$}.\vspace{2mm}
		Thus $\mathcal{Q}$ is an improving localized exchange of $\MP$, and the construction is complete.\vspace{2mm}
		
		\item
		We split into two cases.
		\begin{enumerate}[{Case 4}a)]
			\item In this case \eqref{eq:lmin-case1} holds. Recall this means that
			\[\fnum{u}{P_{k,i}^u} = a ~;~ \fnum{u}{P_{k,j}^u} \leq a - \delta - 1,\]
			implying that we must have	
			\[\fnum{c_k}{\P_{k,i}^{c_k}} = b ~;~ \fnum{c_k}{\P_{k,j}^{c_k}} = b - 1 ~;~
			\fnum{c_\ell}{\P_{k,i}^{c_\ell}} = d-1 ~;~ \fnum{c_\ell}{\P_{k,i}^{c_\ell}} = d.\]
			Furthermore, we must have equality where \mbox{$\fnum{u}{P_{k,j}^u} = a - \delta - 1$}, which we shall argue as follows. Suppose for the purpose of obtaining a contradiction that \mbox{$\fnum{u}{P_{k,j}^u} < a - \delta - 1 \implies \fnum{u}{P_{k,j}^u} \leq a - \delta - 2$}. Then, since by \eqref{e:disjoint-decomp}, we know $\fnum{u}{P_{k,j}^u} = \fnum{u}{P_{k-1,j}^{u}} + \fnum{c_k}{P_{k,j}^{c_k}}$, we obtain
			\begin{equation}\label{eq:equality-contra1}
			\fnum{u}{P_{k-1,j}^u} \leq a - \delta - b - 1. 
			\end{equation}
			But by \eqref{e:disjoint-decomp}, we can similarly obtain that $\fnum{u}{P_{k,i}^u} = \fnum{u}{P_{k-1,i}^{u}} + \fnum{c_k}{P_{k,i}^u}$, which implies that 
			\begin{equation}\label{eq:equality-contra2}
			\fnum{u}{P_{k-1,i}^u} = a - b.
			\end{equation}
			But \eqref{eq:equality-contra1} and \eqref{eq:equality-contra2} together imply that $(i,j) \in pairs(k-1)$, contradicting that $k$ was chosen to be minimum. Thus $\fnum{u}{P_{k,j}^u} = a - \delta - 1$.
			
			In this case, we construct $\mathcal{Q}$, a localized exchange of $\MP$ which removes $(i,j)$ from $pairs(k)$.
			We construct $\mathcal{Q}$ by swapping all replicas of block $i$ with block $j$ at both child $c_k$ and $c_\ell$. Specifically, we set $Q_{k,i}^{c_k} = P_{k,j}^{c_k}$, $Q_{k,j}^{c_k} = Q_{k,i}^{c_k}$, $Q_{k,i}^{c_\ell} = P_{k,j}^{c_\ell}$, and $Q_{k,j}^{c_\ell} = P_{k,i}^{c_\ell}$
			Thus, clearly $\g{\mathcal{Q}} = \g{\MP}$ as observed in the argument used to prove \autoref{lem:implocex}. 
			This removes $(i,j)$ from $pairs(k)$, since, when we swap the replicas of blocks $i$ and $j$ at child $c_k$, we effectively decrement the failure number w.r.t. block $i$ by 1 and increment the failure number w.r.t. block $j$ by 1. Thus, the exchange $\mathcal{Q}$ is summarized as
			\begin{gather*}\fnum{u}{Q_{k,i}^u} = a -1 ~;~ \fnum{u}{Q_{k,j}^u} = a - \delta
			~;~\\
			\fnum{c_k}{Q_{k,i}^{c_k}} = b -1 ~;~ \fnum{c_k}{Q_{k,j}^{c_k}} = b ~;~
			\fnum{c_\ell}{Q_{k,i}^{c_\ell}} = d ~;~ \fnum{c_\ell}{Q_{k,i}^{c_\ell}} = d-1.
			\end{gather*}
			
			The first line clearly implies $(i,j) \notin pairs(k)$. Now we need to show that this exchange does not add any previously visited pairs to $pairs(k')$ for any $k' \leq k$. If $k' < k$, then this is easily seen, since the only \subplacement{s} affected by the exchange are those at $c_k$ and $c_\ell$, and $k' < k < \ell$.
			
			We must do some more work to show that no previously visited pairs are added to $pairs(k)$. That is to say, we must show that any pair $(x,y) \llt (i,j)$ is not added to $pairs(k)$ by the exchange. It is clear that only pairs for which $x = i$ or $y \in \{i,j\}$ might be added, since only blocks $i,j$ are affected by the exchange, and if $x = j$ then $(j,y) \not\llt (i,j)$ as required.
			
			Since $(i,j)$ was the lexicographically smallest element of $pairs(k)$, we know that our bounds on the skew must hold for all block pairs $(x,y)$ which are lexicographically smaller. This allows us to derive bounds on the possible values of $\fnum{u}{P_{k,x}^u}$ and $\fnum{u}{P_{k,y}^u}$ which will make it clear that no previously visited pairs are added by the exchange.
			
			Since $(x,j) \llt (i,j)$ for values of $x < i$, we obtain that \mbox{$\fnum{u}{P_{k,x}^u} \leq \fnum{u}{P_{k,j}^u} + \delta \leq a - 1 $}, since otherwise the block pair $(x,j)$ would be in $pairs(k)$, implying $(i,j)$ was not the lexico-minimum. Likewise, since $(x,i) \llt (i,j)$ we obtain that \mbox{$\fnum{u}{P_{k,x}^u} \geq \fnum{u}{P_{k,i}^u} - \delta = a - \delta$.} Thus \mbox{$\fnum{u}{P_{k,x}^u} \in [a - \delta , a - 1]$}.
			
			We similarly bound $\fnum{u}{P_{k,y}^u}$ for values of $y$ where $i < y < j$ as follows. Since $(i,y) \llt (i,j)$, we have that $\fnum{u}{P_{k,y}^u} \leq a - \delta$. Since, for any value of $x < i$, $(x,y) \llt (i,j)$ we know that $\fnum{u}{P_{k,y}^u} \leq \fnum{u}{P_{k,x}^u} + \delta \leq a + \delta - 1$. Thus $\fnum{u}{P_{k,y}^u} \in [a-\delta, a+\delta - 1]$. 
			
			Since $P_{k,x}^u = Q_{k,x}^u$ and $P_{k,y}^u = Q_{k,y}^u$ for all values of $x$ and $y$, we can show that, after the exchange,
			\begin{enumerate}[I)]
				\item $(x,i) \notin pairs(k)$, since \mbox{$\fnum{u}{Q_{k,x}^u} \in [a-\delta, a-1]$} and \mbox{$\fnum{u}{Q_{k,i}^u} = a-1$}.
				\item $(x,j) \notin pairs(k)$, since \mbox{$\fnum{u}{Q_{k,x}^u} \in [a-\delta, a-1]$} and \mbox{$\fnum{u}{Q_{k,j}^u} = a - \delta$.}
				\item $(i,y) \notin pairs(k)$, since \mbox{$\fnum{u}{Q_{k,y}^u} \in [a-\delta, a+ \delta-1]$}, and $\fnum{u}{Q_{k,i}^u} = a-1$.
			\end{enumerate}
			Since previously visited pairs of $pairs(k)$ are one of (I), (II), or (III), no such pairs are added.
			
			\item If instead \eqref{eq:lmin-case2} holds then the same argument used in the prior case to show that equality holds in \eqref{eq:lmin-case1} also applies to \eqref{eq:lmin-case2}, where the roles of $i$ and $j$ are reversed. Thus, we have
			\[\fnum{u}{P_{k,i}^u} = a -\delta - 1 ~;~ \fnum{u}{P_{k,j}^u} = a,\]
			implying that we must have	
			\[\fnum{c_k}{\P_{k,i}^{c_k}} = b-1 ~;~ \fnum{c_k}{\P_{k,j}^{c_k}} = b ~;~
			\fnum{c_\ell}{\P_{k,i}^{c_\ell}} = d ~;~ \fnum{c_\ell}{\P_{k,i}^{c_\ell}} = d-1.\]
			The exact same exchange is formed using the exact same operations used in the prior case, except that this exchange is now summarized as
			\begin{gather*}\fnum{u}{Q_{k,i}^u} = a -\delta ~;~ \fnum{u}{Q_{k,j}^u} = a - 1
			~;~\\
			\fnum{c_k}{Q_{k,i}^{c_k}} = b ~;~ \fnum{c_k}{Q_{k,j}^{c_k}} = b -1 ~;~
			\fnum{c_\ell}{Q_{k,i}^{c_\ell}} = d-1 ~;~ \fnum{c_\ell}{Q_{k,i}^{c_\ell}} = d.
			\end{gather*}
			Under these conditions we show that no previously visited pairs are added to $pairs(k)$ as follows. First, we bound values of $\fnum{u}{P_{k,x}^u}$ and $\fnum{u}{P_{k,y}^u}$ when $x = i$ or $y \in \{i,j\}$ below.
			
			When $(x,i) \llt (i,j)$, we obtain that $\fnum{u}{P_{k,x}^u} \leq a$. Likewise, when $(x,j) \llt (i,j)$, we obtain that $\fnum{u}{P_{k,x}^u} \geq a - \delta$. Similarly, when $(i,u) \llt (i,j)$, we obtain that $\fnum{u}{P_{k,y}^u} \leq a-1$. And finally, when $(x,y) \llt (i,j)$, we obtain that $\fnum{u}{P_{k,y}^u} \geq a - 2\delta$. 
			
			Thus, after the exchange we can show that
			\begin{enumerate}[I)]
				\item $(x,i) \notin pairs(k)$, since \mbox{$\fnum{u}{Q_{k,x}^u} \in [a-\delta, a]$} and $\fnum{u}{Q_{k,i}^u} = a-\delta$.
				\item $(x,j) \notin pairs(k)$, since \mbox{$\fnum{u}{Q_{k,x}^u} \in [a-\delta, a]$} and $\fnum{u}{Q_{k,j}^u} =$ \mbox{ $a - 1$}.
				\item $(i,y) \notin pairs(k)$, since \mbox{$\fnum{u}{Q_{k,y}^u} \in [a-2\delta, a-1]$}, and $\fnum{u}{Q_{k,i}^u} = a-\delta$.
			\end{enumerate}
			Thus in this case also, previously visited pairs are not added back to $pairs(k)$. 
		\end{enumerate}
	\end{enumerate}
	
	We can perform such exchanges repeatedly to remove all pairs from $pairs(k)$. In other words, by repeatedly performing such exchanges, we can ensure that $pairs(1)$, $pairs(2)$, $\ldots$, $pairs(t)$ all become empty at which point we will obtain an optimal multi-placement that satisfies \eqref{eq:dirsum-skew-delta}. The theorem is thus proved. \qed

\end{proof}

\subsection{An Exact Algorithm for Optimal Multi-placements}\label{s:mp-algorithm}

Our algorithm for finding an optimal multi-placement is a bottom-up dynamic program.  It uses the following key properties to achieve the desirable running time. First, as mentioned earlier, the optimal value of the objective function of a multi-placement only depends on its signature (by \autoref{lem:signature-solution}). This implies that the state information maintained by our dynamic program is a function of signature and not multi-placement. Second, it is sufficient to only consider those signatures whose skews are upper-bounded by the skew of the desired multi-placement (by \autoref{t:delta-skew} and \autoref{t:skew-ordering}). This further implies that the state information maintained by our dynamic program is a function of a signature with bounded skew.

For each node of the tree, we maintain a dynamic programming table in which, for each signature $\vec{\sigma}$ with skew $\delta$, we store $G_u(\vec{\sigma})$, the optimal value of any \subplacement at node $u$ which has signature $\vec{\sigma}$. Given completely filled out tables for every child of $u$, we show how to combine the results to obtain a filled table for $u$ itself in \autoref{s:mp-dp}. Once filled out, the table for $u$ will contain the optimal solution, since a ``\subplacement{"} of the root is a multi-placement.

As we shall see shortly, we will need to understand how the signature is affected when two disjoint multi-placements are combined. Two signatures $\vec{\sigma}_1$ and $\vec{\sigma}_2$ can be combined to form signature $\vec{\sigma}$ only if two disjoint multi-placements with signatures $\vec{\sigma}_1$ and $\vec{\sigma}_2$ can be combined to yield a multi-placement with signature $\vec{\sigma}$. As an example, consider the signatures $\vec{\sigma}_1 = \langle 0,3,1\rangle$ and \mbox{$\vec{\sigma}_2 = \langle 0,2,2 \rangle$}. These signatures can be combined to yield the signature $\vec{\sigma} = \langle2,1,1\rangle$. To see this, notice that we can combine disjoint multi-placements $\MP_1$ and $\MP_2$ with signatures $\vec{\sigma}_1$ and $\vec{\sigma}_2$ to yield a multi-placement $\MP$ with signature $\vec{\sigma}$, by doing the following. Combine two placements of size $1$ from $\MP_1$  with two placements of size $1$ from $\MP_2$. This yields the two placements of size $2$ in $\vec{\sigma}$. To obtain the single placement of size $1$, combine one placement of size $1$ from $\MP_1$ with the empty placement in $\MP_2$. Finally, the remaining empty placements are combined to yield the empty placement in $\vec{\sigma}$. In general, there may be multiple ways to combine two signatures, each of which may yield a different signature as a result. 

In \autoref{s:combo-signatures} we give an algorithm to compute $\Phi(\vec{\sigma}, \delta)$, the set of all signature pairs which can be combined to yield $\vec{\sigma}$, which has skew $\delta$, in which both signatures in the pair \emph{also} have skew at most $\delta$.

Suppose that $u$ and $v$ are the only children of $w$. Then $\Phi(\vec{\sigma}, \delta)$ provides the mapping that we use to compute the value of $G_w(\vec{\sigma})$. Roughly speaking, for every pair ($\vec{\sigma}_1, \vec{\sigma}_2) \in \Phi(\vec{\sigma}, \delta)$ our algorithm uses the values in $G_u(\vec{\sigma}_1)$ and $G_v(\vec{\sigma}_2)$ to update the value in $G_w(\vec{\sigma})$. Because $\Phi(\vec{\sigma},\delta)$ enumerates all possible signature pairs which can be combined to yield $\vec{\sigma}$, by considering every pair in $\Phi(\vec{\sigma}, \delta)$ we consider all possible ways to combine optimal \subplacement{s} at $u$ and $v$ to yield a \subplacement at $w$ which has signature $\vec{\sigma}$. The optimal \subplacement at $w$ will be comprised of one such pair. To ensure we attain the optimum, we try all pairs.

Our algorithm starts by computing a table of values of $\Phi(\vec{\sigma}, \delta)$ for all values of $\vec{\sigma}$ (recall that $\delta$ is fixed and given as input). An algorithm to compute this table is described in \autoref{s:combo-signatures}. The dynamic program itself is described in \autoref{s:mp-dp}. Briefly, the dynamic program fills tables of $G_u(\vec{\sigma})$ for all nodes $u$ and signatures $\vec{\sigma}$ by visiting each edge of the tree according to a modified post-order traversal. Each visited edge combines two previously unconnected portions of the tree. When each edge is visited, we update the dynamic programming table associated with the parent node of the edge in question, using the table of values for $\Phi$ to determine how the cells in the dynamic programming table of $G_u(\vec{\sigma})$ are combined. Once the table for the root node has been filled, we can obtain the optimal solution in the usual way by examining a record of the dynamic programming computation.

\subsubsection{Combining Signatures with Bounded Skew}\label{s:combo-signatures}

Given a signature $\vec{\sigma}$ with skew at most $\delta$, we wish to find all signatures $\vec{\sigma}'$ and $\vec{\sigma}''$, each with skew at most $\delta$ which can be combined to yield $\vec{\sigma}$. Every such combination can be expressed by a positive integer matrix as follows. Let $x_{ij}$ be the number of placements with size $\repfact - i$ from $\vec{\sigma}'$ which are combined with placements of size $\repfact - j$ from $\vec{\sigma}''$ where $i,j \in \{0,...,\repfact\}$ (note we use the convention that matrix rows and columns are zero-indexed). Since no more than $\vec{\sigma}'_i$ placements of size $\repfact - i$ can be taken from $\vec{\sigma}'$, we have that $\vec{\sigma}'_i = \sum_j x_{ij}$. Likewise, we have that $\vec{\sigma}''_j = \sum_i x_{ij}$, since no more than $\vec{\sigma}''_j$ placements of size $j$ can be taken from $\vec{\sigma}''$s. Furthermore, we must have exactly as many placements in $\vec{\sigma}'$ and $\vec{\sigma}''$ as we do in $\vec{\sigma}$, that is,
\[\sum_{i=0}^{\repfact} \vec{\sigma}_i' = \sum_{i=0}^{\repfact} \vec{\sigma}_i'' = \sum_{i=0}^{\repfact}\sum_{j=0}^\repfact x_{ij} = \sum_{i=0}^{\repfact} \vec{\sigma}_i = \m. \]
Finally, the values of $x_{ij}$ must combine to yield $\vec{\sigma}$. Since when placements of sizes $\repfact - i$ and $\repfact - j$ are combined they yield one of size $2\repfact - i - j$. Setting $2 \repfact - i - j = \repfact - k$, we have $\repfact +k = i+j$, and thus
\[\vec{\sigma}_{k} = \sum_{(i,j) :  i + j  = \repfact + k } x_{ij}. \]
For fixed $k$, the above sum ranges over the anti-diagonals\footnote{An anti-diagonal of a square matrix is a set of cells $x_{ij}$ for which $i+j=k$, for some fixed $k$.} of a matrix formed by the entries of $x_{ij}$. Thus, the vectors $\vec{\sigma}$, $\vec{\sigma}'$, and $\vec{\sigma}''$ can each be seen to arise from the row, column, and the \emph{last} $\rho + 1$ anti-diagonal sums, respectively, of a non-negative $(\repfact + 1) \times (\repfact + 1)$ integer matrix $X$, where the sum of all entries in $X$ is $m$. An example matrix $X$ and its relationship to $\vec{\sigma}$, $\vec{\sigma}'$ and $\vec{\sigma}''$ is depicted in \autoref{fig:integer-matrix-example}.

\begin{figure}
	\centering
	\begin{subfigure}{0.48\linewidth}
		\begin{tikzpicture}[diagonal/.style={color=gray!30!white, line width=2pt}]
		\pgfdeclarelayer{bg}
		\pgfsetlayers{bg,main}
		
		\matrix (m) [matrix of math nodes, nodes in empty cells, right delimiter={]}, left delimiter ={[}] {
			0 & 0 & 0 & 0 & 0 & 0 \\
			0 & 0 & 0 & 0 & 0 & 0 \\
			0 & 0 & 0 & 0 & 1 & 1 \\
			0 & 0 & 0 & 1 & 2 & 2 \\
			0 & 0 & 0 & 1 & 1 & 0 \\
			0 & 0 & 0 & 0 & 0 & 0 \\
		} ;
		
		\matrix (sigprimeprime) [matrix anchor= north west, matrix of math nodes, nodes in empty cells, xshift=0.5cm] at (m.north east) {
			0 \\
			0 \\
			2 \\
			5 \\
			2 \\
			0 \\
		} ;
		
		\node [left=-0.05cm of sigprimeprime-1-1] {$\langle$};
		\node [right=-0.05cm of sigprimeprime-6-1.east] {$\rangle = \vec{\sigma}''$};
		
		\matrix (sigprime) [matrix anchor=south west, matrix of math nodes, nodes in empty cells] at (m.north west) {
			0 & 0 & 0 & 2 & 4 & 3 & \rangle \\
		} ; 
		
		\node[left=-0.1cm of sigprime] {$\vec{\sigma}' = \langle $};
		
		\matrix (sigma) [matrix anchor = north west, matrix of math nodes, nodes in empty cells, xshift=-0.75cm] at (m.south west) {
			0 & 2 & 4 & 3 & 0 & 0 & \rangle \\
		} ;
		\node [left=-0.1cm of sigma] {$\vec{\sigma} = \langle$};
		
		\begin{pgfonlayer}{bg}
		\draw[diagonal] (m-1-6.north east) -- (sigma-1-1.north east);
		\draw[diagonal] (m-2-6.north east) -- (sigma-1-2.north east);
		\draw[diagonal] (m-3-6.north east) -- (sigma-1-3.north east);
		\draw[diagonal] (m-4-6.north east) -- (sigma-1-4.north east);
		\draw[diagonal] (m-5-6.north east) -- (sigma-1-5.north east);
		\draw[diagonal] (m-6-6.north east) -- (sigma-1-6.north east);
		\end{pgfonlayer}
		
		\draw[dashed] (m-3-4.north west) rectangle (m-5-6.south east);
		\end{tikzpicture}
		\subcaption{}\label{f:matrix1}
	\end{subfigure}
	\begin{subfigure}{0.48\linewidth}\begin{tikzpicture}[diagonal/.style={color=gray!30!white, line width=2pt}]
		\pgfdeclarelayer{bg}
		\pgfsetlayers{bg,main}
		
		\matrix (m) [matrix of math nodes, nodes in empty cells, right delimiter={]}, left delimiter ={[}] {
			0 & 0 & 0 & 0 & 0 & 0 \\
			0 & 0 & 0 & 0 & 0 & 0 \\
			0 & 0 & 0 & 0 & 0 & 2 \\
			0 & 0 & 0 & 2 & 2 & 1 \\
			0 & 0 & 0 & 0 & 2 & 0 \\
			0 & 0 & 0 & 0 & 0 & 0 \\
		} ;
		
		\matrix (sigprimeprime) [matrix anchor= north west, matrix of math nodes, nodes in empty cells, xshift=0.5cm] at (m.north east) {
			0 \\
			0 \\
			2 \\
			5 \\
			2 \\
			0 \\
		} ;
		
		\node [left=-0.05cm of sigprimeprime-1-1] {$\langle$};
		\node [right=-0.05cm of sigprimeprime-6-1.east] {$\rangle = \vec{\sigma}''$};
		
		\matrix (sigprime) [matrix anchor=south west, matrix of math nodes, nodes in empty cells] at (m.north west) {
			0 & 0 & 0 & 2 & 4 & 3 & \rangle \\
		} ; 
		
		\node[left=-0.1cm of sigprime] {$\vec{\sigma}' = \langle $};
		
		\matrix (sigma) [matrix anchor = north west, matrix of math nodes, nodes in empty cells, xshift=-0.75cm] at (m.south west) {
			0 & 2 & 4 & 3 & 0 & 0 & \rangle \\
		} ;
		\node [left=-0.1cm of sigma] {$\vec{\sigma} = \langle$};
		
		\begin{pgfonlayer}{bg}
		\draw[diagonal] (m-1-6.north east) -- (sigma-1-1.north east);
		\draw[diagonal] (m-2-6.north east) -- (sigma-1-2.north east);
		\draw[diagonal] (m-3-6.north east) -- (sigma-1-3.north east);
		\draw[diagonal] (m-4-6.north east) -- (sigma-1-4.north east);
		\draw[diagonal] (m-5-6.north east) -- (sigma-1-5.north east);
		\draw[diagonal] (m-6-6.north east) -- (sigma-1-6.north east);
		\end{pgfonlayer}
		
		\draw[dashed] (m-3-4.north west) rectangle (m-5-6.south east);
		\end{tikzpicture}
		\subcaption{}\label{f:matrix2}
	\end{subfigure}
	\caption[Illustration of two matrix supports which yield a given signature.]{(\subref{f:matrix1}) One of the 6x6 matrix supports that yields signature $\vec{\sigma} = \langle 0,0,2,4,3,0 \rangle$ as a combination of $\vec{\sigma}' = \langle0,0,0,2,4,3 \rangle$ and \mbox{$\vec{\sigma}'' = \langle 0,0,2,5,2,0 \rangle$.} Entries of $\vec{\sigma}$ are formed by the highlighted diagonal sums. (\subref{f:matrix2}) A second matrix support for the same three vectors.}\label{fig:integer-matrix-example}
\end{figure}

Clearly, every pair of signatures $\vec{\sigma}', \vec{\sigma}''$ which can be validly combined to yield $\vec{\sigma}$ arise from such an integer matrix. We call the matrix $X$ the \emph{matrix support} of the triple $(\vec{\sigma}, \vec{\sigma}', \vec{\sigma}'')$. Notice that a given triple may have multiple matrix supports (see \autoref{f:matrix2}), but every possible pair of signatures $\vec{\sigma}'$ $\vec{\sigma}''$ which can be combined to form $\vec{\sigma}$ has \emph{at least one} matrix support. If $\vec{\sigma}'$ and $\vec{\sigma}''$ each have girth $\rho$, the matrix support will have dimension $(\rho + 1) \times (\rho + 1)$. Moreover each such matrix support is lower anti-triangular\footnote{The \emph{main anti-diagonal} of a square matrix is the set of cells comprising the anti-diagonal starting at the lower-left cell and ending at the upper-right cell. An \emph{lower anti-triangular matrix} is a square matrix in which all entries above the main anti-diagonal are zero.}. To see this, notice that if $X$ is not lower anti-triangular, some entry above the main anti-diagonal is non-zero, which implies that some pair of placements is combined to yield one with size $2\repfact - i - j > \repfact$, since, for any entry above the main anti-diagonal, $i+j < \repfact$. Thus, the placement so yielded would have a size which exceeds the required limit of on the girth of $\vec{\sigma}$, namely $\repfact$. Hence no entry above the main anti-diagonal is non-zero. Therefore, by iterating over all lower anti-triangular matrices $[0,m]^{(\rho+1)\times(\rho+1)}$ we can visit all vector triples which have a matrix support, and thereby compute a table of values for $\Phi(\vec{\sigma}, \delta)$. 

However, when every signature involved has skew bounded by $\delta$, many of the entries of $X$ will be zero. For example in \autoref{fig:integer-matrix-example} each signature has a skew of $2$, and thus only the $3\times 3$ sub-matrix surrounded by a dotted rectangle is non-zero. Since $\vec{\sigma}'$ and $\vec{\sigma}''$ have skews bounded by $\delta$ and these vectors are the column and row sums of $X$, we only need to consider those non-negative integer matrices in which only a $(\delta+1)\times(\delta+1)$ sub-matrix of entries are non-zero. Let $Y$ denote this sub-matrix. Since $X$ is lower anti-triangular, the non-zero entries of $Y$ must on or below the main anti-diagonal of $X$. Moreover, since $\vec{\sigma}$ has a skew bounded by $\delta$, we obtain the additional restriction that at most $\delta+1$ contiguous anti-diagonals of this matrix are non-zero. Such a matrix is referred to as an \emph{anti-banded matrix} and the number of non-zero anti-diagonals is referred to as the \emph{bandwidth} of the matrix. Putting it all together, we have that the set of possible matrix supports is given by the set of lower anti-triangular non-negative integer $(\repfact + 1 )\times (\repfact+1)$ matrices in which the only non-zero entries are contained within a $(\delta+1)\times(\delta+1)$ anti-banded sub-matrix with bandwidth at most $\delta+1$. Such matrices are schematically depicted in \autoref{fig:banded-matrix-schematic}.

\begin{figure}
	\centering
\begin{tikzpicture}[diagonal/.style={color=gray!50!white, line width=2pt}]
\usetikzlibrary{intersections}
\pgfdeclarelayer{bg}
\pgfsetlayers{bg,main}

\matrix (m) [matrix of math nodes, nodes in empty cells, right delimiter={]}, left delimiter ={[}] {
	\phantom{0} & \phantom{0} & \phantom{0} & \phantom{0} & \phantom{0} & \phantom{0} & \phantom{0} & \phantom{0} & \phantom{0}\\
	\phantom{0} & \phantom{0} & \phantom{0} & \phantom{0} & \phantom{0} & \phantom{0} & \phantom{0} & \phantom{0} & \phantom{0}\\
	\phantom{0} & \phantom{0} & \phantom{0} & \phantom{0} & \phantom{0} & \phantom{0} & \phantom{0} & \phantom{0} & \phantom{0}\\
	\phantom{0} & \phantom{0} & \phantom{0} & \phantom{0} & \phantom{0} & \phantom{0} & \phantom{0} & \phantom{0} & \phantom{0}\\
	\phantom{0} & \phantom{0} & \phantom{0} & \phantom{0} & \phantom{0} & \phantom{0} & \phantom{0} & \phantom{0} & \phantom{0}\\
	\phantom{0} & \phantom{0} & \phantom{0} & \phantom{0} & \phantom{0} & \phantom{0} & \phantom{0} & \phantom{0} & \phantom{0}\\
	\phantom{0} & \phantom{0} & \phantom{0} & \phantom{0} & \phantom{0} & \phantom{0} & \phantom{0} & \phantom{0} & \phantom{0}\\
	\phantom{0} & \phantom{0} & \phantom{0} & \phantom{0} & \phantom{0} & \phantom{0} & \phantom{0} & \phantom{0} & \phantom{0}\\
	\phantom{0} & \phantom{0} & \phantom{0} & \phantom{0} & \phantom{0} & \phantom{0} & \phantom{0} & \phantom{0} & \phantom{0}\\
} ;

\matrix (sigprimeprime) [matrix anchor= north west, matrix of math nodes, nodes in empty cells, xshift=0.5cm] at (m.north east) {
	\phantom{0} \\ \phantom{0} \\ \phantom{0} \\ \phantom{0} \\ \vec{\sigma}'' \\ \phantom{0} \\ \phantom{0} \\ \phantom{0} \\ \phantom{0}\\
} ;

\matrix (sigprime) [matrix anchor=south west, matrix of math nodes, nodes in empty cells, yshift=2mm] at (m.north west) {
	\phantom{0} & \phantom{0} & \phantom{0} & \phantom{0} & \phantom{0} & \vec{\sigma}' & \phantom{0}  & \phantom{0} & \phantom{0}\\
} ; 

\matrix (sigma) [matrix anchor = north west, matrix of math nodes, nodes in empty cells, xshift=-0.54cm, yshift=-2mm] at (m.south west) {
	\phantom{0} & \phantom{0} & \phantom{0} &  \vec{\sigma} & \phantom{0}& \phantom{0} & \phantom{0} & \phantom{0} & \phantom{0}\\
} ;

\draw (sigprime.north west) rectangle (sigprime.south east);
\draw (sigprimeprime.north west) rectangle (sigprimeprime.south east);
\draw (sigma.north west) rectangle (sigma.south east);

\node[anchor=center, rotate=50, align=center] at($(m-5-6.south east)!0.4!(m-6-7.south east)$) {\small non-zero \\ entries};

\draw[dashed] (m-4-5.north west) rectangle (m-7-8.south east);

\draw[dashed] (m-4-5.north west) -- ($(sigprimeprime-4-1.north east)+(0.5,0)$);
\draw[dashed] (m-7-8.south west) -- ($(sigprimeprime-7-1.south east)+(0.5,0)$);

\draw[dashed] (m-4-8.north east) -- ($(sigprime-1-8.north east)+(-0.15,0.5)$);
\draw[dashed] (m-7-5.north west) -- ($(sigprime-1-5.north west)+(0,0.5)$);

\draw[dashed] (m-4-7.north east) -- ($(sigma-1-1.south west)-(0,0.5)$);
\draw[dashed] ($(m-6-8.north east)-(0,0.5)$) -- ($(sigma-1-4.south west)+(0.25,-0.5)$);

\draw [decorate,decoration={brace,amplitude=5pt,raise=4pt}]
($(sigprimeprime-4-1.north east)+(0.5,0)$)--($(sigprimeprime-7-1.south east)+(0.5,0)$) node [black,midway,xshift=1cm] {$\delta+1$};

\draw [decorate,decoration={brace,amplitude=5pt,raise=4pt}]
($(sigprime-1-5.north west)+(0.,0.5)$)--($(sigprime-1-8.north east)+(-0.15,0.5)$) node [black,midway,yshift=0.75cm] {$\delta+1$};

\draw [decorate,decoration={brace,amplitude=5pt,raise=4pt}]
($(sigma-1-4.south west)+(0.25,-0.50)$)--($(sigma-1-1.south west)-(0,0.5)$) node [black,midway,yshift=-0.5cm] {$\delta+1$};

\end{tikzpicture}
	\vspace{-0.5cm}
	
	\caption[Schematic representation of an lower anti-triangular matrix in which all non-zero entries are contained an an anti-banded submatrix with bounded bandwidth.]{A schematic representation of an lower anti-triangular non-negative  $(\rho+1)\times(\rho+1)$ matrix in which all non-zero entries are contained in an anti-banded $(\delta+1)\times(\delta+1)$ sub-matrix with bandwidth at most $(\delta+1)$.}\label{fig:banded-matrix-schematic}
\end{figure}
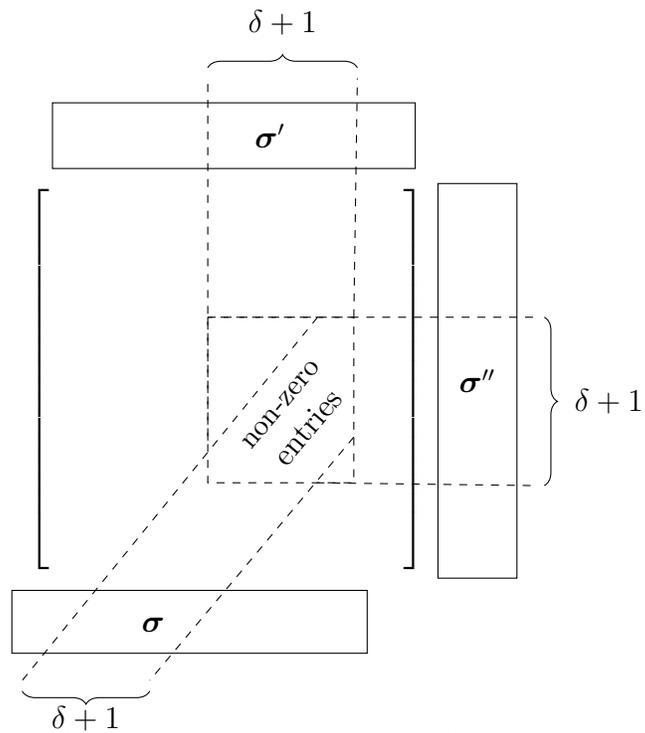

We can enumerate all valid matrix supports by first, selecting the upper-left corner of $Y$ as it sits in $X$. We must keep the $(\delta + 1) \times (\delta + 1)$ sub-matrix of non-zero entries within certain bounds, both to ensure that no entry above the main anti-diagonal is non-zero and to ensure that the sub-matrix does not exceed the boundaries of its parent matrix. The upper left corner of the sub-matrix must lie between diagonal $\rho - \delta$ and diagonal $2(\rho - \delta - 1)$. Thus, it is given by coordinates \[(a,b) \in \{(i,j) \in \{0,...,\rho\}^2 : \rho - \delta \leq i + j \leq 2(\rho - \delta-1)\}.\]  Next, we select the index of the first non-zero diagonal in $Y$, which we denote by $d$. Clearly, our choice of $d$ must lie between $1$ and $\delta+1$. However, we must also take care when selecting $d$ to ensure that no entry above the main anti-diagonal of $X$ is non-zero. This is ensured when $d + a+b \geq \rho+1$, which implies that we need to take \[d \in \{\max[1, \rho+1 - a - b],...,\delta+1\}.\] Notice that every choice of $a,b,$ and $d$ generates a unique set of non-zero cells of $X$. Moreover, we have previously determined that the sum of these non-zero cells must be non-negative integers which sum to $m$, that is, the non-zero cells form a weak composition\footnote{Recall that a \emph{weak composition of an integer $n$ into $k$ parts} is an ordered $k$-tuple of non-negative integers whose sum is $n$.} of $m$ into a number of parts equal to the number of cells available. More specifically, we can generate all matrices $X$ by, for each choice of $(a,b)$ and $d$, enumerating all weak compositions of $m$ into $B(d)$ parts, where $B(d)$ is the number of cells in anti-diagonals \mbox{$d, ..., d + \delta$} of a $(\delta+1)\times(\delta+1)$ matrix. Each such weak composition is placed in the non-zero entries of $X$. We can then sum the rows, columns, and anti-diagonals of $X$ to obtain the triple $(\vec{\sigma}, \vec{\sigma}', \vec{\sigma}'')$ for which $X$ is the matrix support.

The enumeration of possible matrices for $Y$ is easily accomplished using known loopless gray codes for weak compositions \cite{Savage1997}. However, to use these codes, we must have an exact closed-form for $B(d)$, which we now compute.

$B(d)$ is intimately related to the triangular numbers $T_i = \frac{i(i+1)}{2}$. It is not hard to see that $B(d) := (\delta + 1)^2 - T_{d-1} - T_{\delta +1 -d}$ for $d \in \{0,...,\delta+1\}$. This quantity is arrived at by taking the number of cells in a $(\delta+1)^2$ matrix and removing the entries which must be zero. If $d$ is the first non-zero anti-diagonal, then the $T_{d-1}$ entries of the matrix above diagonal $d$ are zero. Likewise, if $\delta+d$ is the last non-zero anti-diagonal then there are $T_{\delta+1 - d}$ entries below diagonal $\delta+d$ which must be zero. Routine simplification of $B(d)$ yields
\[
B(d) = \frac{3\delta^2 + \delta}{2} + d(\delta+2 - d) + 2.
\]
For our running time analysis, we will need a bound on $B(d)$ which it is convenient to describe now. Notice that the only portion of this equation which depends on $d$ is the term $d(\delta+2 - d)$, which we can upper bound by $\frac{1}{4}(\delta+2)^2$ by the AM-GM inequality\footnote{The AM-GM inequality states that $\sqrt{xy} \leq (x+y)/2$. Letting $y = n-x$ and squaring both sides, we have $x(n-x) \leq (x+n-x)^2/4 = n^2/4$. Setting $x = d$ and $n = \delta+2$ yields the claimed bound.}, thus, we have that
\begin{equation}\label{eq:diagonal}
B(d) \leq \frac{3\delta^2 + \delta}{2} + \frac{(\delta+2)^2}{4} + 2 = \frac{7\delta^2 + 6\delta}{4} + 3.
\end{equation}

With this formula in hand, we can use \autoref{alg:phi-table} to iterate over all possible matrices $Y$ and thereby obtain a table of values for $\Phi(\vec{\sigma}, \delta)$, for fixed $\delta$ and values of $\vec{\sigma}$ ranging over all possible signatures with skew at most $\delta$. Recall that $\repfact$ and $\delta$ are fixed, and given as input. The pseudocode of \autoref{alg:phi-table} has slight inefficiencies to preserve clarity, none of which affect the asymptotic bound on the running time.

\begin{algorithm}[h]
	\caption{Algorithm to fill table for $\Phi(\vec{\sigma}, \delta)$.}\label{alg:phi-table}
	\SetKwFunction{rowSum}{row-sum}
	\SetKwFunction{colSum}{col-sum}
	\SetKwFunction{diagSum}{diag-sum}
	Let $X$ be a $(\rho+1) \times (\rho+1)$ zero matrix\;
	\rowSum{$X$} returns the vector of row sums of matrix $X$\;
	\colSum{$X$} returns the vector of column sums of matrix $X$\;
	\diagSum{$X,i,j$} returns the vector of diagonal sums over diagonals $i,i+1,...,j$ of matrix $X$ for $1 \leq i < j \leq 2\repfact + 1$\;
	\For{$(a,b) \in \{(i,j) : \rho - \delta \leq i + j \leq 2(\rho - \delta - 1)\}$}{	
		\For{$d \gets \max[1, \rho + 1 - a - b], ...,\delta+1]$}{
			\For{each weak composition $C$ of $\m$ into $B(d)$ parts}{
				fill the non-zero entries of submatrix $Y$ according to the entries of $C$\;
				copy sub-matrix $Y$ into $X$, using $(a,b)$ as the top-left corner of $Y$ in $X$\;
				$\vec{\sigma}' \gets $\rowSum{$X$} \; 
				$\vec{\sigma}'' \gets $\colSum{$X$} \;
				$\vec{\sigma} \gets $\diagSum{$X,\repfact+1,2\repfact-1$}\; 
				$\Phi(\vec{\sigma}, \delta) \gets \Phi(\vec{\sigma}, \delta) \cup \{(\vec{\sigma}', \vec{\sigma}'')\}$\;
				reset matrix $X$ so every entry is zero\;
			}
		}
	}
\end{algorithm}

The body of the inner-most loop can be implemented to run in $O(\delta^2)$ time. The weak compositions of $m$ into $B(d)$ parts may each be visited in constant time per composition using a loopless gray code \cite{Savage1997}. Thus, the two innermost loops yield a total number of iterations given by the following formula.
\begin{align*}
\sum_{d=1}^{\delta+1} \binom{m + B(d) - 1}{B(d) - 1} &\leq \sum_{d=1}^{\delta+1} \binom{m + \frac{7\delta^2 + 6\delta}{4} + 2}{\frac{7\delta^2 + 6\delta}{4}+2} \\
&<(\delta+1)\left( \frac{e(m + \frac{7\delta^2 + 6\delta}{4} + 2)}{\frac{7\delta^2 + 6\delta}{4} + 2}   \right)^{\frac{7\delta^2 + 6\delta}{4} + 2}\\
&= (\delta+1)\left( \frac{4em}{7\delta^2 + 6\delta + 8}  +e \right)^{\frac{7\delta^2 + 6\delta}{4} + 2.}
\intertext{Where the first bound follows from our bound on $B(d)$, the second bound follows by Stirling's approximation, and the remainder is routine simplification. Combining it with the two outer-most loops, and the running time of the body of the inner-most loop, we obtain a running time of}
&O\left((\rho-\delta)^2\delta^3\left( \frac{4em}{7\delta^2 + 6\delta + 8}  +e \right)^{\frac{7\delta^2 + 6\delta}{4} + 2}\right).
\end{align*} 

We also use the above expression to bound the size of table generated by the algorithm. While each triple of \signature{s} may have multiple matrix supports, the problem of determining an exact number of triples appears to be a difficult one. In fact, mathematicians are still working to find a closed formula for the number of non-negative integer matrices with prescribed row and column sums \cite{Barvinok2012}. Moreover, this problem does not include any constraints on the diagonal sums, or the requirement that all entries sum to a given integer $m$.

In the next section we will see how this table is used in a dynamic program to find an optimal solution to \autoref{p:graph-multi-placement}.

\subsubsection{Bottom-Up Dynamic Program}\label{s:mp-dp}

Recall that $G_u(\vec{\sigma})$ is the value of an optimal \subplacement at node $u$ which has signature $\vec{\sigma}$. In this section, we present a recurrence which determines a table of values for $G_u(\vec{\sigma})$ at each node $u$. This recurrence is computed for each edge in order of a post-order traversal. For each node $u$, the first edge of $u$ which is visited by our algorithm first is termed an \emph{\up} edge, while all other edges are termed \emph{\out} edges.

During the algorithm, we compute several intermediate values of $G_u(\vec{\sigma})$ as edges connecting $u$ to its children are included. Suppose that node $u$ has children $c_1,...,c_\nChildren$. We refer to the subtree rooted at node $u$ and containing children $c_1,...,c_k$ and all of their descendants by $\subtree{u}{k}$. See \autoref{f:notation} for a diagram illustrating this notation. We define $G_u^k(\vec{\sigma})$ as the minimum objective value obtainable in $\subtree{u}{k}$ by a multi-placement with signature $\vec{\sigma}$. Once the table of values for $G_u^k(\vec{\sigma})$ is obtained, we can determine the values of $G_u^{k+1}(\vec{\sigma})$ using the recurrence described below. After every edge connecting $u$ to its children has been visited, we set $G_u(\vec{\sigma}) = G_u^t(\vec{\sigma})$.

\begin{figure}
	\centering
\begin{tikzpicture}
\tikzset{
	circnode/.style={circle,draw,minimum size=.6cm,inner sep=0},
	rectnode/.style={draw,minimum width=.25cm,fill=white,minimum height=.6cm,child anchor=north,anchor=north},
}

\Tree[.\node {};
[.\node [circnode] (u) {$u$};	\edge node  [auto=right,pos=0.4] {up};
[.\node [circnode]{$\child{1}$}; 
[.\node (start) [circnode,minimum size=0.25cm] {}; ]
[.\node [circnode, minimum size=0.25cm] {}; ]
[.\node [circnode, draw=none] {$...$}; ]			
[.\node (end)  [circnode, minimum size=0.25cm] {}; ]			
]
\edge node [auto=right,pos=0.95,xshift=6pt, yshift=-5pt] {out};
[.\node [circnode] {$\child{2}$};
[.\node (start2) [circnode,minimum size=0.25cm] {}; ]
[.\node [circnode, minimum size=0.25cm] {}; ]
[.\node [circnode, draw=none] {$...$}; ]			
[.\node (end2)  [circnode, minimum size=0.25cm] {}; ]
]
\edge node  [auto=right,pos=0.5,xshift=7pt,yshift=2pt] {out};
[.\node [circnode] {$\child{3}$};	
[.\node (start3) [circnode,minimum size=0.25cm] {}; ]
[.\node [circnode, minimum size=0.25cm] {}; ]
[.\node [circnode, draw=none] {$...$}; ]			
[.\node (end3)  [circnode, minimum size=0.25cm] {}; ]
]
\edge node [auto=left,midway,pos=0.4] {out};
[.\node [circnode] {$\child{4}$}; 	]
]
];

\draw[thick,dashed] ($(start.west |- u.north)+(-0.2,0.2)$) rectangle ($(end3.south east)+(0.05,-0.2)$);

\node at ($(start.west |- u.north)+(0,0.5)$) {$\subtree{u}{3}$};

\end{tikzpicture}
	\caption{Diagram illustrating our use of notation. Nodes within the dashed box comprise subtree $\subtree{u}{3}$, along with all of their descendants.}\label{f:notation}
\end{figure}
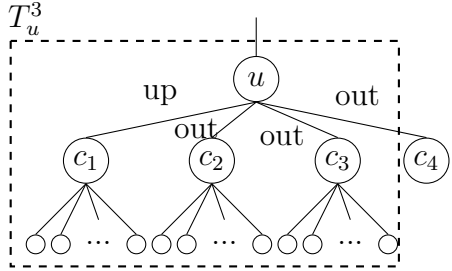

The recurrence we describe for $G_u^k(\vec{\sigma})$ is comprised of three cases, the leaf case, up case, and out case, stated below. The up and out cases handle the inclusion of up and out edges respectively, while the leaf case forms the base case of the recursion. We will justify the recurrence immediately after its statement.
\begin{align*}
G_u^1(\vec{\sigma}) &= G_{c_1}(\vec{\sigma}) + \vec{\sigma} &&\text{(up case)}\\
G_u^{k+1}(\vec{\sigma}) &= \min_{\vec{\sigma}', \vec{\sigma}'' \in \Phi(\vec{\sigma}, \delta)} \left[ G_u^k(\vec{\sigma}') + G_{c_{k+1}}(\vec{\sigma}'') + \vec{\sigma} - \vec{\sigma}' \right] &&\text{(out case)}
\intertext{The base case occurs for a leaf $\ell$, with capacity $c(\ell)$, in which}
G_\ell(\vec{\sigma}) &= \left\{\begin{array}{@{}ll@{}}
\vec{\sigma} ~~~~& \text{if } \vec{\sigma}_i \leq 1 \text{ for all $i$ and }  \sum_{i=0}^{\repfact} i\vec{\sigma}_i \leq c(\ell)\\
\infty       ~~~~& \text{otherwise}
\end{array}\right.&&\text{(leaf case)}
\end{align*}

We justify this recurrence as follows. Consider first the leaf case. Since each placement is a subset of leaves, no placement may include more than one replica at any given leaf node, justifying the upper bound on values of $\vec{\sigma}_i$. For a leaf node $\ell$, the only possible multi-placements are those which use no more capacity than $\ell$ has been allotted. The capacity used by a multi-placement with \signature $\vec{\sigma}$ is easily seen to be $\sum_{i=0}^\repfact i \vec{\sigma}_i$. When the \signature exceeds the capacity, no optimal value exists, which we represent by $\infty$, defined to be lexicographically larger than any vector.

In the up case, the only leaves available are those under child $c_1$. In this case, all we must do is include the additional contribution of node $u$ to the optimal value computed for $c_1$. Since the signature of the placement at child $c_1$ is $\vec{\sigma}$, by \autoref{lem:signature-solution}, node $u$ contributes an additional factor of $\vec{\sigma}$ to the optimal solution.

To help explain the out case, we appeal to the illustration in \autoref{f:out-case-explained}. Intuitively, we are merely splitting the signature $\vec{\sigma}$ among the subtrees $T_u^k$ and child $c_{k+1}$ in the optimal way. Since $\Phi(\vec{\sigma}, \delta)$ contains all signature pairs which can be combined to form $\vec{\sigma}$, taking the minimum over all such signature pairs yields the best possible combination. Finally, we must adjust the value to account for the fact that the term $G_u^k(\vec{\sigma})$ includes a contribution of $\vec{\sigma}'$ from node $u$, which is now inaccurate. Adding the multi-placement of \signature $\vec{\sigma}''$ has increased the contribution of node $u$ to $\vec{\sigma}$. We account for this increase by including the correction factor $\vec{\sigma} - \vec{\sigma}'$, which removes the incorrect value and includes the correct value in its place.

\begin{figure}
	\centering
\begin{tikzpicture}
\tikzset{
	circnode/.style={circle,draw,minimum size=.7cm,inner sep=0},
	smallnode/.style={circle,draw,minimum size=.25cm,inner sep=0},
	rectnode/.style={draw,minimum width=.25cm,fill=white,minimum height=.7cm,child anchor=north,anchor=north},
}

\Tree[.\node {};
[.\node [circnode] (u) {$u$};
[.\node (start) [smallnode] {}; ]
[.\node [smallnode] {}; ]
[.\node [smallnode, draw=none] {$...$}; ]			
[.\node (end)  [smallnode] {}; ]
]
]
];

\node [circnode, right=1.25cm of end] (ck) {$c_k$};
\node [smallnode, below=0.25cm of ck] (b) {};
\node [smallnode, right=0.25cm of b] (a) {};
\node [smallnode, left=0.25cm of b] (c) {};
\node [smallnode, below left =0.25cm of b] (e) {};
\node [smallnode, below right =0.25cm of b] (f) {};
\node [smallnode, below right =0.25cm of a] (g) {};

\draw (u.south) -- (ck.north);
\draw (b.north) -- (ck);
\draw (a.north) -- (ck);
\draw (c.north) -- (ck);
\draw (b) -- (e.north);
\draw (b) -- (f.north);
\draw (a) -- (g.north);

\draw [decorate,decoration={brace,amplitude=5pt, mirror},yshift=-5pt]
(start.south) -- (end.south) node [midway, yshift=-10pt] {\footnotesize $k - 1$};

\draw[thick,dotted] ($(start.west |- u.north)+(-0.3,0.3)$) rectangle ($(end.south east)+(0.2,-0.55)$);
\draw[thick,dotted] ($(c.west |- ck.north)+(-0.3,0.3)$) rectangle ($(g.south east)+(0.2,-0.2)$);

\node at ($(start.west |- u.north)+(0,0.7)$) (anchor) {$\subtree{u}{k-1}$};

\node at ($(c.west |- ck.north)+(0,0.7)$) {$\subtree{\child{k}}{}$};

\node at ($(u)+(0,-2.1)$) {signature: $\vec{\sigma}'$};
\node at ($(ck)+(0,-1.7)$) (anchor2){signature: $\vec{\sigma}''$};
\draw [thick, dashed] ($(anchor)+(-0.8,0.4)$) rectangle ($(anchor2 -| g.south east)+(0.4,-0.3)$);
\node at ($(anchor)+(-0.8,0.8)$) {$\subtree{u}{k}$};
\node[align=center] at ($(anchor2 -| g.south east)+(-2.1,-1)$) {signature: $\vec{\sigma}$\\ $(\vec{\sigma}', \vec{\sigma}'') \in \Phi(\vec{\sigma}, \delta)$};

\end{tikzpicture}
	\caption{Diagram illustrating how out edges are added to the tree.}\label{f:out-case-explained}
\end{figure}
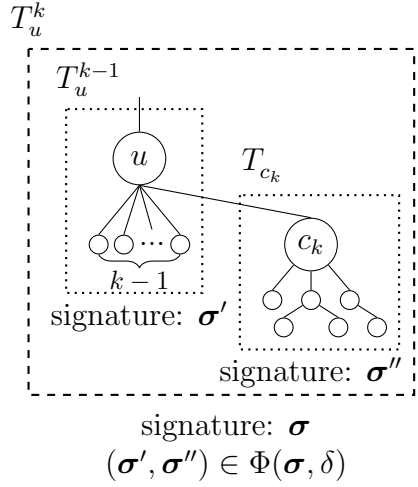

One minor detail is not accounted for in the above recurrence. At no point should we report that an optimal multi-placement with signature $\vec{\sigma}$ exists unless the leaves of $T_u^k$ have sufficient capacity to store such a multi-placement. This is easily rectified by defining $\G_u^k = \infty$ when $\sum_{i=0}^\repfact i\vec{\sigma}_i > c(T_u^k)$, where $c(T^u_k)$ is defined as the sum of capacities of leaves in the subtree rooted at $T_u^k$. These capacities can be easily obtained as a preprocessing step. Recall that, by convention, we take $\infty$ to be lexicographically larger than any vector.

Using the above recurrence it is a simple matter to compute the value of an optimal multi-placement with a particular signature $\vec{\sigma_0}$ with skew bounded by $\delta$. Simply compute, bottom up, $G_u(\vec{\sigma})$ for all values of $\vec{\sigma}$ which have skew at most $\delta$. Once the root node has had the value $\G_u(\vec{\sigma_0}$ filled in, the computation can stop. Of course, we are not merely interested in the value of the optimal solution, we must produce the optimal solution itself. The multi-placement itself is easily obtained by storing the pair of signatures which results in the minimum value for each out-case above. This record of the computation is all that is needed to obtain a multi-placement which has the given signature.

Finally, we note that once the table $G_u^{k+1}(\vec{\sigma})$ has been entirely filled out, the table containing values of $G_u^k$ is no longer needed, and can be overwritten by values of $G_u^{k+1}$. In this way, only two tables for $G_u$ ever need to be stored for each node $u$ throughout the procedure. Moreover, the second table can be discarded or reused after the table for $G_u^t$ has been recorded. In total, only $n+1$ tables need to be stored for a complete run of the dynamic program on a tree with $n$ nodes.

The running time of the entire procedure can be bounded as follows. Filling out a single table for $G_u^{k}$ takes, at most, time bounded by the time taken to build the table of values for $\Phi$. For each table, we compare and update the objective function, which takes $O(\rho)$ time. Finally, we only perform an update of the table for $G_u$ once for each edge $(u,v)$ in the tree, so, summing over all $\n-1$ edges of the tree, we obtain an overall running time of
\[O\left(\n\repfact(\repfact-\delta)^2\delta^3\left( \frac{4em}{7\delta^2 + 6\delta + 12}  +e \right)^{\frac{7\delta^2 + 6\delta}{4} + 3}\right).\]

This establishes that \autoref{p:graph-multi-placement} can be solved in polynomial time for fixed values of $\delta$.

\section{NP-hardness of Single-block Replica Placement in Bipartite Graphs}
\label{s:np-hard}

We could also consider replica placement in directed graphs which are more general than trees. To this end, we can define the following extension of \autoref{p:graph-single-placement} by replacing the arborescence by a directed graph, and the set of leaves by a set of \emph{candidate placement} nodes.
\begin{problem}[Optimal Single-block Placement in Graphs]
	\label{p:graph-sp}
	Given a directed graph $\graph{\V}{\A}$, a set of candidate nodes $C \subseteq \V$, and positive integer $\repfact$, with $\repfact < |C|$, find a placement $\P \subset C$ with size $\repfact$, such that $\ff{\P}$ is lexico-minimum.
\end{problem}
A similar extension of \autoref{p:graph-multi-placement} is immediate. Note that since the failure number was defined with regards to the reachability relation, no changes are necessary to ensure that the concepts of failure number and failure aggregate remain well-defined.

We sketch an argument that \autoref{p:graph-sp} is NP-hard by reduction from \textsc{Dominating Set} \cite{Garey1979} as follows. The input to the \textsc{Dominating Set} problem is an undirected graph $G$, and an integer $k$. The question asked is: ``is there a subset $D$ of exactly $k$ vertices from $G$ such that every vertex not in $?D$ is adjacent to at least one vertex in $D$?" To answer this question, we can form an instance of \autoref{p:graph-sp} by constructing a bipartite graph $(S,T,E)$, where there is a vertex in $S$ and in $T$ for each vertex in $G$. For each edge $(u,v)$ \emph{not included in} $G$, we connect the representative of $u$ in $S$ with the representative of $v$ in $T$ by an edge directed from $S$ to $T$. Thus the neighborhood of each node $u$ in $S$ is the set of vertices which are \emph{not} adjacent to $u$ in $G$. The bipartite graph $(S,T,E)$ then forms a directed graph as required by the input to \autoref{p:graph-sp}. Let $T$ be the set of candidate nodes. Then, a dominating set of size $k$ exists in $G$ if and only if a placement of size $k$ can be made on the nodes of $T$ such that $\ff{\P} \lleq \langle 0, |\V|, ..., |\V| \rangle.$ The key to the reduction is the zero in the first entry of the vector, which counts the number of nodes with failure number $k$. In a placement $P$, if node $u$ in $S$ has failure number $k$, it indicates that every node in the placement (i.e. the purported dominating set) is adjacent to $u$ \emph{in the complement of $G$}. Equivalently, a node $u$ has a failure number $k$ only if it is is \emph{not} adjacent to any element of $P$. Thus, a placement on $T$ which has no nodes with failure number $k$ must be a dominating set in $G$. The remaining direction is easily shown.

This reduction shows that \autoref{p:graph-sp} remains NP-hard even when the input is restricted to \emph{bipartite} graphs. It additionally shows that \emph{even minimizing the first entry of the failure aggregate} is an intractable problem for bipartite graphs. One may wonder whether intractability can be circumvented by instead lexicominimizing a suffix of the failure aggregate. Unfortunately, this does not work. An easy reduction from \textsc{Independent Set} \cite{Garey1979} shows that it is also NP-hard to lexico-minimize even the last three entries of the failure aggregate.

Both \textsc{Independent Set} and \textsc{Dominating Set} are special cases of the more general \textsc{Set Cover} \cite{Garey1979} problem, for which an $\Omega(\log n)$ lower bound on the approximation ratio is well known \cite{Dinur2014}. While \textsc{Set Cover} does reduce to \autoref{p:graph-sp}, whether this lower-bound can be made to carry over remains an open question. Indeed, the very notion of what it means to \emph{approximate} a vector quantity \emph{in the lexicographic order} is an interesting one which, to the best of our knowledge, has yet to be explored in the complexity literature.

\section{Future Work}\label{s:future-work}

In this paper, we have described two algorithms for solving two replica placement problems in trees. The first problem considers how best to select a placement of size $\repfact$ on which to place replicas of a single block of data. For this problem, we present an $O(n + \repfact \log \repfact)$ algorithm. The second problem considers how to optimally select a multi-placement of $\m$ placements. For this problem, we present an algorithm which runs in polynomial time when the skew of the multi-placement is at most a fixed constant.

We first propose studying approximations to \autoref{p:graph-sp}, including notions of approximation for lexicographic optimization problems. In addition, we propose pursuing exact algorithms for restricted classes of bipartite graphs such as multi-trees, and bipartite graphs with adjacency matrices which have the consecutive ones property. Such graphs can also be used to model failure in data centers. We also plan to investigate several variants of weighted objectives for both the single-block and multi-placement problems which tie together our approach with that considered in \cite{Korupolu2016}. 

While the complexity of multi-placement optimization is currently unknown, we do know that the running time can be improved in special cases. For instance, when the skew is at most $1$, we can attain a running time of $O(\n\m^2\repfact^3)$ by exploiting some additional balancing properties.
Interesting directions for future work on multi-placements include establishing fixed-parameter tractability of the multi-placement problem and proving NP-hardness. It would also be interesting to explore whether the exponent can be brought down to some term which is $o(\delta^2)$. The matrix system described in \autoref{s:combo-signatures} is an important step towards this goal.

\section*{Acknowledgements}
The authors would like to thank Ian Cook, Conner Davis and Balaji Raghavachari for insightful conversations and comments on drafts.


\section*{References}
\bibliographystyle{elsarticle-num}      
\bibliography{bibliography}   

\appendix

\section{Proof of Correctness of the Greedy Algorithm for the Single Placement Problem}\label{app:greedy-proof}

To establish the correctness of the greedy algorithm, we first introduce some concepts and notation.  We now consider partial placements possibly containing fewer than $\repfact$ leaf nodes. 
Given a partial placement $P \subseteq C$ and a subset of nodes $S \subseteq V$, let $\ef(S,P) = \langle p_{0}, p_{1}, \ldots, p_{\repfact} \rangle$, where $p_i \defined | \{ x \in S \cap E \mid f(x,P) = \repfact - i \} |$. 
Note that $\ef$ differs from $\f$ in that $\ef$ gives the failure aggregate for a given subset of nodes, whereas $\f$ gives the aggregate for the entire set of nodes.
Thus, $\f(P)$ and $\ef(V,P)$ are the same vectors.  

Unless otherwise stated, in this section, we only consider vectors of size $\repfact + 1$.

Given a vector $\vec{a} = \langle a_0, a_1 \ldots, a_{\repfact} \rangle$, we denote the vector shifted one index to the left as $\reallywidehat{\vec{a}}$ and is given by
$\langle a_1, a_2, \ldots, a_{\repfact}, 0 \rangle$.
The following two propositions about $\reallywidehat{~~}$ operator are used to prove the correctness of the greedy algorithm.

\begin{proposition}
	\label{prop-vector-leftshift}
	Given a vector $\vec{a}$ with $\vec{a}[0] = 0$ (i.e., the leftmost entry of $\vec{a}$ is 0), we have:
	$$\vec{0} \lleq \vec{a} \iff \vec{a} \lleq \reallywidehat{\vec{a}}$$
\end{proposition}

\begin{proposition}
	\label{prop-leftshift-diff}
	Given vectors $\vec{a}$ and $\vec{b}$, we have:
	$$\reallywidehat{\vec{a}} - \reallywidehat{\vec{b}} = \widehat{\vec{a}-\vec{b}}$$
\end{proposition}

We denote the set of nodes on a path from node $u$ to node $v$ as $u \rightsquigarrow v$, both $u$ and $v$ inclusive. 
The following three equations relate $\f$, $\ef$ and $\reallywidehat{\ef}$ vectors. In the equations, 
$r$ is the root node and $u$ is a leaf node. For \eqref{eq-g-add} and \eqref{eq-g-leftshift}, assume that $|P| < \repfact$.

\begin{align}
\label{eq-f-g}
\f(P) & = \ef(E,P) \\
\label{eq-g-add}
\ef(E, P \cup \{u\}) & = \ef(E, P) - \ef(r \rightsquigarrow u, P) + \ef(r \rightsquigarrow u, P \cup \{u\}) \\
\label{eq-g-leftshift}
\ef(r \rightsquigarrow x, P \cup \{ u \}) & = \reallywidehat{\ef}(r \rightsquigarrow u, P)
\end{align}

We first prove two lemmas that are used to prove the correctness of the greedy algorithm.

\begin{lemma}\label{lem-path-ineq}
	Given a partial placement $P \subseteq C$ with $|P| < \repfact$ and
	candidate nodes $u,v \in C - P$, we have:
	$$\ef(r\rightsquigarrow u, P) \lleq \ef(r\rightsquigarrow v, P) \iff \f(P \cup \{u\}) \lleq \f(P \cup \{v\})$$
\end{lemma}

\begin{proof}
	\begin{align*}
	\f(P \cup \{u\}) & \lleq \f(P \cup \{v\}) \\
	& \iff \{ \text{using \eqref{eq-f-g}} \} \\
	\ef(E, P \cup \{u\}) & \lleq \ef(E, P \cup\{v\}) \\
	& \iff \{ \text{using \eqref{eq-g-add}} \} \\
	\left(\begin{array}{@{}rcl@{}}
	\ef(E,P) & - & \ef(r \rightsquigarrow u, P) \\ & + & \ef(r \rightsquigarrow u, P \cup \{u\}) 
	\end{array}
	\right)
	& \lleq 
	\left(\begin{array}{@{}rcl@{}}
	\ef(E,P) & - & \ef(r \rightsquigarrow v, P) \\
	& + & \ef(r \rightsquigarrow y, P \cup \{v\})
	\end{array}
	\right) \\
	& \iff \{ \text{rearranging terms} \} \\
	\ef(r \rightsquigarrow v, P) - \ef(r \rightsquigarrow u, P) 
	& \lleq 
	\ef(r \rightsquigarrow v, P \cup \{v\}) - \ef(r \rightsquigarrow u, P \cup \{u\}) \\
	& \iff \{ \text{using \eqref{eq-g-leftshift}} \} \\
	\ef(r \rightsquigarrow v, P) - \ef(r \rightsquigarrow u, P) 
	& \lleq
	\reallywidehat{\ef}(r \rightsquigarrow v, P) - \reallywidehat{\ef}(r \rightsquigarrow u, P) \\
	& \iff \{ \text{using Proposition \ref{prop-leftshift-diff}} \} \\
	\ef(r \rightsquigarrow v, P) - \ef(r \rightsquigarrow u, P) 
	& \lleq
	\reallywidehat{\ef(r \rightsquigarrow v, P) - \ef(r \rightsquigarrow u, P)} \\
	& \iff \{ \text{let } \vec{h_{u,v}}(P) = \ef(r \rightsquigarrow v, P) - \ef(r \rightsquigarrow u, P) \}\\
	\vec{h_{u,v}}(P) & \lleq \widehat{\vec{h_{u,u}}}(P) \\
	& \iff \{ \text{using Proposition \ref{prop-vector-leftshift}} \} \\
	\vec{0} & \lleq \vec{h_{u,v}}(P)  \\
	& \iff \{ \text{using definition of } \vec{h_{u,v}} \} \\
	\vec{0} & \lleq \ef(r \rightsquigarrow v, P) - \ef(r \rightsquigarrow u, P) \\
	& \iff \{ \text{rearranging terms} \} \\
	\ef(r \rightsquigarrow u, P) & \lleq \ef(r \rightsquigarrow v, P)
	\end{align*}
	This establishes the lemma.
	\qed
\end{proof}


\begin{lemma}\label{lem-technical}
	Consider a partial placement $P \subseteq C$ and two distinct candidate nodes $u,v \in C - P$. Let $a_{u,v}$ be the least common ancestor of $u$ and $v$, and 
	let $a_u$ be the child of $a_{u,v}$ on the path from $a_{u,v}$ to $u$ (see Fig. \ref{fig:thm-1}). Consider a subset of candidate nodes $S \subseteq C$ such that 
	\begin{inparaenum}[(i)]
		\item $|P \cup S| < \repfact$,
		\item $\Csub{a_u} \cap S = \emptyset$ and 
		\item $v \notin S$.
	\end{inparaenum}
	Then,
	$$\ef(r\rightsquigarrow u, P) \lleq \ef(r\rightsquigarrow v, P) \implies \f(P \cup S \cup \{u\})  \lleq \f(P \cup S \cup \{v\})$$ 
\end{lemma}

\begin{proof}
	Since $\Csub{a_u} \cap S = \emptyset$, any increase in failure numbers due to the addition of replicas in $S$ to $P$ cannot effect nodes on the path 
	$a_u \rightsquigarrow v$. Thus we have that:
	\begin{equation}\label{eqn-greedy-lemma-1}
	\ef(r \rightsquigarrow u, P) \lleq \ef(r \rightsquigarrow v, P) \implies \ef(r \rightsquigarrow u, P \cup S) \lleq \ef(r \rightsquigarrow v, P \cup S)
	\end{equation}
	Applying Lemma \ref{lem-path-ineq} using $P \cup S$ as the placement, we have:
	\begin{align}
	\label{eqn-greedy-lemma-2}
	\begin{array}{@{}c@{}}
	\ef(r \rightsquigarrow u, P \cup S) \lleq \ef(r \rightsquigarrow v, P \cup S) \\ \iff  \\ \f(P \cup S \cup \{u\})  \lleq \f(P \cup S \cup \{v\})
	\end{array}
	\end{align}
	Combining \eqref{eqn-greedy-lemma-1} and \eqref{eqn-greedy-lemma-2}, we obtain the lemma.
	\qed
\end{proof}

\begin{figure}[tb]
	\centering
	\begin{tikzpicture}[scale=0.5, transform shape]
	\tikzstyle{every node}=[minimum size=1cm, align=center, font=\LARGE]
	\tikzset{>=latex}
	\tikzset{snake it/.style={decorate, decoration=snake},
		postaction={decoration={markings,mark=at position 1 with {\arrow{>}}},decorate}}
	
	\pgfmathsetmacro{\xoffset}{1.5}
	\pgfmathsetmacro{\yoffset}{1.5}
	
	\node (auv)[circle, draw=black] at    (0, 0) {$a$};
	\node  (av)[circle, draw=black]    at (\xoffset,      -\yoffset) {$a_v$};
	\node  (au)[circle, draw=black]    at (-\xoffset,     -\yoffset) {$a_u$};
	\node   (u)[rectangle, draw=black] at (-2*\xoffset, -3*\yoffset) {$u$};
	\node   (v)[rectangle, draw=black] at (2*\xoffset,  -3*\yoffset) {$v$};

	\draw[->]       (auv) -- (av);
	\draw[-<]       (auv) -- (au);
	\draw[snake it] (au) -- (u);
	\draw[snake it] (av) -- (v);
	\draw[snake it] (1,1) -- (auv);
	\draw[->]       (v) -> (u) node [midway, fill=white] {swap};
	\end{tikzpicture}
	
	\caption{Named nodes used in Theorem \ref{thm-greedy}. The arrow labeled ``swap'' illustrates the leaf nodes between which replicas are moved, and is not an edge of the graph.}\label{fig:thm-1}
\end{figure}
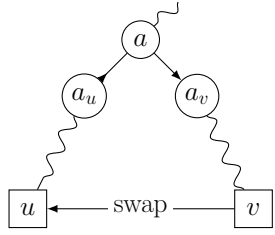

We are now in a position to prove the correctness of the greedy algorithm.

\begin{theorem}\label{thm-greedy}
	Let $P_i$ be the partial placement from step $i$ of the greedy algorithm. Then there exists an optimal placement $P^\ast$, with $|P^\ast| = \repfact$ such that $P_i \subseteq P^\ast$.
\end{theorem}

\begin{proof}
	The proof proceeds by induction on $i$. $P_0 = \emptyset$ is clearly a subset of any optimal solution. Given $P_i \subseteq P^\ast$ for some optimal solution $P^\ast$, we must show that there is an optimal solution $Q^\ast$ for which $P_{i+1} \subseteq Q^\ast$. Clearly, if $P_{i+1} \subseteq P^\ast$, then we are done, 
	since $P^\ast$ is optimal. In the case where $P_{i+1} \not\subseteq P^\ast$ we must exhibit some optimal solution $Q^\ast$ for which $P_{i+1} \subseteq Q^\ast$. Let $u$ be the leaf which was added to $P_i$ to form $P_{i+1}$. Let $v$ be the leaf in $P^\ast - P_{i+1}$ which has the greatest-depth least common ancestor with $u$, where the depth of a node is given by its distance from the root (see Fig. \ref{fig:thm-1}). 
	We set $Q^\ast = (P^\ast - \{v\}) \cup \{u\}$, and claim that $\f(Q^\ast) \lleq \f(P^\ast)$. Since $\f(P^\ast)$ is optimal, and $P_{i+1} \subseteq Q^\ast$ this will complete our proof.
	
	By our greedy choice of $u$,  $\f(P_i \cup \{u\}) \lleq \f(P_i \cup \{v\})$. Using Lemma \ref{lem-path-ineq}, it follows that:
	\begin{align} \label{eqn-crux-thm2}
	\ef(r \rightsquigarrow u, P_i) \lleq \ef(r \rightsquigarrow v, P_i)
	\end{align}
	
	Note that $u,v \notin (P^\ast - P_i - \{v\})$. Moreover, by our choice of $v$, we have that $\Csub{a_u} \cap (P^\ast - P_i - \{v\}) = \emptyset$, since the only nodes from $P^\ast$ in $\Csub{a_u}$ must also be in $P_i$. To complete the proof, we use \eqref{eqn-crux-thm2} and apply Lemma \ref{lem-technical}, setting \mbox{$S = P^\ast - P_i - \{v\}$}. This choice of $S$ is made so as to yield the following equalities:
	$$Q^\ast = (P^\ast - \{v\}) \cup \{u\} = P_i \cup (P^\ast - P_i - \{v\}) \cup \{u\},$$ and
	$$P^\ast = P_i \cup (P^\ast - P_i - \{v\}) \cup \{v\}. $$
	
	By Lemma \ref{lem-technical}, we obtain the following:
	$$\f(Q^\ast) = \f(P_i \cup (P^\ast - P_i - \{v\}) \cup \{u\}) \lleq \f(P_i \cup (P^\ast - P_i - \{v\}) \cup \{v\}) = \f(P^\ast).$$
	This establishes the theorem.
	\qed 
\end{proof}

\section{Proof of Correctness of the Labeling Algorithm in the Divide Phase}\label{app:labeling-proof}

Given a subset of children $X$, let $\min(X)$ (respectively, $\max(X)$) 
denote the smallest (respectively, largest) capacity of any child in $X$. In case $X 
= \emptyset$, we set $\min(X) = \max(X) = 0$. 
Also, let $\sigma(X)$ denote the sum of the capacities of all children in $X$. 
Note that the labeling algorithm only needs to be run if the sum of the capacities of all the children is strictly greater than the number of replicas that need to be placed on the children. This implies that at least one child will be labeled as unfilled.

Assume that the while loop is executed $I$ times. Let $F^{(i)}$, $U^{(i)}$, $M^{(i)}$ denote the values of $F$, $U$ and $M$ at the end of iteration $i$ of the while loop with $F^{(0)} = \emptyset$, $U^{(0)} = \emptyset$ and $M^{(0)} = \{ 1, \ldots, t \}$. 
We first show the following lemma.

\begin{lemma}
	\label{lem:while:invariants}
	For each $i$ with $0 \leq i \leq I$, we have:
	\begin{align}
	\label{eq:while:lower}
	\max(F^{(i)}) \cdot (|U^{(i)}| + |M^{(i)}|) \; \leq \; & r - \sigma(F^{(i)}) \\
	\label{eq:while:upper}
	& r - \sigma(F^{(i)} \cup M^{(i)}) \; < \; \min(U^{(i)}) \cdot |U^{(i)}| 
	\end{align}
\end{lemma}
\begin{proof}
	The proof is by induction on $i$. Clearly, \eqref{eq:while:lower} and \eqref{eq:while:upper} hold for $i=0$. Assuming, by induction hypothesis, 
	that \eqref{eq:while:lower} and \eqref{eq:while:upper} hold for $i = j$ for some $j$ in $[0,I)$. We now prove that \eqref{eq:while:lower} and \eqref{eq:while:upper} also hold for $i = j+1$. We use $M_{l}^{(i)}$, $M_{e}^{(i)}$, $M_{g}^{(i)}$ and $\med^{(i)}$ to denote the values computed for the variables $M_{l}$, $M_{e}$, $M_{g}$ and $\med$, respectively, during the iteration $i$. There are three cases depending on the value of $\frac{r - \sigma(F^{(j)} \cup M_{l}^{(j)})}{|U^{(j)}| + |M_{e}^{(j)}| + |M_{g}^{(j)}|}$.
	
	\begin{enumerate}[{Case} 1]
		
		\item \{$\left(\frac{r - \sigma(F^{(j)} \cup M_{l}^{(j)})}{|U^{(j)}| + |M_{e}^{(j)}| + |M_{g}^{(j)}|} < \med^{(j)} - 1 \right)$\}: In this case, $F^{(j+1)} = F^{(j)}$, $U^{(j+1)} = U^{(j)} \cup M_{e}^{(j)} \cup M_{g}^{(j)}$ and $M^{(j+1)} = M_{l}^{(j)}$. We consider the two inequalities separately.
		
		\begin{enumerate}
			\setlength{\itemindent}{10pt}
			
			\item To show that \eqref{eq:while:lower} holds for $i = j+1$:
			\begin{align*}
			r - \sigma(F^{(j+1)})  
			& = r - \sigma(F^{(j)})  \\
			& \geq \max(F^{(j)}) \cdot (|U^{(j)}| + |M^{(j)}|) \\
			& = \max(F^{(j+1)}) \cdot (|U^{(j)}| + |M^{(j)}|) \\
			& = \max(F^{(j+1)}) \cdot (|U^{(j+1)}| + |M^{(j+1)}|)
			\end{align*}
			
			\item To show that \eqref{eq:while:upper} holds for $i = j+1$:
			\begin{align*}
			r - \sigma(F^{(j+1)} \cup M^{(j+1)}) 
			& = r - \sigma(F^{(j)} \cup M_{l}^{(j)}) \\
			& < (\med^{(j)} - 1) \cdot (|U^{(j)}| + |M_{e}^{(j)}| + |M_{g}^{(j)}|) \\
			& = (\med^{(j)} - 1) \cdot |U^{(j+1)}| \\
			& < \min(U^{(j+1)}) \cdot |U^{(j+1)}|
			\end{align*}
			
		\end{enumerate}
		
		\item \{$\left(\frac{r - \sigma(F^{(j)} \cup M_{l}^{(j)})}{|U^{(j)}| + |M_{e}^{(j)}| + |M_{g}^{(j)}|} \geq \med^{(j)}\right)$\}:  In this case, $F^{(j+1)} = F^{(j)}  \cup M_{l}^{(j)} \cup M_{e}^{(j)}$, $U^{(j+1)} = U^{(j)}$ and $M^{(j+1)} = M_{g}^{(j)}$. We consider the two inequalities separately.

		\begin{enumerate}
			\setlength{\itemindent}{10pt}
			
			\item To show that \eqref{eq:while:lower} holds for $i=j+1$:
			\begin{align*}
			r - \sigma(F^{(j+1)})  
			& = r - \sigma(F^{(j)} \cup M_{l}^{(j)} \cup M_{e}^{(j)})  \\
			& = r - \sigma(F^{(j)} \cup M_{l}^{(j)}) - \med^{(j)} \cdot |M_{e}^{(j)}|  \\
			& \geq \med^{(j)} \cdot (|U^{(j)}| + |M_{e}^{(j)}| + |M_{g}^{(j)}|) - \med^{(j)} \cdot |M_{e}^{(j)}| \\
			& =  \med^{(j)} \cdot (|U^{(j)}| + |M_{g}^{(j)}|) \\
			& = \max(F^{(j)})  \cdot (|U^{(j)}| + |M_{g}^{(j)}|) \\
			& =  \max(F^{(j)})  \cdot (|U^{(j+1)}| + |M^{(j+1)}|)
			\end{align*}
			
			\item To show that \eqref{eq:while:upper} holds for $i=j+1$:
			\begin{align*}
			r - \sigma(F^{(j+1)} \cup M^{(j+1)}) 
			& = r - \sigma(F^{(j)} \cup M_{l}^{(j)} \cup M_{e}^{(j)} \cup M_{g}^{(j)}) \\
			& = r - \sigma(F^{(j)} \cup M^{(j)}) \\
			& < \min(U^{(j)}) \cdot |U^{(j)}| \\
			& = \min(U^{(j+1)}) \cdot |U^{(j+1)}|
			\end{align*}

		\end{enumerate}
		
		\item \{$\left( \med^{(j)}  - 1 \leq \frac{r - \sigma(F^{(j)} \cup M_{l}^{(j)})}{|U^{(j)}| + |M_{e}^{(j)}| + |M_{g}^{(j)}|} < \med^{(j)} \right)$\}: In this case,  $F^{(j+1)} = F^{(j)}  \cup M_{l}^{(j)}$, $U^{(j+1)} = U^{(j)} \cup M_{e}^{(j)} \cup M_{g}^{(j)}$ and $M^{(j+1)} = \emptyset$. We consider the two inequalities separately.

		\begin{enumerate}
			\setlength{\itemindent}{10pt}
			
			\item To show that \eqref{eq:while:lower} holds for $i=j+1$:
			\begin{align*}
			r - \sigma(F^{(j+1)})  
			& = r - \sigma(F^{(j)} \cup M_{l}^{(j)})  \\
			& \geq (\med^{(j)} - 1) \cdot (|U^{(j)}| + |M_{e}^{(j)}| + |M_{g}^{(j)}|)  \\
			& =  (\med^{(j)} - 1 ) \cdot |U^{(j+1)}|  \\
			& \geq \max(F^{(j+1)}) \cdot  |U^{(j+1)}|
			\end{align*}

			\item To show that \eqref{eq:while:upper} holds for $i=j+1$:
			\begin{align*}
			r - \sigma(F^{(j+1)} \cup M^{(j+1)}) 
			& = r - \sigma(F^{(j)} \cup M_{l}^{(j)}) \\
			& < \med^{(j)} \cdot (|U^{(j)}| +  |M_{e}^{(j)}| + |M_{g}^{(j)}|) \\
			& = \med^{(j)} \cdot |U^{(j+1)}| \\
			& = \min(U^{(j+1)}) \cdot |U^{(j+1)}|
			\end{align*}
			
		\end{enumerate}
		
	\end{enumerate}
	
	\noindent
	In all three cases, \eqref{eq:while:lower} and \eqref{eq:while:upper} hold for $i=j+1$. Thus, by induction, 
	\eqref{eq:while:lower} and \eqref{eq:while:upper} hold for each $i$ with $0 \leq i \leq I$.
	\qed
\end{proof}

At the end of the while loop (iteration $I$), $M^{(I)} = \emptyset$ and, by our assumption, $U^{(I)} \neq \emptyset$. The two inequalities can be combined to yield the following result.

\begin{corollary} The output of the labeling algorithm satisfies the following condition:
	\label{cor:label:condition}
	\begin{align}
	\label{eq:label:condition}
	\max(F^{(I)})  \; \leq \; \frac{r - \sigma(F^{(I)})}{|U^{(I)}|} \; < \; \min(U^{(I)}) 
	\end{align}
\end{corollary}

Using Corollary~\ref{cor:label:condition}, we can establish that any balanced placement must satisfy the following two properties:

\begin{enumerate}
	
	\item Every child in $F^{(I)}$ is filled (to capacity). Otherwise, the placement will have at least one unfilled child with at most $\max(F^{(I)}) - 1$ replicas and one child with at least $\lceil \frac{r - \sigma(F^{(I)})}{|U^{(I)}|} \rceil$ replicas. The latter is at least $\max(F^{(I)}) + 1$ implying that the placement is not balanced. 
	
	\item No child in $U^{(I)}$ is assigned more than $\lceil \frac{r - \sigma(F^{(I)})}{|U^{(I)}|} \rceil$ replicas.
	Otherwise, the placement will have at least one child with at least $\lceil \frac{r - \sigma(F^{(I)})}{|U^{(I)}|} \rceil + 1$ replicas and one unfilled child 
	with at most $\lceil \frac{r - \sigma(F^{(I)})}{|U^{(I)}|} \rceil - 1$ implying that the placement is not balanced. 
	
\end{enumerate}

Note that \eqref{eq:label:condition} implies that every child in $U^{(I)}$ has enough capacity to store the replicas assigned to it by our placement algorithm.  
This shows that our labeling algorithm correctly computes filled and unfilled children at a node.

\section{Pseudocode for Transform Phase}\label{app:transform-pseudo}

\begin{algorithm}[tb]
	\SetKwFunction{transf}{Transform}
	\SetKwFunction{mkPseudo}{Make-Pseudonode}
	\SetKwFunction{fill}{Filled}
	\SetKwProg{Fn}{Function}{begin}{end}
	
	\small
	
	\Fn{\transf{$u, chain, s$}}{
		Let $u$ have children labeled $1,...,t$\;
		\If(\tcp*[f]{\footnotesize not chain node}){$u$ has two or more \punfill children}{ \label{line-bottom}
			\ForEach{child $i$ \punfill}{
				\label{line-pass-on}$(-, -, -, x) \gets $\transf{$i, false, \bot$} \;
				\lIf(\tcp*[f]{\footnotesize replace $i$ with pseudonode}){$i \neq x$}{ $i \gets x$ \label{line-update}}
			}
			\lIf{$chain = false$} { \Return{$(\bot, \bot, \bot, u)$} \label{line-vt}} 
			\lElse(\tcp*[f]{\footnotesize last node of chain}){ \Return{$(\vec{0}_{s + 1}, \vec{0}_{s+1},  \vec{0}_{s+1}, u)$} \label{line-alloc} } 
		}
		\If(\tcp*[f]{\footnotesize chain node}){$u$ has one \punfill child, $v$}{ \label{line-test2}
			\uIf(\tcp*[f]{\footnotesize first node of chain}){$chain = false$} { 
				$(\vec{a},\vec{b},\vec{f},x) \gets$ \transf{$v, true, \minRepOn{v}$} \tcp*[r]{\footnotesize pass down $\minRepOn{v}$}\label{line-mark-ru}
			}\lElse{	$(\vec{a},\vec{b},\vec{f},x) \gets$ \transf{$v, true, s$} \label{line-mark-rv1}}
			\ForEach(\tcp*[f]{\footnotesize $O(n_i)$ time}){filled child $i$} { \label{line-filledend} \label{line-contribstart} \fill{$i, \vec{f}$} }
			$k \gets \sum_{i \text{ filled}} \Csub{i} + \minRepOn{v} - 1$ \label{line-ustart} \tcp*[r]{\footnotesize $k$ is min failure number of $u$}
			$\vec{a}[k+1] \gets \vec{a}[k+1] + 1$; \hspace{1em} 	$\vec{b}[k] \gets \vec{b}[k] + 1$ \tcp*[r]{\footnotesize update $\vec{a}$ and $\vec{b}$}\label{line-contribend}
			\If{$chain = false$}{ \label{line-chainback-to-start}	$x \gets $ \mkPseudo{$\vec{a}, \vec{b}, \vec{f}, x$} \label{line-pseudonodecreate}
			}
			\Return{$(\vec{a},\vec{b},\vec{f},x)$}\;\label{line-return}
		}
	}
	\BlankLine
	
	\small
	
	\Fn{\fill{$u$, $\vec{f}$}}{
		
		\If{$u$ is a leaf}{
			$\vec{f}[0] \gets \vec{f}[0]  + 1$\;
			\Return\;
		}{ 
			\ForEach{child $i$}{
				\fill{$i, \vec{f}$}
			}
			$a \gets \sum_i |\Csub{i}|$ \;
			$\vec{f}[a] \gets \vec{f}[a] + 1$\;
			\Return \;
		}
	}
	
	\BlankLine
	
	\small
	
	\Fn{\mkPseudo{$\vec{a}$, $\vec{b}$, $\vec{f}$, $x$}}{
		allocate a new node $node$\;
		$node.\vec{a} \gets \vec{a} + \vec{f}$\;
		$node.\vec{b} \gets \vec{b} + \vec{f}$\;
		$node.child \gets x$\;
		\Return $node$\;
	}
	
	\caption{Transform phase}\label{a:transform}
\end{algorithm}

Pseudocode for the transform phase is given in Algorithm \ref{a:transform}. For convenience, define $S_w \defined T_1 \cup \ldots \cup T_{k-1} \cup \{v_1, \ldots, v_{k-1}\}$, and let the contribution of nodes in $S_w$ to $\vec{a_w}$ and $\vec{b_w}$ be given by vectors $\vec{a}$ and $\vec{b}$ respectively. The transform phase is started at the root of the tree by invoking \transf{$root,false, \repfact$}. \transf is a modified recursive breadth-first search, which returns a $4$-tuple $(\vec{a}, \vec{b}, \vec{f}, x)$, where $x$ is the node $v_t$ which ends the degenerate chain. As the recursion proceeds down the tree, each node is tested to see if it is part of a degenerate chain (lines \ref{line-bottom} and \ref{line-test2}). If a node is not part of a degenerate chain, the call continues on all \punfill children (line \ref{line-pass-on}). The first node ($v_1$) in a degenerate chain is marked (by passing down $chain \gets true$ at lines \ref{line-mark-ru} and \ref{line-mark-rv1}). Once the bottom of the chain (node $v_k$) has been reached, the algorithm allocates memory for three vectors, $\vec{a}, \vec{b}$ and $\vec{f}$, each of size $s+1$ (line \ref{line-alloc}). The value of $\minRepOn{v_1}$ is passed down to the bottom of the chain at lines \ref{line-mark-ru} and \ref{line-mark-rv1}. These vectors are then passed up through the entire degenerate chain (cf. lines \ref{line-alloc} and \ref{line-return}), along with node $u$ (at line \ref{line-alloc}), whose use will be explained later. When a node $u$ in a degenerate chain receives $\vec{a}, \vec{b}$, and $\vec{f}$, $u$ adds its contribution to each vector (lines \ref{line-contribstart}-\ref{line-contribend}). The contribution of node $u$ consists of two parts. First, the contribution of the \dfill nodes is added to $\vec{f}$ by invoking a special \fill subroutine which computes the sum of the failure aggregates of each \dfill child of $u$ (line \ref{line-filledend}). Note that \fill uses pass-by-reference semantics when passing in the value of $\vec{f}$. Then, the contribution of node $u$ itself is added, by summing the number of leaves in all of the \dfill children, and the number of replicas on the single \punfill child, $v$ (lines \ref{line-ustart}-\ref{line-contribend}). By the time that the recursion reaches the start of the chain on the way back up, all nodes have added their contribution, and the pseudonode is created and returned (line \ref{line-pseudonodecreate}).

\end{document}